\pdfoutput=1
\documentclass[11pt,a4paper]{article}
\usepackage[margin=1.1in]{geometry} %
\usepackage[utf8]{inputenc}
\usepackage{mathtools} %
\usepackage{amssymb}

\usepackage{microtype}
\usepackage{enumerate}
\usepackage{accents} %
\usepackage{graphicx}
\usepackage{authblk}
\usepackage[table,dvipsnames]{xcolor}
\usepackage{tikz}
\usepackage{booktabs}
\usepackage{array}

\usepackage{float}
\usepackage[colorlinks, pagebackref=true]{hyperref}
\usepackage{xurl}
\hypersetup{
  linkcolor=[rgb]{0.3,0.3,0.6},
  citecolor=[rgb]{0.2, 0.6, 0.2},
  urlcolor=[rgb]{0.6, 0.2, 0.2}
}
\hypersetup{breaklinks=true}

\allowdisplaybreaks[1]

\usepackage{amsthm}
\usepackage{thmtools, thm-restate}
\usepackage[capitalize, nameinlink]{cleveref}

\declaretheorem[name=Theorem, parent=section]{theorem}
\declaretheorem[name=Corollary, sibling=theorem]{corollary}
\declaretheorem[name=Proposition, sibling=theorem]{proposition}
\declaretheorem[name=Lemma, sibling=theorem]{lemma}

\declaretheorem[name=Definition, sibling=theorem, style=definition]{definition}
\declaretheorem[name=Remark, sibling=theorem, style=definition]{remark}

\declaretheorem[name=Example, sibling=theorem, style=definition]{example}

\definecolor{reviewcolor}{HTML}{8E24AA}

\numberwithin{equation}{section}

\crefname{theorem}{Theorem}{Theorems}
\crefname{lemma}{Lemma}{Lemmas}
\crefname{proposition}{Proposition}{Propositions}
\crefname{corollary}{Corollary}{Corollaries}
\crefname{conjecture}{Conjecture}{Conjectures}
\crefname{observation}{Observation}{Observations}
\crefname{fact}{Fact}{Facts}
\crefname{definition}{Definition}{Definitions}
\crefname{figure}{Figure}{Figures}
\crefname{table}{Table}{Tables}
\crefname{equation}{Equation}{Equations}

\crefname{remark}{Remark}{Remarks}
\crefname{problem}{Problem}{Problems}
\crefname{example}{Example}{Examples}
\crefname{question}{Question}{Questions}
\crefname{claim}{Claim}{Claims}
\crefname{section}{Section}{Sections}
\crefname{subsection}{Section}{Sections}
\crefname{subsubsection}{Section}{Sections}

\DeclareMathAccent{\wtilde}{\mathord}{largesymbols}{"65}

\DeclareMathOperator{\tensorrank}{R}

\newcommand{\FF}{\mathbb{F}}

\newcommand{\RR}{\mathbb{R}}
\newcommand{\CC}{\mathbb{C}}

\newcommand{\ZZ}{\mathbb{Z}}

\newcommand{\Deltam}{\Delta_\FF^{\!\mathsf{m}}}
\newcommand{\eps}{\varepsilon}

\def\<#1>{\left\langle\ignorespaces#1\unskip\right\rangle}

\newcommand{\asympleq}{\lesssim}

\usepackage{scalerel}

\DeclareMathOperator{\SL}{SL}

\DeclareMathOperator{\linspan}{span}

\DeclareMathOperator{\Sym}{Sym}

\DeclareMathOperator{\conv}{conv}

\DeclarePairedDelimiterX\braket[2]{\langle}{\rangle}{#1\,\delimsize\vert\,\mathopen{}#2}

\DeclareMathOperator{\R}{R}
\newcommand{\flatrank}[1]{\R_{(#1)}}
\newcommand{\degen}{\mathrel{\trianglelefteq}}
\newcommand{\degengeq}{\mathrel{\trianglerighteq}}
\newcommand{\degenTT}{\mathrel{\trianglelefteq_1}}
\newcommand{\degengeqTT}{\mathrel{\trianglerighteq_1}}
\DeclareRobustCommand{\ruleref}[1]{\hyperref[rule:#1]{\textup{(#1)}}}
\DeclareRobustCommand{\fruleref}[1]{\hyperref[frule:#1]{\textup{(#1)}}}
\newcommand{\oplusX}{\mathbin{\oplus_1}}
\newcommand{\polmul}{\mathsf{P}}

\definecolor{psumA}{HTML}{2A78D6}\definecolor{psumB}{HTML}{EB6834}\definecolor{psumC}{HTML}{1BAF7A}
\newcommand\psumbox[7]{%
\filldraw[draw=#1!75!black,fill=#1!55] (#5,#3,#4)--(#5,#6,#4)--(#5,#6,#7)--(#5,#3,#7)--cycle;
\filldraw[draw=#1!75!black,fill=#1!22] (#2,#3,#4)--(#5,#3,#4)--(#5,#3,#7)--(#2,#3,#7)--cycle;
\filldraw[draw=#1!75!black,fill=#1!38] (#2,#3,#4)--(#5,#3,#4)--(#5,#6,#4)--(#2,#6,#4)--cycle;}
\newcommand\psumframeback[3]{%
\draw[gray!70,dashed] (0,0,#3)--(0,#2,#3)--(#1,#2,#3) (0,#2,#3)--(0,#2,0);}
\newcommand\psumframefront[3]{%
\draw[gray!70] (0,0,0)--(#1,0,0)--(#1,#2,0)--(0,#2,0)--cycle (0,0,0)--(0,0,#3)--(#1,0,#3)--(#1,0,0) (#1,0,#3)--(#1,#2,#3)--(#1,#2,0);}

\title{Multiplying matrices via polynomials:\\ exposition and barriers}

\author{Alimzhan~Amanov}
\author{Maxim~van~den~Berg}
\author{Jeroen~Zuiddam}
\affil{University of Amsterdam}

\begin{document}
\maketitle

\begin{abstract}
We give an exposition of OpenAI's recent proof that 
the exponent~$\omega_\CC$ of matrix multiplication over the complex numbers satisfies $\omega_\CC\leq9/4$. %
Our version of the proof gives $\omega_\FF\leq9/4$  over every field $\FF$. 
We also treat 
the rectangular matrix multiplication exponent $\omega_\FF(a,b,c)$. Here we have %
$\omega_\FF(1,k,1)=2$ for $0\leq k\leq1/2$ and $\omega_\FF(1,k,1)\leq1+k+1/(4k)$ for $k\geq1/2$, 
and, for real numbers $a\geq b\geq c\geq0$, $\omega_\FF(a,b,c)=a+b$ if $c(a+b)\leq ab$ and $\omega_\FF(a,b,c)\leq(a+b+c)/2+(ab/c+ac/b+bc/a)/4$ otherwise.
Our exposition is in the language of degeneration inequalities among polynomial multiplication tensors and Strassen's asymptotic spectrum duality.
We moreover prove an operational barrier (which can be thought of as coming from an operational asymptotic spectrum) showing that the current analysis of all these upper bounds on $\omega_\FF$ and $\omega_\FF(a,b,c)$ is optimal and that in fact the current degeneration ingredients cannot improve the bounds.

\end{abstract}

\tableofcontents

\section{Introduction}\label{sec:intro}

The exponent~$\omega_\FF$ of matrix multiplication over a field~$\FF$ is the infimum of all real numbers~$\tau$ such that the tensor rank of the matrix multiplication tensor~$\langle n,n,n\rangle$ over~$\FF$ is $\mathcal{O}(n^\tau)$. %
OpenAI recently released a proof that $\omega_\CC\leq9/4$~\cite[Theorem~1.1]{OAI:Matrix-Multiplication-Nine-Fourths-October-2-2026}, significantly improving over the previous upper bounds and using new ideas. %
We give an exposition of their proof in the language of degeneration inequalities and Strassen's asymptotic spectrum of tensors~\cite[Section~3]{strassen1988asymptotic} (see also the survey of Wigderson and Zuiddam~\cite{wigderson2022asymptotic}). Our version of the proof works over every field.

\begin{theorem}\label{thm:main}
For every field~$\FF$, $\omega_\FF\leq9/4$.
\end{theorem}

The previous upper bound was $\omega_\FF<2.371339$~\cite{almanMoreAsymmetryYields2025}, from the laser method.

We also discuss rectangular matrix multiplication exponents. For integers $n_1,n_2,n_3\geq1$, the matrix multiplication tensor~$\langle n_1,n_2,n_3\rangle$ describes the multiplication of an $n_1\times n_2$ matrix by an $n_2\times n_3$ matrix.
For real numbers $a,b,c\geq0$, let $\omega_\FF(a,b,c)$ be the infimum of the real numbers~$\gamma$ such that the tensor rank of $\langle\lceil n^a\rceil,\lceil n^b\rceil,\lceil n^c\rceil\rangle$ over~$\FF$ is $\mathcal{O}(n^\gamma)$. In particular, $\omega_\FF=\omega_\FF(1,1,1)$.
Following Uzna\'nski~\cite{uznanski2026rectangularmatrixmultiplicationpolynomial}, the same methods give the following bounds.

\begin{restatable}{theorem}{thmrect}\label{thm:rect}
For every field~$\FF$ and all real numbers $a\geq b\geq c\geq0$,
\[
\omega_\FF(a,b,c)\leq\begin{cases}a+b&\text{if }c(a+b)\leq ab,\\[2pt]\dfrac{(ab+bc+ca)^2}{4abc}&\text{if }c(a+b)>ab,\end{cases}
\]
with equality if $c(a+b)\leq ab$.
\end{restatable}

In particular, $\omega_\FF(1,k,1)=2$ for $0\leq k\leq1/2$, so $\alpha_\FF\coloneqq\sup\{k\geq0:\omega_\FF(1,k,1)=2\}$ is at least~$1/2$ (\cref{cor:rect-k}). This implies for instance that $n^2 \times n$ by $n \times n^2$ matrices can be multiplied in $\mathcal{O}(n^{4 + \eps})$ operations! 
This improves previous $\alpha_\CC>0.465$~\cite{OAI:Complex-Matrix-Multiplication-Below-2.258-and-Rectangular-Bounds-September-24-2026} and $\alpha_\FF>0.321334$ from the laser method~\cite{williams2023new}. \Cref{fig:rect-general} shows \cref{thm:rect}. %

We will prove that the bound of \cref{thm:rect} is ``the best that its proof can give'': the ingredients that our proof use do not give a better upper bound on~$\omega_\FF(a,b,c)$, and neither does any combination (in a precise way) of the ingredients (\cref{thm:barrier}). To improve the bounds further requires going around this ``operational barrier''.
\begin{figure}[t]
\centering
\includegraphics[width=\linewidth]{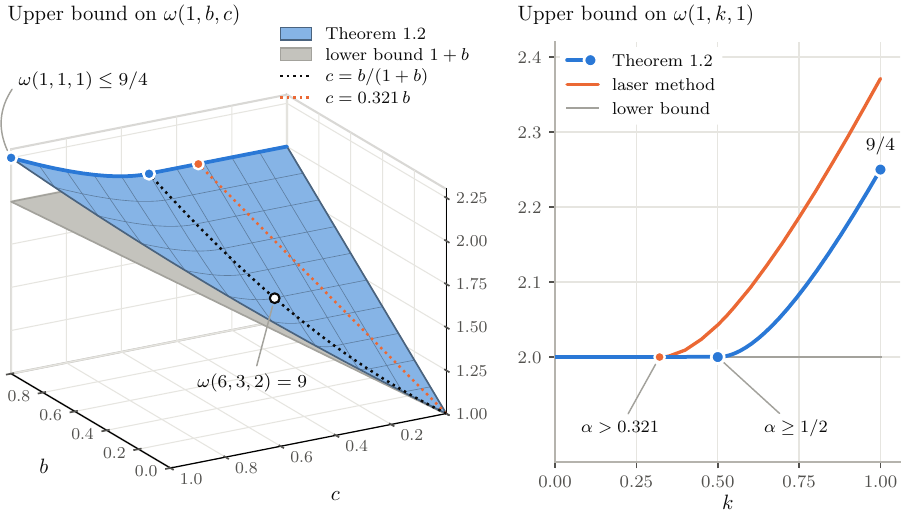}
\caption{(Left) The upper bound of \cref{thm:rect} on $\omega_\FF(1,b,c)$ for $1\geq b\geq c\geq0$ %
as a blue surface, %
above the grey plane of the lower bound~$1+b$. The black dotted curve $c=b/(1+b)$ indicates until where upper bound and lower bound match and give $\omega_\FF(1,b,c)=1+b$. The orange dotted line $c=0.321334\,b$ shows until where this match was known by previous results from the laser method~\cite{williams2023new}. %
(Right) The upper bound of \cref{cor:rect-k} on $\omega_\FF(1,k,1)$ for $0\leq k\leq1$, compared with the previously known upper bound from the laser method, as the lower convex hull of the values tabulated in~\cite{williams2023new,almanMoreAsymmetryYields2025}.} %
\label{fig:rect-general}
\end{figure}

In this introduction we will give a brief overview of the proof ideas and results, focussing on $\omega_\FF = \omega_\FF(1,1,1)$, but the ideas work for the rectangular exponents $\omega_\FF(a,b,c)$ as well. Formal definitions and statements are from \cref{sec:tensors} onwards. %
Standard references for background are the survey of Bläser~\cite{blaser2013fast} and the book of Bürgisser, Clausen and Shokrollahi~\cite{burgisser1997algebraic}.

\subsection{Degeneration ingredients and combination rules}\label{sec:ingredients}
The matrix multiplication tensor is $\langle n_1,n_2,n_3\rangle=\sum_{i=1}^{n_1}\sum_{j=1}^{n_2}\sum_{k=1}^{n_3}E_{ij}\otimes E_{jk}\otimes E_{ki}$, where~$E_{ij}$ is the matrix with entry~$1$ at~$(i,j)$ and~$0$ elsewhere.
The matrix multiplication exponent $\omega_\FF$ is defined in terms of tensor rank (or equivalently in terms of arithmetic circuit complexity). For a tensor $T$ the tensor rank $\tensorrank(T)$ is the smallest number $r$ such that $T \leq \langle r\rangle$ where $\langle r\rangle=\sum_{i=1}^re_i\otimes e_i\otimes e_i$ is the unit tensor and where $\leq$ denotes restriction: the left hand side is obtained from the right hand side by applying linear maps to the three legs. It goes back to Strassen~\cite{strassen1969gaussian} that if $\langle n,n,n\rangle\leq\langle r\rangle$ for integers $n\geq2$ and $r\geq1$, 
then $\omega_\FF\leq\log r/\log n$. Throughout, $\log$ denotes the logarithm to base~$2$.

It is known that~$\omega_\FF$ does not change if we replace restriction by its approximative version, degeneration (i.e.,~we replace tensor rank by border rank): if $\langle n,n,n\rangle\degen\langle r\rangle$ for integers $n\geq2$ and $r\geq1$, 
then still $\omega_\FF\leq\log r/\log n$. Degeneration $\degen$ is a restriction in which we allow to use matrices whose entries are polynomials in a formal variable~$\eps$, which we think of as going to zero (also in positive characteristic one can define this correctly, symbolically). %
So to prove upper bounds on~$\omega_\FF$ it suffices to prove degenerations of the form 
\[
\langle n,n,n\rangle\degen\langle r\rangle,
\]
for which $\log r/ \log n$ is small;
We keep that in mind as our goal. 

We may think of $\leq$ and $\degen$ as notions of reduction between computational problems, in this case using $r$ pairwise scalar multiplications to compute $n\times n$ matrix multiplication, and we will see many more. Restriction $S \leq T$ implies degeneration $S \degen T$, but in our analysis (following Strassen's way of working) we will go even further and use the asymptotic restriction $S \asympleq T$ defined by $S^{\otimes n} \leq T^{\otimes n} \otimes  \langle2^{o(n)}\rangle$, for which $S \degen T$ implies $S \asympleq T$. (Indeed, also $\langle n,n,n\rangle \asympleq \langle r\rangle$ implies $\omega_\FF\leq \log r / \log n$.)

We explain the proof at a high level in terms of degenerations. We first discuss the ingredients. The main objects, besides the matrix multiplication tensors, are the polynomial multiplication tensors~$\polmul_{a,b}$, which describe the multiplication of a polynomial of degree at most~$a-1$ by a polynomial of degree at most~$b-1$. %
In other words, with~$e_i$ standing for the monomial~$t^{i-1}$,
\[
\polmul_{a,b}=\sum_{i=1}^{\smash{a}}\sum_{j=1}^{\smash{b}}e_i\otimes e_j\otimes e_{i+j-1}\in\FF^a\otimes\FF^b\otimes\FF^{a+b-1}.
\]
\subsubsection{Degeneration ingredients}
The starting point is that the simplest kind of polynomial multiplication (multiplying a degree~0 polynomial with a degree $b-1$ polynomial, i.e.~$\polmul_{1,b}$) is the same thing as the simplest kind of rectangular matrix multiplication (multiplying a constant with a length $b$ vector, i.e.~$\langle1,1,b\rangle$):
\begin{enumerate}[({D}1)]\setcounter{enumi}{-1}
\item\label{rule:D0} \textbf{Matrix multiplication.} $\polmul_{1,b}\cong\langle1,1,b\rangle$ (\cref{def:polmul}).
\end{enumerate}

By ingredient~\ruleref{D0}, and since $\langle k,\ell,m\rangle\cong\langle k,1,1\rangle\otimes\langle1,\ell,1\rangle\otimes\langle1,1,m\rangle$, we see that the semiring, with direct sum $\oplus$ as addition and tensor product $\otimes$ as multiplication, generated by the polynomial multiplication tensors~$\polmul_{a,b}$ and their permutations (where we mean permuting the three tensor legs), contains all matrix multiplication tensors $\langle n_1,n_2,n_3\rangle$, and their direct sums. It will thus be natural to study the matrix multiplication tensors through (permuted) polynomial multiplication tensors.

Next, what makes polynomial multiplication appealing, is the fundamental fact that polynomial multiplication is computationally very cheap (by interpolation!), which in our language corresponds to a border rank upper bound:
\begin{enumerate}[({D}1)]\setcounter{enumi}{0}
\item\label{rule:D1} \textbf{Border rank.} $\polmul_{a,b}\degen\langle a+b-1\rangle$ (\cref{prop:polmul-rank}).
\end{enumerate}

In the study of matrix multiplication, as we already hinted at above, tensors are commonly combined by the direct sum~$\oplus$ and the tensor product~$\otimes$. Tensor products are natural because $\langle n_1,n_2,n_3\rangle\otimes\langle m_1,m_2,m_3\rangle\cong\langle n_1m_1,n_2m_2,n_3m_3\rangle$ (``block matrices are multiplied blockwise''). Direct sums are natural because of Sch\"onhage's asymptotic sum inequality: if $\langle n_1,n_1,n_1\rangle\oplus\cdots\oplus\langle n_p,n_p,n_p\rangle\degen\langle r\rangle$ for integers $p,n_1,\ldots,n_p,r\geq1$, then we can obtain from this degenerations of the form $\langle m,m,m\rangle \degen \langle s\rangle$ that imply $n_1^{\omega_\FF}+\cdots+n_p^{\omega_\FF}\leq r$~\cite[Theorem~7.1]{MR623057}. So a direct sum of matrix multiplication tensors that is a degeneration of a unit tensor also bounds~$\omega_\FF$, like the degenerations $\langle n,n,n\rangle\degen\langle r\rangle$ above.

Now if we look at $T \otimes \langle 1,1,2\rangle$ we see that this gives a notion of partial direct sum of $T$ with itself, which we denote by $T \oplus_1 T$: two copies of $T$ are placed next to each other with one shared leg and two disjoint legs. More generally $\oplusX$ denotes a partial direct sum, whose summands share their first leg (\cref{def:partial-direct-sum}).
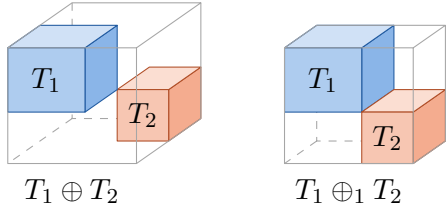
\begin{figure}[H]
\centering
\begin{minipage}[c]{0.42\linewidth}
\centering
\begin{tikzpicture}[x={(1cm,0cm)},y={(0cm,-1cm)},z={(0.5cm,0.35cm)},line join=round,thin,scale=0.85]
\begin{scope}
\psumframeback{2}{1.8}{2}
\psumbox{psumB}{1.2}{1}{1}{2}{1.8}{2}
\psumbox{psumA}{0}{0}{0}{1.2}{1}{1}
\psumframefront{2}{1.8}{2}
\node at (0.6,0.5,0) {$T_1$};
\node at (1.6,1.4,1) {$T_2$};
\node[below=2pt] at (1,1.8,0) {$T_1\oplus T_2$};
\end{scope}
\begin{scope}[xshift=4.3cm]
\psumframeback{2}{1.8}{1}
\psumbox{psumA}{0}{0}{0}{1.2}{1}{1}
\psumbox{psumB}{1.2}{1}{0}{2}{1.8}{1}
\psumframefront{2}{1.8}{1}
\node at (0.6,0.5,0) {$T_1$};
\node at (1.6,1.4,0) {$T_2$};
\node[below=2pt] at (1,1.8,0) {$T_1\oplusX T_2$};
\end{scope}
\end{tikzpicture}
\end{minipage}\hfill
\begin{minipage}[c]{0.54\linewidth}
\setlength{\abovecaptionskip}{0pt}
\caption{The direct sum $T_1\oplus T_2$ (left) and the partial direct sum $T_1\oplusX T_2$ (right), drawn as three-dimensional arrays.
All coefficients outside the two boxes are zero.}
\label{fig:partial-direct-sum}
\end{minipage}
\end{figure}

We will be working a lot with these partial direct sums $\oplusX$, but in the end we will need a way to transform these back into the usual world of direct sums $\oplus$:
\begin{enumerate}[({D}1)]\setcounter{enumi}{1}
\item\label{rule:D2} \textbf{Separating partial direct sums into direct sums.} For a partial direct sum $T_1\oplusX\cdots\oplusX T_m$ and every integer $t\geq1$,
\[
(T_1\oplus\cdots\oplus T_m)\otimes\polmul_{1,t}\degen(T_1\oplusX\cdots\oplusX T_m)\otimes\polmul_{m,t}
\]
(\cref{prop:sep}).
\end{enumerate}
Note how ingredient~\ruleref{D2} turns a partial direct sum into a direct sum, at the cost of turning a factor~$\polmul_{m,t}$ into~$\polmul_{1,t}$, thus using $\polmul_{m,t}$ as a ``partially consumed catalyst'' for this degeneration.  %

Finally, we need two families of degenerations that relate partial direct sums of polynomial multiplication tensors. The first family is:
\begin{enumerate}[({D}1)]\setcounter{enumi}{2}
\item\label{rule:D3} \textbf{Concavity degeneration.} For $b\geq2$, $\polmul_{a,b+1}\oplusX\polmul_{a,b-1}\degenTT\polmul_{a,b}\oplusX\polmul_{a,b}$, and the same with arbitrary $\oplus_1 C$ added on both sides (\cref{lem:det-degen}).
\end{enumerate}
(Here the subscript in~$\degenTT$ indicates that these degenerations do not act on the first leg, which indeed allows adding arbitrary $\oplus_1 C$ on both sides.)
In terms of polynomials, ingredient \ruleref{D3} is using two multiplications of the same polynomial~$f$ by polynomials of degree at most~$b-1$ to compute multiplications of~$f$ by polynomials of degree at most~$b$ and~$b-2$. The construction of \ruleref{D3} is related to the representation theory of~$\SL_2$.

The second family is:
\begin{enumerate}[({D}1)]\setcounter{enumi}{3}
\item\label{rule:D4} \textbf{Tripling degeneration.} $\polmul_{a,b}\oplusX((23)\cdot\polmul_{a,b})\oplusX\polmul_{a,b}\degenTT\polmul_{a,3b+a-1}$, where $(23)\cdot\polmul_{a,b}$ is~$\polmul_{a,b}$ with its second and third legs exchanged; and again the same with arbitrary $\oplus_1 C$ added on both sides (\cref{lem:tripling-degen}).
\end{enumerate}
In terms of polynomials, ingredient \ruleref{D4} is using one multiplication of~$f$ by a polynomial~$g$ of degree at most~$3b+a-2$ to compute three smaller ones. We can use the lowest~$b$ and the highest~$b$ coefficients of~$g$ to do two multiplications of~$f$ by polynomials of degree at most~$b-1$, the two copies of~$\polmul_{a,b}$. Using degeneration we are able to squeeze in an extra, permuted polynomial multiplication, $(23)\cdot\polmul_{a,b}$. %

\subsubsection{Combination rules}
The proof will combine the above degeneration ingredients \ruleref{D0}--\ruleref{D4} by basic combination rules to obtain more degenerations, as to eventually obtain degenerations of the form $\langle n,n,n\rangle \degen \langle r\rangle$, hopefully such that $\log r / \log n$ is small. These rules are: %
\begin{enumerate}[({C}1)]
\item\label{rule:C1} reflexivity: $T\degen T$;
\item\label{rule:C2} transitivity: if $T_1\degen T_2$ and $T_2\degen T_3$, then $T_1\degen T_3$;
\item\label{rule:C3} direct sums and tensor products: if $T_1\degen S_1$ and $T_2\degen S_2$, then $T_1\oplus T_2\degen S_1\oplus S_2$ and $T_1\otimes T_2\degen S_1\otimes S_2$;
\item\label{rule:C4} permutations of the legs: if $T\degen S$, then $\pi\cdot T\degen\pi\cdot S$ for every permutation~$\pi \in \mathfrak{S}_3$ of the three legs.
\end{enumerate}
Finding the optimal combination of the degeneration ingredients given these combination rules seems a daunting task! Not only do we have an infinite family of degeneration ingredients (parametrised by the several parameters $a,b,t,m$), we may also need to use large $n$ to obtain our target degenerations $\langle n,n,n\rangle \degen \langle r\rangle$ for which $\log r / \log n$ is small. %

Let us first illustrate what sort of combination strategies we may want to do, and then let Strassen's asymptotic spectrum duality come to the rescue to turn this into a tractable optimisation problem (with a clean analysis, optimal objective value and ``dual certificate''). 

\subsection{Understanding how the ingredients can be combined}\label{sec:intro-outline}
Let us try to combine the degeneration ingredients according to the combination rules to start understanding them, including the role of catalysts and asymptotics.

\subsubsection{From \texorpdfstring{$\polmul_{a,b}$}{P(a,b)} to partial direct sum of \texorpdfstring{$\langle 1,1,n\rangle$}{<1,1,n>} and permutations}
Fix an integer $a\geq2$. By ingredient~\ruleref{D1}, we have $\langle2a-1\rangle\degengeq\polmul_{a,a}$. Tensoring both sides with~$\langle1,1,n\rangle$ gives
\begin{equation}\label{eq:outline-start}
\langle2a-1\rangle\otimes\langle1,1,n\rangle\degengeq\polmul_{a,a}\otimes\langle1,1,n\rangle\cong\polmul_{a,a}\oplusX\cdots\oplusX\polmul_{a,a},
\end{equation}
the partial direct sum of $n$~copies of~$\polmul_{a,a}$.

Ingredients~\ruleref{D3} and~\ruleref{D4} are degenerations that act only on the second and third legs, so they can be applied to some of the summands of a partial direct sum while the other summands stay as they are. With ingredient~\ruleref{D3}, the second parameter branches out: two summands~$\polmul_{a,b}$ become $\polmul_{a,b+1}$ and~$\polmul_{a,b-1}$. For example, ingredient~\ruleref{D3} with $b=a$, applied to two of the summands, continues~\eqref{eq:outline-start}:
\[
\polmul_{a,a}\oplusX\polmul_{a,a}\oplusX\polmul_{a,a}\oplusX\cdots\oplusX\polmul_{a,a}\degengeq\polmul_{a,a+1}\oplusX\polmul_{a,a-1}\oplusX\polmul_{a,a}\oplusX\cdots\oplusX\polmul_{a,a}.
\]
With ingredient~\ruleref{D4}, a summand~$\polmul_{a,3b+a-1}$ with a large second parameter $3b+a-1$ becomes three summands with second parameter~$b$, one of them with its second and third legs exchanged. In this way one can reach summands with second parameter~$1$, which are matrix multiplication tensors, since $\polmul_{a,1}\cong\langle a,1,1\rangle$ and $(23)\cdot\polmul_{a,1}\cong\langle1,a,1\rangle$. For example, a summand~$\polmul_{a,a+2}$, which three applications of ingredient~\ruleref{D3} produce from four copies of~$\polmul_{a,a}$, becomes three matrix multiplication tensors by ingredient~\ruleref{D4} with $b=1$:
\[
\polmul_{a,a+2}\degengeq\polmul_{a,1}\oplusX((23)\cdot\polmul_{a,1})\oplusX\polmul_{a,1}\cong\langle a,1,1\rangle\oplusX\langle1,a,1\rangle\oplusX\langle a,1,1\rangle.
\]

All in all, we see that after \eqref{eq:outline-start}, applying \ruleref{D3} and \ruleref{D4} will let us produce certain partial direct sums of matrix multiplication tensors.

\subsubsection{From partial direct sums to \texorpdfstring{$\langle n,n,n\rangle$}{<n,n,n>} via symmetrisation and powering}
Suppose that in the above manner the $n$~copies of~$\polmul_{a,a}$ of \eqref{eq:outline-start} become a partial direct sum that contains $s$~copies of~$\langle a,1,1\rangle$ and $t$~copies of~$\langle1,a,1\rangle$. (Note that here we have not yet specified explicitly what $s$ and $t$ can be; applying the above ideas cleverly will have to optimise this; We will discuss this via an example in a moment.) 
The partial direct sum of the $s$~copies of~$\langle a,1,1\rangle$ is $\langle a,1,1\rangle\otimes\langle1,1,s\rangle\cong\langle a,1,s\rangle$, and likewise the partial direct sum of the $t$~copies of~$\langle1,a,1\rangle$ is~$\langle1,a,t\rangle$. Together with~\eqref{eq:outline-start}, this gives
\begin{equation}\label{eq:outline-mm}
\langle2a-1\rangle\otimes\langle1,1,n\rangle\degengeq\langle a,1,s\rangle\oplusX\langle1,a,t\rangle.
\end{equation}

To make the matrix multiplication tensors appearing in \eqref{eq:outline-mm} ``square'' (i.e., of the form~$\langle p,p,p\rangle$), we use the symmetrisation operation $\bigotimes_{\pi\in\mathfrak{S}_3}\pi\cdot T$ for any  tensor~$T$, the tensor product of the six tensors obtained from~$T$ by permuting its legs. For all integers $k,\ell,m\geq1$,
\[
\bigotimes_{\pi\in\mathfrak{S}_3}\pi\cdot\langle k,\ell,m\rangle\cong\langle k\ell m,k\ell m,k\ell m\rangle^{\otimes2},
\]
which depends only on the product~$k\ell m$. In particular, the symmetrisation of~$\langle1,1,n\rangle$ is~$\langle n,n,n\rangle^{\otimes2}$, and that of~$\langle2a-1\rangle$ is~$\langle2a-1\rangle^{\otimes6}$. So symmetrising both sides of~\eqref{eq:outline-mm} gives
\begin{equation}\label{eq:outline-sym}
\langle2a-1\rangle^{\otimes6}\otimes\langle n,n,n\rangle^{\otimes2}\degengeq\bigotimes_{\pi\in\mathfrak{S}_3}\pi\cdot\bigl(\langle a,1,s\rangle\oplusX\langle1,a,t\rangle\bigr).
\end{equation}

From the right-hand side of \eqref{eq:outline-sym} we will construct an asymptotic restriction to one square matrix multiplication tensor (in a manner reminiscent of the laser method, but simpler):
\begin{equation}\label{eq:outline-add}
\bigotimes_{\pi\in\mathfrak{S}_3}\pi\cdot\bigl(\langle a,1,s\rangle\oplusX\langle1,a,t\rangle\bigr)\gtrsim\langle a(s+t),a(s+t),a(s+t)\rangle^{\otimes2}.
\end{equation}
The idea of the proof of~\eqref{eq:outline-add} is to take the $N$th tensor power of $\langle a,1,s\rangle\oplusX\langle1,a,t\rangle$, for large~$N$, and to keep only its summands with a fixed number~$k$ of factors~$\langle a,1,s\rangle$. These are $h\coloneqq\binom Nk$ tensors that share the first leg, each isomorphic to $M\coloneqq\langle a^k,a^{N-k},s^kt^{N-k}\rangle$. By ingredients~\ruleref{D2} and~\ruleref{D1}, with $\polmul_{1,h}\cong\langle1,1,h\rangle$ and $\polmul_{h,h}\degen\langle2h-1\rangle$, their partial direct sum~$Y$ satisfies
\[
\langle h\rangle\otimes\bigl(M\otimes\langle1,1,h\rangle\bigr)\degen\langle2h-1\rangle\otimes Y.
\]
The left-hand side consists of $h$~copies of $M\otimes\langle1,1,h\rangle$, and the right-hand side of fewer than $2h$~copies of~$Y$. The product of the three parameters of $M\otimes\langle1,1,h\rangle$ is $a^N\binom Nks^kt^{N-k}$. For the~$k$ at which $\binom Nks^kt^{N-k}$ is largest, this product is at least $(a(s+t))^N/(N+1)$, since the $N+1$ numbers $\binom Nks^kt^{N-k}$ add up to $(s+t)^N$. The factors~$2$ and~$N+1$ do not matter as $N\to\infty$, and the symmetrisation makes the result square.

Composing the degeneration~\eqref{eq:outline-sym} and the asymptotic restriction~\eqref{eq:outline-add} %
gives
\begin{equation}\label{eq:outline-catalyst}
\langle2a-1\rangle^{\otimes6}\otimes\langle n,n,n\rangle^{\otimes2}\gtrsim\langle a(s+t),a(s+t),a(s+t)\rangle^{\otimes2}.
\end{equation}

Suppose (essentially without loss of generality, after taking tensor powers and rounding down) that $a(s+t)=cn$ for an integer~$c\geq2$. Then $\langle cn,cn,cn\rangle\cong\langle c,c,c\rangle\otimes\langle n,n,n\rangle$, and~\eqref{eq:outline-catalyst} reads
\[
\langle2a-1\rangle^{\otimes6}\otimes\langle n,n,n\rangle^{\otimes2}\gtrsim\langle c,c,c\rangle^{\otimes2}\otimes\langle n,n,n\rangle^{\otimes2}.
\]
Then $\langle n,n,n\rangle^{\otimes2}$ acts as a catalyst: it is needed as input, but also returned as an output. A common nonzero tensor factor can be cancelled in an asymptotic restriction~\cite[Lemmas~3.5 and~3.7]{wigderson2022asymptotic}, which for us gives \[
\langle2a-1\rangle^{\otimes6}\gtrsim\langle c,c,c\rangle^{\otimes2},
\]
so $\omega_\FF\leq3\log(2a-1)/\log c$.

\subsubsection{Concrete example of applying ingredients (D3) and (D4)}\label{toy-example}
As a concrete example (already leading to a non-trivial upper bound on $\omega_\FF$), we now carry out the above steps for $a=2$. Ingredients~\ruleref{D3} and~\ruleref{D4} turn four copies of~$\polmul_{2,2}$ into matrix multiplication tensors and one copy of~$\polmul_{2,2}$:
\begin{align*}
\polmul_{2,2}\otimes\langle1,1,4\rangle&\cong\polmul_{2,2}\oplusX\polmul_{2,2}\oplusX\polmul_{2,2}\oplusX\polmul_{2,2}\\
&\degengeq\polmul_{2,1}\oplusX\polmul_{2,3}\oplusX\polmul_{2,3}\oplusX\polmul_{2,1}&&\text{\ruleref{D3}, $b=2$, twice}\\
&\degengeq\polmul_{2,1}\oplusX\polmul_{2,2}\oplusX\polmul_{2,4}\oplusX\polmul_{2,1}&&\text{\ruleref{D3}, $b=3$}\\
&\degengeq\polmul_{2,1}\oplusX\polmul_{2,2}\oplusX\polmul_{2,1}\oplusX((23)\cdot\polmul_{2,1})\oplusX\polmul_{2,1}\oplusX\polmul_{2,1}&&\text{\ruleref{D4}, $b=1$}\\
&\cong\langle2,1,4\rangle\oplusX\langle1,2,1\rangle\oplusX\polmul_{2,2},
\end{align*}
where in the last step, we use that the partial direct sum of the four copies of $\polmul_{2,1}\cong\langle2,1,1\rangle$ is $\langle2,1,1\rangle\otimes\langle1,1,4\rangle\cong\langle2,1,4\rangle$, and that $(23)\cdot\polmul_{2,1}\cong\langle1,2,1\rangle$. The copy of~$\polmul_{2,2}$ comes back and goes into the next round. So $k$~rounds, followed by removing the last copy of~$\polmul_{2,2}$, give
\[
\polmul_{2,2}\otimes\langle1,1,3k+1\rangle\degengeq\langle2,1,4k\rangle\oplusX\langle1,2,k\rangle.
\]
With~\eqref{eq:outline-start}, this is~\eqref{eq:outline-mm} with $n=3k+1$, $s=4k$ and $t=k$. Since $a(s+t)=10k>3k+1$, we get $\omega_\FF\leq3\log(3)/\log(10k/(3k+1))$, and letting $k\to\infty$ gives $\omega_\FF\leq3\log(3)/\log(10/3)<2.74$. 

Starting instead from $\langle5\rangle\degengeq\polmul_{2,4}$, ingredient~\ruleref{D4} with $b=1$ gives $\polmul_{2,4}\degengeq\langle2,1,2\rangle\oplusX\langle1,2,1\rangle$ in one step, and the same argument, with~$\langle5\rangle$ in place of~$\langle3\rangle$, gives the better bound $\omega_\FF\leq3\log5/\log6<2.695$.

Again, this is just an example, but it illustrates components of the actual construction, and it can indeed be extended (by starting with $\polmul_{a,a}$ with $a \to \infty$) %
to explicit combinations of the ingredients whose bounds tend to~$9/4$.

\subsection{High-level proof description and role of asymptotic  spectrum duality}\label{sec:intro-spectrum}
Our proof of \cref{thm:main} will not combine the ingredients by hand, as the \cref{toy-example} example does. Instead, Strassen's spectral points turn the search for the best combination of the ingredients, including their tensor powers, into an optimisation problem. Spectral points are semiring homomorphisms~$\phi$ from tensors to the nonnegative real numbers that are monotone under restriction (\cref{def:spectral-point}). By Strassen's duality theorem, $T\asympleq S$ if and only if $\phi(T)\leq\phi(S)$ for every spectral point~$\phi$~\cite[Theorem~3.2]{strassen1988asymptotic}. In particular, $\omega_\FF$ is the maximum of $\log\phi(\langle2,2,2\rangle)$ over all spectral points~$\phi$ (\cref{thm:spectral}).

For a spectral point~$\phi$, let $p_1$, $p_2$ and~$p_3$ be the logarithms of~$\phi(\langle1,1,2\rangle)$, $\phi(\langle2,1,1\rangle)$ and~$\phi(\langle1,2,1\rangle)$. Then $\log\phi(\langle2,2,2\rangle)=p_1+p_2+p_3$, so it suffices to show that $p_1+p_2+p_3\leq9/4$. We will show using ingredient~\ruleref{D2} that if $p_1>0$, then $\phi^{1/p_1}$ is superadditive on partial direct sums, whose summands share the first leg (\cref{lem:superadd-general}). Subsequently, \ruleref{D3} makes the function $b\mapsto\phi(\polmul_{a,b})^{1/p_1}$ concave, and \ruleref{D4} gives the ``tripling inequality''
\[
2\phi(\polmul_{a,b})^{1/p_1}+\phi\bigl((23)\cdot\polmul_{a,b}\bigr)^{1/p_1}\leq\phi(\polmul_{a,3b+a-1})^{1/p_1}
\]
(\cref{lem:profiles}). The same holds for the tensors $\pi\cdot\polmul_{a,b}$ with permuted legs in place of $\polmul_{a,b}$, with~$p_{\pi(1)}$ in place of~$p_1$. From these inequalities, an elementary argument shows that the product of the six numbers $\phi(\pi\cdot\polmul_{a,a})$, which is the value of~$\phi$ at the symmetrisation $\bigotimes_{\pi\in\mathfrak{S}_3}\pi\cdot\polmul_{a,a}$ of \cref{sec:intro-outline}, is at least $a^{8(p_1+p_2+p_3)/3}$ (\cref{prop:angular-square}). By ingredient~\ruleref{D1}, it is at most $(2a-1)^6$. Letting $a\to\infty$ gives $p_1+p_2+p_3\leq9/4$ for every spectral point~$\phi$, and hence $\omega_\FF\leq9/4$.

In terms of asymptotic restriction, the above construction leads to a chain in which symmetrised polynomial multiplication is in the middle: %
For every integer $a\geq1$,
\begin{equation}\label{eq:intro-distil}
\langle a^4,a^4,a^4\rangle^{\otimes2}\asympleq\Bigl(\bigotimes_{\pi\in\mathfrak{S}_3}\pi\cdot\polmul_{a,a}\Bigr)^{\otimes3}\degen\langle(2a-1)^{18}\rangle,
\end{equation}
where the asymptotic restriction follows from the inequalities above for all spectral points, by Strassen's duality theorem, and the degeneration from ingredient~\ruleref{D1}, with the legs permuted and tensor products taken. %

\subsection{Barriers for obtaining better bounds}\label{sec:intro-optimality}
We have seen how some combination of the degeneration ingredients gives $\omega_\FF \leq 9/4$ (and in a similar way bounds on the rectangular matrix multiplication exponents). Can we do better? Not with current ingredients.
\subsubsection{No barrier from irreversibility}
First of all, standard barriers are not effective here. Christandl, Vrana and Zuiddam~\cite[Theorem~3.1]{v017a002} (see also Alman~\cite{alman2018limits} and Christandl, Le Gall, Lysikov and Zuiddam \cite{MR4870821}) proved a barrier for upper bounds on~$\omega_\FF$ from a fixed tensor~$S$ (generalising previous barriers of Ambainis, Filmus and Le Gall~\cite{MR3388238}): every upper bound obtained from asymptotic restrictions $\langle2^n,2^n,2^n\rangle\asympleq S^{\otimes m}$ is at least $2\operatorname{irr}(S)$, where the irreversibility $\operatorname{irr}(S)$ is the ratio of the logarithms of the asymptotic rank and the asymptotic subrank of~$S$. 

This irreversibility barrier does not explain why the proof of \cref{thm:main} gives no bound below~$9/4$: for the tensors $S_a=\bigotimes_{\pi\in\mathfrak{S}_3}\pi\cdot\polmul_{a,a}$ of~\eqref{eq:intro-distil}, with $a\geq2$, the degenerations $\langle a\rangle\degen\polmul_{a,a}\degen\langle2a-1\rangle$ give $\operatorname{irr}(S_a)\leq\log(2a-1)/\log a$, so the barrier tends to the trivial value~$2$ as $a\to\infty$. A similar analysis for rectangular matrix multiplication gives no irreversibility barrier there.

\subsubsection{Tight barrier for (D0)--(D4) + (C1)--(C4)}\label{sec:intro-certificate}
No derivation from the ingredients~\ruleref{D0}--\ruleref{D4} by the rules~\ruleref{C1}--\ruleref{C4}, with~\ruleref{D3} and~\ruleref{D4} also inside partial direct sums, gives an upper bound on~$\omega_\FF$ below~$9/4$, and neither does a limit of such bounds (\cref{thm:barrier}). This remains true for the asymptotic consequences of the rules, among them the asymptotic restriction~\eqref{eq:intro-distil}. Likewise, for all real numbers $a\geq b\geq c\geq0$, no such derivation gives a better upper bound on~$\omega_\FF(a,b,c)$ than \cref{thm:rect}.

We prove the barrier with a dual certificate. For all values of their parameters, write the degenerations of the ingredients~\ruleref{D0}--\ruleref{D4} as $S_i\degen T_i$, for~$i$ in an index set~$I$. Add the restrictions $T\leq T\oplus S$, $T\leq T\oplusX S$, $T\oplusX S\leq T\oplus S$, $\polmul_{a,c}\leq\polmul_{a,b}$ for $c\leq b$, and $\polmul_{a,b}\leq\polmul_{a,1}\otimes\langle1,1,b\rangle$, and the isomorphisms that express the associativity and commutativity of~$\oplus$, $\otimes$ and~$\oplusX$ and the distributivity of~$\otimes$ over~$\oplus$ and~$\oplusX$, the identities of the permutations of the legs, the isomorphisms between products and permutations of matrix multiplication tensors, the isomorphism between the partial direct sum of $r$~copies of~$T$ and $T\otimes\langle1,1,r\rangle$, and $\polmul_{a,b}\cong(12)\cdot\polmul_{b,a}$ (\cref{sec:combinations}). By the rules~\ruleref{C1}--\ruleref{C4}, with~\ruleref{D3} and~\ruleref{D4} also applied inside partial direct sums, these generate a relation~$\preceq$ on formal combinations of the tensors.

Let~$\Delta_{\preceq}$ be the set of functions~$\phi$ on these formal combinations that are additive and multiplicative, satisfy $\phi(\langle1\rangle)=1$, and satisfy $\phi(T)\leq\phi(S)$ whenever $T\preceq S$; in particular, $\phi(S_i)\leq\phi(T_i)$ for every $i\in I$. The set~$\Delta_{\preceq}$ does not depend on the field~$\FF$. For every field~$\FF$, every spectral point agrees with an element of~$\Delta_{\preceq}$ on the tensors $\pi\cdot\polmul_{a,b}$ and the matrix multiplication tensors, so~$\Delta_{\preceq}$ is an outer approximation of the asymptotic spectrum: $\Delta_\FF\subseteq\Delta_{\preceq}$ in this sense. The proof of \cref{thm:main} is in terms of spectral points $\phi\in\Delta_\FF$, but it works almost verbatim for every $\phi\in\Delta_{\preceq}$, so
\begin{equation}\label{eq:intro-chain}
2^{\omega_\FF}=\max_{\phi\in\Delta_\FF}\phi(\langle2,2,2\rangle)\leq\max_{\phi\in\Delta_{\preceq}}\phi(\langle2,2,2\rangle)\leq2^{9/4}.
\end{equation}

The point is that the second inequality in~\eqref{eq:intro-chain} is an equality: we will construct an explicit function $F^\star\in\Delta_{\preceq}$ with $F^\star(\langle2,2,2\rangle)=2^{9/4}$ (\cref{sec:optimality}), a ``dual certificate''. This proves the barrier, since a relation $\langle n,n,n\rangle\preceq\langle r\rangle$ gives $n^{9/4}=F^\star(\langle n,n,n\rangle)\leq F^\star(\langle r\rangle)=r$. The same argument applies if new degenerations $T\degen S$ are added as rules, as long as $F^\star(\pi\cdot T)\leq F^\star(\pi\cdot S)$ for every permutation~$\pi$ of the legs (\cref{cor:new-rules}). So a new degeneration can lead to a bound below~$9/4$, together with the ingredients, only if it violates one of these inequalities. For all real numbers $a\geq b\geq c\geq0$, other functions in~$\Delta_{\preceq}$ of the same kind show in the same way that no derivation gives a better upper bound on~$\omega_\FF(a,b,c)$ than \cref{thm:rect}. 

We do not know whether~$F^\star$ agrees with a spectral point on the tensors $\pi\cdot\polmul_{a,b}$ and the matrix multiplication tensors. If it did, the first inequality in~\eqref{eq:intro-chain} would be an equality as well, and $\omega_\FF=9/4$.

\subsection{Organisation of the paper}\label{sec:organisation}
In \cref{sec:tensors} we will define tensors, restrictions and degenerations. In \cref{sec:constructions} we will construct the degenerations of ingredients~\ruleref{D1}--\ruleref{D4}. In \cref{sec:spectral} we will introduce spectral points, obtain $\omega_\FF<2.3212$ from a tensor of Strassen as a first example, and prove \cref{thm:main} and \cref{thm:rect}. In \cref{sec:optimality} we will show that the current ingredients  of the proofs cannot give better bounds.

\section{Restriction and degeneration}\label{sec:tensors}

In this section we define tensors, restrictions and degenerations, and state the facts about them that we will use. From now on, $\FF$ is a field, and all vector spaces are finite-dimensional over~$\FF$. We use the convention $0^0=1$, and for an integer $n\geq0$ we write $[n]=\{1,\ldots,n\}$.

\subsection{Tensors}\label{sec:tensors-basic}
A \emph{tensor} over~$\FF$ is an element $T$ of a tensor product $U\otimes V\otimes W$ of three vector spaces $U,V$ and $W$, which we call the first, second and third \emph{leg} of $T$. For tensors $S\in U_1\otimes V_1\otimes W_1$ and $T\in U_2\otimes V_2\otimes W_2$, the \emph{direct sum} $S\oplus T$ lies in $(U_1\oplus U_2)\otimes(V_1\oplus V_2)\otimes(W_1\oplus W_2)$, and the \emph{tensor product} $S\otimes T$ lies in $(U_1\otimes U_2)\otimes(V_1\otimes V_2)\otimes(W_1\otimes W_2)$, after regrouping the factors.

We define a version of the direct sum in which the first leg is shared. For an integer $m\geq1$, vector spaces~$U$, $V_1,\ldots,V_m$ and $W_1,\ldots,W_m$, and tensors $T_i\in U\otimes V_i\otimes W_i$ for $i\in[m]$, their \emph{partial direct sum}\phantomsection\label{def:partial-direct-sum} is %
\[
T_1\oplusX\cdots\oplusX T_m\coloneqq T_1+\cdots+T_m\in U\otimes(V_1\oplus\cdots\oplus V_m)\otimes(W_1\oplus\cdots\oplus W_m).
\]%
\Cref{fig:partial-direct-sum} compares the partial direct sum with the direct sum.

Let~$\mathfrak{S}_3$ be the group of permutations of~$\{1,2,3\}$. For a tensor $T\in U\otimes V\otimes W$ and a permutation $\pi\in\mathfrak{S}_3$, let~$\pi\cdot T$ be the tensor obtained from~$T$ by permuting the factors, so that, for every $i\in\{1,2,3\}$, the~$i$th leg of~$T$ is the~$\pi(i)$th leg of~$\pi\cdot T$, as in~\cite[Section~2]{MR882307}.

For an integer $n\geq0$, let $e_1,\ldots,e_n$ be the standard basis of~$\FF^n$. The \emph{unit tensor} is $\langle n\rangle\coloneqq\sum_{i=1}^ne_i\otimes e_i\otimes e_i$. For integers $n_1,n_2,n_3\geq1$, the \emph{matrix multiplication tensor} is
\[
\langle n_1,n_2,n_3\rangle\coloneqq\sum_{i=1}^{n_1}\sum_{j=1}^{n_2}\sum_{k=1}^{n_3} E_{ij}\otimes E_{jk}\otimes E_{ki}\in\FF^{n_1\times n_2}\otimes\FF^{n_2\times n_3}\otimes\FF^{n_3\times n_1},
\]
where~$E_{ij}$ denotes the matrix with entry~$1$ at~$(i,j)$ and~$0$ elsewhere.

Generally, the tensor of a bilinear map $\mu\colon U\times V\to W$ is $T = \sum_{i,j}e_i^*\otimes e_j^*\otimes\mu(e_i,e_j)$ in $U^*\otimes V^*\otimes W$,
where the sums run over bases of~$U$ and~$V$ and~$e_i^*$ is the dual basis element of~$e_i$. We get~$\mu$ back from $T$ by evaluating the first two factors: $\mu(u,v)=\sum_{i,j}e_i^*(u)\,e_j^*(v)\,\mu(e_i,e_j)$ for all $u\in U$ and $v\in V$. %

\subsection{Restriction and tensor rank}\label{sec:tensors-rank}
Let $S\in U_1\otimes V_1\otimes W_1$ and $T\in U_2\otimes V_2\otimes W_2$ be tensors. Then~$T$ is a \emph{restriction} of~$S$, written $T\leq S$, if there are linear maps $A_1:U_1\to U_2$, $A_2:V_1\to V_2$ and $A_3:W_1\to W_2$ such that $T=(A_1\otimes A_2\otimes A_3)S$. If these maps can be chosen invertible, then~$S$ and~$T$ are \emph{isomorphic}, written $S\cong T$.

Two tensors are equivalent if each is a restriction of the other. Let~$\mathcal{S}_\FF$ be the set of equivalence classes. Strassen showed that direct sum and tensor product make~$\mathcal{S}_\FF$ a commutative semiring: both operations are associative and commutative and have identity elements, and the tensor product distributes over the direct sum~\cite[Section~2]{MR882307}. Restriction induces a partial order on~$\mathcal{S}_\FF$, and it is compatible with both operations: if $T_1\leq S_1$ and $T_2\leq S_2$, then $T_1\oplus T_2\leq S_1\oplus S_2$ and $T_1\otimes T_2\leq S_1\otimes S_2$~\cite[Section~3]{MR882307}.

The \emph{tensor rank} $\R_\FF(T)$ of a tensor $T\in U\otimes V\otimes W$ is the least integer $r\geq0$ such that~$T = \sum_{i=1}^r u_i \otimes v_i \otimes w_i$ with $u_i\in U$, $v_i\in V$ and $w_i\in W$ for $i \in [r]$. The \emph{exponent of matrix multiplication} over~$\FF$ is
\[
\omega_\FF\coloneqq\inf\{\tau\in\RR:\R_\FF(\langle n,n,n\rangle)=\mathcal{O}(n^\tau)\},
\]
as in~\cite[(1.3)]{MR882307}.

The \emph{flattening ranks} $\flatrank{1}(T)$, $\flatrank{2}(T)$ and~$\flatrank{3}(T)$ of a tensor $T\in U\otimes V\otimes W$ are the matrix ranks of the images of~$T$ in $U\otimes(V\otimes W)$, $V\otimes(U\otimes W)$ and $W\otimes(U\otimes V)$.

\begin{lemma}\label{lem:restriction-facts}
Let~$S$ and~$T$ be tensors over~$\FF$.
\begin{enumerate}[\upshape(i)]
\item\label{rf:unit} If $T\neq0$, then $\langle1\rangle\leq T$.
\item\label{rf:rank} The tensor rank $\R_\FF(T)$ is the least integer $r\geq0$ such that $T\leq\langle r\rangle$.
\item\label{rf:product} For all integers $n_1,n_2,n_3,m_1,m_2,m_3\geq1$, we have 
\[
\langle n_1,n_2,n_3\rangle\otimes\langle m_1,m_2,m_3\rangle\cong\langle n_1m_1,n_2m_2,n_3m_3\rangle.
\]
\item\label{rf:permutation} For every $\pi\in\mathfrak{S}_3$, we have $\pi\cdot(S\oplus T)\cong(\pi\cdot S)\oplus(\pi\cdot T)$, $\pi\cdot(S\otimes T)\cong(\pi\cdot S)\otimes(\pi\cdot T)$ and $\pi\cdot\langle1\rangle=\langle1\rangle$, and if $T\leq S$, then $\pi\cdot T\leq\pi\cdot S$.
\item\label{rf:cyclic} For all integers $n_1,n_2,n_3\geq1$,
\[
\begin{aligned}
(123)\cdot\langle n_1,n_2,n_3\rangle&=\langle n_3,n_1,n_2\rangle, &\qquad (12)\cdot\langle n_1,n_2,n_3\rangle&\cong\langle n_3,n_2,n_1\rangle,\\
(132)\cdot\langle n_1,n_2,n_3\rangle&=\langle n_2,n_3,n_1\rangle, & (13)\cdot\langle n_1,n_2,n_3\rangle&\cong\langle n_1,n_3,n_2\rangle,\\
&& (23)\cdot\langle n_1,n_2,n_3\rangle&\cong\langle n_2,n_1,n_3\rangle.
\end{aligned}
\]
Hence $\bigotimes_{\pi\in\mathfrak{S}_3}\pi\cdot\langle n_1,n_2,n_3\rangle\cong\langle n_1n_2n_3,n_1n_2n_3,n_1n_2n_3\rangle^{\otimes2}$.
\end{enumerate}
\end{lemma}
\begin{proof}
Part~(\ref{rf:unit}) is in~\cite[Section~3, before Proposition~3.3]{MR882307}, and part~(\ref{rf:rank}) is~\cite[Proposition~3.2]{MR882307}. For part~(\ref{rf:product}), the maps $E_{ij}\otimes E_{k\ell}\mapsto E_{(i,k),(j,\ell)}$ on the three legs give the isomorphism. Part~(\ref{rf:permutation}) holds since permuting the factors of a tensor product commutes with direct sums, tensor products and linear maps on the factors. For part~(\ref{rf:cyclic}), write $\langle n_1,n_2,n_3\rangle=\sum_{i,j,k}E_{ij}\otimes E_{jk}\otimes E_{ki}$. The permutation~$(123)$ moves the third leg to the front, which gives $\sum_{i,j,k}E_{ki}\otimes E_{ij}\otimes E_{jk}=\langle n_3,n_1,n_2\rangle$, and~$(132)$ moves the first leg to the back, which gives $\sum_{i,j,k}E_{jk}\otimes E_{ki}\otimes E_{ij}=\langle n_2,n_3,n_1\rangle$. The permutation~$(12)$ gives $\sum_{i,j,k}E_{jk}\otimes E_{ij}\otimes E_{ki}$, and transposing the matrices on all three legs turns it into $\sum_{i,j,k}E_{kj}\otimes E_{ji}\otimes E_{ik}=\langle n_3,n_2,n_1\rangle$. In the same way, for each transposition, exchanging the two legs and transposing the matrices keeps the parameter that the two legs have in common in its place and exchanges the other two: legs~$1$ and~$3$ have~$n_1$ in common, and legs~$2$ and~$3$ have~$n_3$. The six tensors~$\pi\cdot\langle n_1,n_2,n_3\rangle$ are thus isomorphic to the six tensors $\langle n_{\sigma(1)},n_{\sigma(2)},n_{\sigma(3)}\rangle$ with $\sigma\in\mathfrak{S}_3$, and by part~(\ref{rf:product}) their tensor product is isomorphic to $\langle n_1n_2n_3,n_1n_2n_3,n_1n_2n_3\rangle^{\otimes2}$.
\end{proof}

\subsection{Degeneration and border rank}\label{sec:degenerations}
Let~$\eps$ be an indeterminate over~$\FF$, let $S\in U_1\otimes V_1\otimes W_1$ and $T\in U_2\otimes V_2\otimes W_2$ be tensors, and choose bases of their legs. Then~$T$ is a \emph{degeneration} of~$S$, written $T\degen S$, if there are integers $q,d\geq0$, matrices~$A_1$, $A_2$ and~$A_3$ with entries in~$\FF[\eps]$ of sizes $\dim U_2\times\dim U_1$, $\dim V_2\times\dim V_1$ and $\dim W_2\times\dim W_1$, and tensors $T_1,\ldots,T_d\in U_2\otimes V_2\otimes W_2$, such that
\[
(A_1\otimes A_2\otimes A_3)S=\eps^qT+\eps^{q+1}T_1+\cdots+\eps^{q+d}T_d.
\]
This is the polynomial formulation of Strassen's definition~\cite[Lemmas~5.5 and~5.9]{MR882307}; see also the book of Bürgisser, Clausen and Shokrollahi~\cite[Definition~15.19]{burgisser1997algebraic}. The \emph{border rank} of a tensor~$T$ is the least integer $r\geq0$ such that $T\degen\langle r\rangle$.

From now on, $z$ is an indeterminate over~$\FF$. The field~$\FF(z)$ of rational functions in~$z$ with coefficients in~$\FF$ is an infinite extension field of~$\FF$, and every tensor over~$\FF$ is also a tensor over~$\FF(z)$.

\begin{lemma}[{Strassen~\cite[Lemmas~5.6 and~5.9]{MR882307}}]\label{lem:degen-facts}
Let~$S$ and~$T$ be tensors over~$\FF$.
\begin{enumerate}[\upshape(i)]
\item\label{df:restriction} If $T\leq S$, then $T\degen S$.
\item\label{df:transitive} $\degen$ is transitive.
\item\label{df:product} For all tensors $S_1$, $T_1$, $S_2$ and~$T_2$ over~$\FF$ with $T_1\degen S_1$ and $T_2\degen S_2$, we have $T_1\oplus T_2\degen S_1\oplus S_2$ and $T_1\otimes T_2\degen S_1\otimes S_2$.
\item\label{df:extension} If $T\leq S$ as tensors over~$\FF(z)$, then $T\degen S$.
\item\label{df:permutation} For every $\pi\in\mathfrak{S}_3$, if $T\degen S$, then $\pi\cdot T\degen\pi\cdot S$.
\end{enumerate}
\end{lemma}
\begin{proof}
Part~(\ref{df:restriction}) holds with $q=d=0$ and the matrices of the restriction. Parts~(\ref{df:transitive}) and~(\ref{df:product}) are~\cite[Lemma~5.6(ii) and~(iii)]{MR882307}. 
Part~(\ref{df:permutation}) follows from the definition of degeneration, by permuting the three matrices.

Part~(\ref{df:extension}) is the implication from~(ii) to~(i) in~\cite[Lemma~5.9]{MR882307}, with $g=0$ and~$z$ in the role of~$\eps$. We give a proof for convenience of the reader.
Suppose that $T=(A_1\otimes A_2\otimes A_3)S$ for matrices~$A_i$ with entries in~$\FF(z)$. Choose nonzero polynomials $p_1,p_2,p_3\in\FF[z]$ such that the matrices $B_i=p_iA_i$ have entries in~$\FF[z]$. Then $(B_1\otimes B_2\otimes B_3)S=fT$, where $f=p_1p_2p_3$ is a nonzero polynomial. %
Write $f(z)=\lambda z^q+z^{q+1}g(z)$ with $\lambda\in\FF\setminus\{0\}$, an integer $q\geq0$ and $g\in\FF[z]$. Then
\[
(\lambda^{-1}B_1\otimes B_2\otimes B_3)S=z^qT+z^{q+1}\lambda^{-1}g(z)T,
\]
and replacing~$z$ by~$\eps$ shows that $T\degen S$.
\end{proof}

\subsection{Degeneration from block triangular tensors and weights}\label{sec:monomial}
We will give a degeneration from a block triangular tensor to its diagonal blocks, and then a sufficient condition for degeneration in terms of weights. The first degeneration acts only on the second and third legs, in the following sense.

Let $S\in U\otimes V_1\otimes W_1$ and $T\in U\otimes V_2\otimes W_2$ be tensors with the same first leg~$U$, and choose bases of their legs, with the same basis of~$U$ for both. Then~$T$ is a \emph{degeneration of~$S$ on the second and third legs}, written $T\degenTT S$, if $T\degen S$ with the identity matrix~$I$ as~$A_1$ (\cref{sec:degenerations}), that is, if
\[
(I\otimes A_2\otimes A_3)S=\eps^qT+\eps^{q+1}T_1+\cdots+\eps^{q+d}T_d
\]
for integers $q,d\geq0$, matrices~$A_2$ and~$A_3$ with entries in~$\FF[\eps]$, and tensors $T_1,\ldots,T_d\in U\otimes V_2\otimes W_2$. Like~$\degen$ (\cref{lem:degen-facts}(\ref{df:transitive})), the relation~$\degenTT$ is transitive: after~$\eps$ is replaced by a suitable power of~$\eps$ in the degeneration that is applied first, the matrices of two such degenerations compose leg by leg, so the matrix on the first leg stays the identity.

\begin{lemma}[{Strassen~\cite[Proposition~6.2]{MR882307}}]
\label{lem:triangular degeneration}
Let~$U$, $V_1$, $V_2$, $W_1$ and~$W_2$ be vector spaces. Write $T \in U \otimes (V_1 \oplus V_2) \otimes (W_1 \oplus W_2)$ as $T = T_{1,1} + T_{1,2} + T_{2,1} + T_{2,2}$ with $T_{i,j} \in U \otimes V_i \otimes W_j$. If $T_{2,1} = 0$, then $T \degengeqTT T_{1,1}\oplusX T_{2,2}$. In block form, with~$T_{i,j}$ in row~$i$ and column~$j$, this reads
\[
\begin{pmatrix}T_{1,1}&T_{1,2}\\0&T_{2,2}\end{pmatrix}\degengeqTT\begin{pmatrix}T_{1,1}&0\\0&T_{2,2}\end{pmatrix}.
\]
\end{lemma}
\begin{proof}
    Let $B(\eps)=\eps I_{V_1}\oplus I_{V_2}$ and $C(\eps)=I_{W_1}\oplus\eps I_{W_2}$.
    Then $(I\otimes B(\eps)\otimes C(\eps))T=\eps(T_{1,1}+T_{2,2})+\eps^2T_{1,2}+T_{2,1}$, which is $\eps(T_{1,1}+T_{2,2})+\eps^2T_{1,2}$ since $T_{2,1}=0$.
\end{proof}
\Cref{fig:triangular} draws \cref{lem:triangular degeneration} in three dimensions. Strassen applied the general version, for filtrations, to show that every $n$-dimensional algebra with unit degenerates to the null algebra, in which $u_iu_n=u_nu_i=u_i$ and $u_iu_j=0$ for $i,j<n$, and that $\langle\lfloor(h+1)^2/4\rfloor\rangle\degen D_h\degen\langle h,h,h\rangle$ for the algebra~$D_h$ of upper triangular $h\times h$ matrices~\cite[Section~6]{MR882307}.

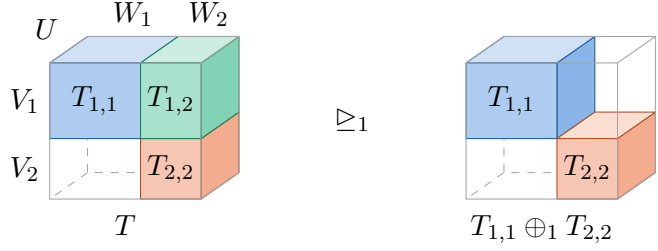
\begin{figure}[t]
\centering
\begin{tikzpicture}[x={(1cm,0cm)},y={(0cm,-1cm)},z={(0.5cm,0.35cm)},line join=round,thin]
\begin{scope}
\psumframeback{2}{1.8}{1}
\psumbox{psumB}{1.2}{1}{0}{2}{1.8}{1}
\psumbox{psumA}{0}{0}{0}{1.2}{1}{1}
\psumbox{psumC}{1.2}{0}{0}{2}{1}{1}
\psumframefront{2}{1.8}{1}
\node at (0.6,0.5,0) {$T_{1,1}$};
\node at (1.6,0.5,0) {$T_{1,2}$};
\node at (1.6,1.4,0) {$T_{2,2}$};
\node[left] at (0,0.5,0) {$V_1$};
\node[left] at (0,1.4,0) {$V_2$};
\node[above] at (0.6,0,1) {$W_1$};
\node[above] at (1.6,0,1) {$W_2$};
\node[above left] at (0,0,0.5) {$U$};
\node[below=2pt] at (1,1.8,0) {$T$};
\end{scope}
\node at (4,0.75,0) {$\degengeqTT$};
\begin{scope}[xshift=5.5cm]
\psumframeback{2}{1.8}{1}
\psumbox{psumA}{0}{0}{0}{1.2}{1}{1}
\psumbox{psumB}{1.2}{1}{0}{2}{1.8}{1}
\psumframefront{2}{1.8}{1}
\node at (0.6,0.5,0) {$T_{1,1}$};
\node at (1.6,1.4,0) {$T_{2,2}$};
\node[below=2pt] at (1,1.8,0) {$T_{1,1}\oplus_1T_{2,2}$};
\end{scope}
\end{tikzpicture}
\caption{\Cref{lem:triangular degeneration}, with the tensors drawn as three-dimensional arrays: the first leg~$U$ points into the page, the second leg $V_1\oplus V_2$ downwards and the third leg $W_1\oplus W_2$ to the right. The coefficients of~$T_{i,j}$ lie in the box in block row~$i$ and block column~$j$, and all other coefficients are zero. The tensor~$T$ (left) has no box in block row~$2$ and block column~$1$, and degenerates on the second and third legs to the sum $T_{1,1}+T_{2,2}$ of its diagonal blocks (right).}
\label{fig:triangular}
\end{figure}
Since the decompositions $V_1\oplus V_2$ and $W_1\oplus W_2$ are arbitrary, \cref{lem:triangular degeneration} also applies to a tensor with more blocks that has this block form after its block rows and block columns are permuted and grouped. We will use this for three blocks.

\begin{lemma}\label{lem:three-block}
Let~$U$, $V_1$, $V_2$, $V_3$, $W_1$, $W_2$ and~$W_3$ be vector spaces. Write $T\in U\otimes(V_1\oplus V_2\oplus V_3)\otimes(W_1\oplus W_2\oplus W_3)$ as $T=\sum_{i,j=1}^3T_{i,j}$ with $T_{i,j}\in U\otimes V_i\otimes W_j$. If $T_{i,j}=0$ for all distinct $i,j\in[3]$ with $i\neq2$, then $T\degengeqTT T_{1,1}\oplusX T_{2,2}\oplusX T_{3,3}$. In block form, with~$T_{i,j}$ in row~$i$ and column~$j$, this reads
\[
\begin{pmatrix}T_{1,1}&0&0\\T_{2,1}&T_{2,2}&T_{2,3}\\0&0&T_{3,3}\end{pmatrix}\degengeqTT\begin{pmatrix}T_{1,1}&0&0\\0&T_{2,2}&0\\0&0&T_{3,3}\end{pmatrix}.
\]
\end{lemma}
\begin{proof}
Apply \cref{lem:triangular degeneration} to the decompositions $V_2\oplus(V_1\oplus V_3)$ of the second leg and $W_2\oplus(W_1\oplus W_3)$ of the third leg; this moves the middle block row and block column to the front and groups the other two. The component of~$T$ in $U\otimes(V_1\oplus V_3)\otimes W_2$ is $T_{1,2}+T_{3,2}=0$. So the lemma gives $T\degengeqTT T_{2,2}\oplusX(T_{1,1}+T_{1,3}+T_{3,1}+T_{3,3})$, which is $T_{1,1}\oplusX T_{2,2}\oplusX T_{3,3}$ since $T_{1,3}=T_{3,1}=0$.
\end{proof}

The following lemma states \cref{lem:triangular degeneration} in the language of bilinear maps and quotients.
\begin{lemma}[{Strassen~\cite[Proposition~6.2]{MR882307}}]\label{lem:restriction-quotient}
Let $\mu\colon U\times V\to W$ be a bilinear map, and let $V_1\subseteq V$ and $W_1\subseteq W$ be subspaces with $\mu(U\times V_1)\subseteq W_1$. Let $\mu_1$ be the induced map of $\mu$ on the quotients,
\[
\mu_1\colon U\times(V/V_1)\to W/W_1,
\]
and let $\mu_2$ be the induced map of $\mu$ on the subspaces,
\[
\mu_2\colon U\times V_1\to W_1.
\]
Let $T$, $T_1$ and~$T_2$ be the tensors of~$\mu$, $\mu_1$ and~$\mu_2$. Then
\[
T\degengeqTT T_1\oplusX T_2.
\]
\end{lemma}
\begin{proof}
Strassen's proposition, applied to the filtrations $V\supseteq V_1\supseteq0$ and $W\supseteq W_1\supseteq0$ and the trivial filtration of~$U$, gives $T_1\oplusX T_2\degen T$. We give a proof for convenience of the reader, which also keeps the first leg fixed. Choose complements~$V_2$ of~$V_1$ in~$V$ and~$W_2$ of~$W_1$ in~$W$. Extending functionals on each summand of~$V$ by zero on the other summand gives $V^*=V_1^*\oplus V_2^*$, and $W=W_1\oplus W_2$. The quotient maps restrict to isomorphisms $V_2\cong V/V_1$ and $W_2\cong W/W_1$. Since $\mu(U\times V_1)\subseteq W_1$, the component of~$T$ in $U^*\otimes V_1^*\otimes W_2$ is zero. Under these isomorphisms, the components of~$T$ in $U^*\otimes V_2^*\otimes W_2$ and in $U^*\otimes V_1^*\otimes W_1$ are~$T_1$ and~$T_2$. So \cref{lem:triangular degeneration}, applied with the first leg~$U^*$ and with the summands of the second and third legs in the orders $V_2^*,V_1^*$ and $W_2,W_1$, gives $T\degengeqTT T_1\oplusX T_2$.
\end{proof}

More generally, we can use integer-valued weights to algebraically select a region to degenerate to, as follows.

\begin{lemma}[{Strassen~\cite[Theorem~6.1 and~(6.6)]{MR882307}}]\label{lem:monomial-degen}
Let $n_1,n_2,n_3\geq0$ be integers, let $S,T\in\FF^{n_1}\otimes\FF^{n_2}\otimes\FF^{n_3}$ be tensors, and write $S=\sum_{a,b,c}S_{abc}\,e_a\otimes e_b\otimes e_c$ and $T=\sum_{a,b,c}T_{abc}\,e_a\otimes e_b\otimes e_c$ in the standard bases. Suppose that there are functions $w_i:[n_i]\to\ZZ$ for $i\in\{1,2,3\}$ such that $w_1(a)+w_2(b)+w_3(c)\geq0$ for all $a$, $b$ and~$c$ with $S_{abc}\neq0$, and, for all $a$, $b$ and~$c$,
\[
T_{abc}=\begin{cases}S_{abc}&\text{if }w_1(a)+w_2(b)+w_3(c)=0,\\0&\text{otherwise.}\end{cases}
\]
Then $S\degengeq T$, and $S\degengeqTT T$ if $w_1=0$.
\end{lemma}
\begin{proof}
For every $i\in\{1,2,3\}$, let $k_i\geq0$ be an integer such that $w_i(a)+k_i\geq0$ for every $a\in[n_i]$, with $k_1=0$ if $w_1=0$, and let~$A_i$ be the diagonal matrix with entries~$\eps^{w_i(a)+k_i}$ for $a\in[n_i]$. Let $k\coloneqq k_1+k_2+k_3$. Then $(A_1\otimes A_2\otimes A_3)S=\sum_{a,b,c}\eps^{w_1(a)+w_2(b)+w_3(c)+k}S_{abc}\,e_a\otimes e_b\otimes e_c$. Every nonzero term in this sum has exponent at least~$k$, and the terms with exponent~$k$ sum to~$\eps^kT$. If $w_1=0$, then~$A_1$ is the identity.
\end{proof}
\Cref{lem:triangular degeneration} also follows from \cref{lem:monomial-degen}, applied in bases of~$U$, $V_1$, $V_2$, $W_1$ and~$W_2$ with $w_1=0$, with~$w_2$ equal to~$1$ on~$V_1$ and~$0$ on~$V_2$, and with~$w_3$ equal to~$-1$ on~$W_1$ and~$0$ on~$W_2$.

All the above degenerations are of a special kind, called \emph{monomial degeneration}: in the chosen basis the degeneration matrices have at most one nonzero entry per row and per column, and that entry is a power of~$\eps$. Strassen calls them $M$-degenerations~\cite[Section~6]{MR882307}.

\section{Degenerations among polynomial multiplication tensors}\label{sec:constructions}

In this section we will construct, over every field, the degenerations that the proofs of \cref{thm:main} and \cref{thm:rect} use: degenerations between unit tensors and polynomial multiplication tensors (\cref{sec:conv}), the degeneration from a partial direct sum to a direct sum (\cref{sec:separation}), and two degenerations between polynomial multiplication tensors (\cref{sec:det} and \cref{sec:tripling}).

\subsection{Polynomial multiplication tensors}\label{sec:conv}

We will define the polynomial multiplication tensors and show that, over every field, $\polmul_{a,b}$ has border rank at most $a+b-1$ and $\polmul_{n,n}$ degenerates to~$\langle n\rangle$.

\begin{definition}\label{def:polmul}
For integers $a,b\geq1$, the \emph{polynomial multiplication tensor} is
\[
\polmul_{a,b}\coloneqq\sum_{i=1}^a\sum_{j=1}^be_i\otimes e_j\otimes e_{i+j-1}\in\FF^a\otimes\FF^b\otimes\FF^{a+b-1}.
\]
\end{definition}
The tensor~$\polmul_{a,b}$ corresponds to the multiplication of a polynomial of degree at most~$a-1$ by a polynomial of degree at most~$b-1$: if~$e_i$ stands for the monomial~$t^{i-1}$ in an indeterminate~$t$, then the term $e_i\otimes e_j\otimes e_{i+j-1}$ of~$\polmul_{a,b}$ records the product $t^{i-1}t^{j-1}=t^{i+j-2}$. 
For an integer $d\geq0$, let $\FF[t]_{\leq d}$ be the space of polynomials of degree at most~$d$ in~$t$. In this way the tensor~$\polmul_{a,b}$ corresponds to the bilinear multiplication map
\begin{align}
\label{eq:pol mult as bilinear map}
\FF[t]_{\leq a-1} \times \FF[t]_{\leq b-1} \to \FF[t]_{\leq a+b-2}: (f,g) \mapsto fg.
\end{align}

For $a=1$, $\polmul_{1,b}=\sum_{j=1}^be_1\otimes e_j\otimes e_j$ is isomorphic to $\langle1,1,b\rangle=\sum_{k=1}^bE_{11}\otimes E_{1k}\otimes E_{k1}$: it describes the multiplication of a polynomial of degree at most~$b-1$ by a constant, which is nothing else than multiplying a vector of length $b$ by a constant.

\begin{proposition}\label{prop:polmul-rank}
Let $a,b\geq1$ be integers. Then $\polmul_{a,b}\degen\langle a+b-1\rangle$.
\end{proposition}
\begin{proof}
We extend the field by a formal variable $z$ and consider the function field $\FF(z)$.
The field~$\FF(z)$ is infinite 
so $\R_{\FF(z)}(\polmul_{a,b})\leq a+b-1$~\cite[Proposition~14.47]{burgisser1997algebraic}, that is, $\polmul_{a,b}\leq\langle a+b-1\rangle$ as tensors over~$\FF(z)$ (\cref{lem:restriction-facts}(\ref{rf:rank})). A restriction over~$\FF(z)$ is a degeneration over~$\FF$ (\cref{lem:degen-facts}(\ref{df:extension})), which proves the claim.
\end{proof}
\begin{lemma}[{\cite[Proposition~14.48(1)]{burgisser1997algebraic}}]\label{lem:unit-polmul}
For every integer $n\geq1$, we have $\langle n\rangle\degen\polmul_{n,n}$.
\end{lemma}
\begin{proof}
If~$\FF$ has at least~$n$ elements, the cited proposition gives $\langle n\rangle\leq\polmul_{n,n}$: for a product~$f$ of~$n$ distinct linear factors, the multiplication tensor of the algebra $\FF[t]/(f)\cong\FF^n$ is~$\langle n\rangle$. Over every field, we apply \cref{lem:monomial-degen} to $S=\polmul_{n,n}$, with the weights $w_1(i)=2(i-1)^2$, $w_2(j)=2(j-1)^2$ and $w_3(k)=-(k-1)^2$. The term $e_i\otimes e_j\otimes e_{i+j-1}$ of~$\polmul_{n,n}$ has weight $2(i-1)^2+2(j-1)^2-(i+j-2)^2=(i-j)^2\geq0$, which is~$0$ exactly when $i=j$. So $\polmul_{n,n}\degengeq\sum_{i=1}^ne_i\otimes e_i\otimes e_{2i-1}$. The linear map on the third leg that maps~$e_{2i-1}$ to~$e_i$ for every $i\in[n]$ and the other basis vectors to~$0$ maps this tensor to~$\langle n\rangle$, so $\langle n\rangle\degen\polmul_{n,n}$.
\end{proof}

\subsection{Separating partial direct sums into direct sums}\label{sec:separation}

In this subsection we will prove a separation lemma that degenerates any partial direct sum (\cref{def:partial-direct-sum}) to a direct sum, in the presence of a polynomial multiplication tensor factor.

A degeneration on the second and third legs (\cref{sec:monomial}) can be applied to some of the summands of a partial direct sum, while the other summands stay as they are: if $T\degenTT S$, then $T\oplusX C\degenTT S\oplusX C$ for every tensor~$C$ with the same first leg as~$S$ and~$T$. If $(I\otimes A_2\otimes A_3)S=\eps^qT+\eps^{q+1}T_1+\cdots+\eps^{q+d}T_d$, then $I\otimes(A_2\oplus\eps^qI)\otimes(A_3\oplus I)$ maps $S\oplusX C$ to $\eps^q(T\oplusX C)+\eps^{q+1}T_1+\cdots+\eps^{q+d}T_d$.

\begin{lemma}\label{lem:two-copies}
Let $T\in U\otimes V\otimes W$ be a tensor and let $k\geq1$ be an integer. Then the partial direct sum $T\oplusX\cdots\oplusX T$ of $k$ copies of~$T$ is isomorphic to $T\otimes\langle1,1,k\rangle$.
\end{lemma}
\begin{proof}
We have $\langle1,1,k\rangle=\sum_{\ell=1}^kE_{11}\otimes E_{1\ell}\otimes E_{\ell1}$. So $T\otimes\langle1,1,k\rangle$ is a partial direct sum of $k$ copies of~$T$, the $\ell$th in $(U\otimes\FF^{1\times1})\otimes(V\otimes\FF E_{1\ell})\otimes(W\otimes\FF E_{\ell1})$.
\end{proof}

\begin{proposition}[From partial direct sum to direct sum]\label{prop:sep}
Let $m,t\geq1$ be integers. 
Let $T_1,\ldots, T_m$ be tensors with $T_\ell \in U \otimes V_\ell \otimes W_\ell$.
Then
\[
(T_1\oplus\cdots\oplus T_m)\otimes\polmul_{1,t}
\degen (T_1\oplusX\cdots\oplusX T_m) \otimes \polmul_{m,t}.
\]
\end{proposition}

Note how in the degeneration of \cref{prop:sep}, the right-hand side has the polynomial multiplication tensor $\polmul_{m,t}$ as a factor, and the left-hand side the polynomial multiplication tensor $\polmul_{1,t}$. We may think of $\polmul_{m,t}$ as a resource that enables the degeneration while being partially ``consumed'' in the process, leaving $\polmul_{1,t}$ in the result.

\begin{proof}
Let $T=T_1\oplusX\cdots\oplusX T_m$.
Without loss of generality (embed the spaces into~$\FF^n$ for a large enough~$n$, which gives restrictions in both directions), we take $U = V_\ell = W_\ell = \FF^n$. 
Let $\{e_a\}_a$, $\{e_b^\ell\}_b$ and $\{e_c^\ell\}_c$ be the standard bases for $U$, $V_\ell$ and $W_\ell$ respectively.
By definition,
\begin{align}
\label{equation:sum kronecker pol}
    T
    \otimes\polmul_{m,t} 
    = 
    \sum_{\ell=1}^m
    \sum_{a,b,c=1}^n 
    \
    \sum_{\substack{i\in[m],\, j\in[t]}} (T_\ell)_{a,b,c}\,
    \bigl(e_a\otimes e_i\bigr) \otimes 
    \bigl(e_b^\ell \otimes e_j\bigr) \otimes
    \bigl(e_c^\ell \otimes e_{i+j-1}\bigr).
\end{align}
We will apply a monomial degeneration, as in \cref{lem:monomial-degen}, to this tensor. Define the weights
\[
w_1(e_a\otimes e_i)=i^2,\qquad w_2(e_b^\ell\otimes e_j)=\ell^2+2\ell(j-1),\qquad w_3(e_c^\ell\otimes e_k)=-2\ell k.
\]
On the term with indices $\ell$, $i$ and~$j$, the third-leg index is $k=i+j-1$, and the sum of the weights is $i^2+\ell^2+2\ell(j-1)-2\ell(i+j-1)=(i-\ell)^2$. This is nonnegative, and zero exactly when $i=\ell$. So $T\otimes\polmul_{m,t}$ degenerates to
\[
\sum_{\ell=1}^m\sum_{a,b,c}\ \sum_{j\in[t]}(T_\ell)_{a,b,c}\,(e_a\otimes e_\ell)\otimes(e_b^\ell\otimes e_j)\otimes(e_c^\ell\otimes e_{\ell+j-1}).
\]
The $\ell$th summand lies in $(U\otimes e_\ell)\otimes(V_\ell\otimes\FF^t)\otimes(W_\ell\otimes\FF^{m+t-1})$, so the summands no longer share the first leg. Mapping $e_c^\ell\otimes e_{\ell+j-1}$ to $e_c^\ell\otimes e_j$, and all other basis vectors to~$0$, on the third leg restricts this tensor to $(T_1\otimes\polmul_{1,t})\oplus\cdots\oplus(T_m\otimes\polmul_{1,t})=(T_1\oplus\cdots\oplus T_m)\otimes\polmul_{1,t}$.
\end{proof}

\subsection{Concavity degeneration for $\polmul_{a,b}$}\label{sec:det}

In this subsection we will prove a degeneration relation among polynomial multiplication tensors $\polmul_{a,b}$, which can be interpreted as a discrete concavity in $b$. 
We will first prove it using \cref{lem:restriction-quotient}.
Afterwards, we will give a second proof over~$\CC$, using the Clebsch--Gordan rule for representations of the special linear group~$\SL_2(\CC)$. %

\begin{lemma}[Concavity degeneration]\label{lem:det-degen}
Let $a\geq1$ and $b\geq2$ be integers. Then
\[
\polmul_{a,b+1}\oplusX \polmul_{a,b-1}\degenTT\polmul_{a,b}\oplusX \polmul_{a,b}.
\]
\end{lemma}

\begin{proof}
We will apply \cref{lem:restriction-quotient} to the bilinear map that sends $(f,(g,h))$ to~$(fg,fh)$, whose tensor is $\polmul_{a,b}\oplusX\polmul_{a,b}$, and to subspaces for which the induced maps on the subspaces and on the quotients have tensors~$\polmul_{a,b-1}$ and~$\polmul_{a,b+1}$.
For every integer $d\geq0$, let $\mathcal{W}_d\coloneqq\FF[t]_{\leq d}$. For every integer $d\geq2$, let
\[
\iota_d\colon\mathcal{W}_{d-2}\to\mathcal{W}_{d-1}\oplus\mathcal{W}_{d-1}\colon p\mapsto(-tp,p),\qquad
\sigma_d\colon\mathcal{W}_{d-1}\oplus\mathcal{W}_{d-1}\to\mathcal{W}_d\colon(g,h)\mapsto g+th.
\]
Then~$\iota_d$ is injective and $\sigma_d$ is surjective. For $g,h\in\mathcal{W}_{d-1}$, we have $g+th=0$ exactly when $g=-th$ and $h\in\mathcal{W}_{d-2}$, so the kernel of~$\sigma_d$ is the image of~$\iota_d$. Let $V\coloneqq\mathcal{W}_{b-1}\oplus\mathcal{W}_{b-1}$ and $W\coloneqq\mathcal{W}_{a+b-2}\oplus\mathcal{W}_{a+b-2}$, and consider the bilinear map
\[
\mu\colon\mathcal{W}_{a-1}\times V\to W\colon\bigl(f,(g,h)\bigr)\mapsto(fg,fh),
\]
whose tensor~$T$ is $\polmul_{a,b}\oplusX\polmul_{a,b}$. %
Let~$V_1$ be the image of~$\iota_b$ and~$W_1$ the image of~$\iota_{a+b-1}$. We claim \[
\mu(\mathcal{W}_{a-1}\times V_1)\subseteq W_1.
\]
This claim will let us apply \cref{lem:restriction-quotient}.
Indeed, for all $f\in\mathcal{W}_{a-1}$ and $p\in\mathcal{W}_{b-2}$,
\begin{align}\label{eq:concavity any field mu1}
\mu\bigl(f,\iota_b(p)\bigr)=(-tfp,fp)=\iota_{a+b-1}(fp).
\end{align}
By \cref{lem:restriction-quotient}, $T\degengeqTT T_1\oplusX T_2$ for the tensors $T_1$ and~$T_2$ of the bilinear maps
\[
    \mu_1 \colon \mathcal{W}_{a-1} \times ( V/V_1) \to (W/W_1)
    \quad
    \text{and}
    \quad
    \mu_2 \colon \mathcal{W}_{a-1} \times V_1 \to W_1
\]
induced by $\mu$.
It remains to show that %
$T_1 \cong \polmul_{a,b+1}$ and $T_2 \cong \polmul_{a,b-1}$.

For $T_2$ this follows from \eqref{eq:concavity any field mu1}, where we identify $\mathcal{W}_{b-2} \cong V_1$ using the isomorphism $\iota_b\colon\mathcal{W}_{b-2}\to V_1$, and similarly $\mathcal{W}_{a+b-3} \cong W_1$.

For $T_1$ we use that 
the maps~$\sigma_b$ and~$\sigma_{a+b-1}$ are surjective with kernels~$V_1$ and~$W_1$, so they induce isomorphisms from the quotients to~$\mathcal{W}_b$ and~$\mathcal{W}_{a+b-1}$. 
Since 
\[
\sigma_{a+b-1}\bigl(\mu(f,(g,h))\bigr)=fg+tfh=f\sigma_b(g,h)
\]
for all $f\in\mathcal{W}_{a-1}$ and $g,h\in\mathcal{W}_{b-1}$, these isomorphisms identify the induced map~$\mu_1$ on the quotients with the multiplication map $\mathcal{W}_{a-1}\times\mathcal{W}_b\to\mathcal{W}_{a+b-1}$, whose tensor is~$\polmul_{a,b+1}$. 
\end{proof}

We now give the second proof, over~$\CC$, which is insightful for readers familiar with representation theory. 
For $d \geq 1$ let $\mathcal{V}_d \coloneqq \Sym^d \CC^2 \subseteq (\CC^2)^{\otimes d}$ be the subspace of symmetric tensors of order $d$ and let $\mathcal{V}_0 \coloneqq \CC$. 
Then $\mathcal{V}_d$ is an irreducible $\SL_2(\CC)$ representation under the diagonal action $g \mapsto g^{\otimes d}$ of $\SL_2(\CC)$, and its dual $\mathcal{V}_d^*$ is isomorphic to the space $\CC[x,y]_d$ of homogeneous polynomials of degree $d$ as an $\SL_2(\CC)$ representation (see Fulton and Harris~\cite[Chapter~11]{MR1153249}).
We can identify $\CC[t]_{\leq d}$ with $\CC[x,y]_{d}$ via homogenization $t^k \mapsto x^k y^{d-k}$.
Under this identification, the polynomial multiplication tensor $\polmul_{a,b}$ by definition describes the (bilinear) multiplication map $\CC[x,y]_{a-1} \times \CC[x,y]_{b-1} \to \CC[x,y]_{a+b-2}$. This map is $\SL_2(\CC)$ equivariant. Under the standard identification of bilinear maps with tensors we have that $\polmul_{a,b} \in \mathcal{V}_{a-1} \otimes \mathcal{V}_{b-1} \otimes \mathcal{V}_{a+b-2}^*$.
Because this identification is $\SL_2(\CC)$ equivariant, we obtain that $\polmul_{a,b}$ is invariant under the diagonal $\SL_2(\CC)$ action, that is, $\polmul_{a,b} \in (\mathcal{V}_{a-1} \otimes \mathcal{V}_{b-1} \otimes \mathcal{V}_{a+b-2}^*)^{\SL_2(\CC)}$. 

\begin{lemma}\label{lem:pol mult rep theory}
Let $i,j\geq0$ be integers. We have that $\mathcal{V}_i \otimes \mathcal{V}_j \cong
    \bigoplus_{k=0}^{\min(i,j)}
    \mathcal{V}_{i+j-2k}$.
In particular, for every integer $\ell>i+j$,
$(\mathcal{V}_i \otimes \mathcal{V}_j \otimes \mathcal{V}_\ell^*)^{\SL_2(\CC)} = 0$
and
$(\mathcal{V}_i \otimes \mathcal{V}_j \otimes \mathcal{V}_{i+j}^*)^{\SL_2(\CC)} = \CC \polmul_{i+1,j+1}.$
\end{lemma}
\begin{proof}
    The first statement is standard, see e.g.\ \cite[Exercise~11.11]{MR1153249}. The two equalities follow from the standard identification of $\SL_2(\CC)$-equivariant linear maps $\mathcal{V}_\ell \to \mathcal{V}_i \otimes \mathcal{V}_j$ with 
    $(\mathcal{V}_i \otimes \mathcal{V}_j \otimes \mathcal{V}_\ell^*)^{\SL_2(\CC)}$:
    by the first statement, $\mathcal{V}_\ell$ with $\ell>i+j$ does not occur in $\mathcal{V}_i\otimes\mathcal{V}_j$ and $\mathcal{V}_{i+j}$ occurs once, so Schur's lemma tells us the two spaces have dimension 0 and 1 respectively; the second equality follows because $\polmul_{i+1,j+1} \in (\mathcal{V}_i \otimes \mathcal{V}_j \otimes \mathcal{V}_{i+j}^*)^{\SL_2(\CC)}$.
\end{proof}

\begin{proof}[Proof of \cref{lem:det-degen} over~$\CC$]
By \cref{lem:two-copies}, we have
\begin{align*}
    \polmul_{a,b} \oplusX \polmul_{a,b} 
    \cong \polmul_{a,b} \otimes \polmul_{1,2}
    &\in 
    \bigl( ( \mathcal{V}_{a-1} \otimes \mathcal{V}_{0} ) \otimes
    ( \mathcal{V}_{b-1} \otimes \mathcal{V}_{1} ) \otimes
    ( \mathcal{V}_{a+b-2}^* \otimes \mathcal{V}_{1}^* ) \bigr)^{\SL_2(\CC)}.
\end{align*}
Since $\mathcal{V}_0 = \CC$, we have $\mathcal{V}_{a-1} \otimes \mathcal{V}_{0} = \mathcal{V}_{a-1}$.
By \cref{lem:pol mult rep theory} we know that $\mathcal{V}_{d} \otimes \mathcal{V}_{1} \cong \mathcal{V}_{d+1} \oplus \mathcal{V}_{d-1}$ as $\SL_2(\CC)$ representations, for every integer $d\geq1$.
Thus, there exist $\SL_2(\CC)$-equivariant isomorphisms $B:\mathcal{V}_{b-1}\otimes\mathcal{V}_1\to\mathcal{V}_b\oplus\mathcal{V}_{b-2}$ and $C:\mathcal{V}_{a+b-2}^*\otimes\mathcal{V}_1^*\to\mathcal{V}_{a+b-1}^*\oplus\mathcal{V}_{a+b-3}^*$ such that 
\begin{align}
    \label{equation:polmul repr components}
    T \coloneqq (I \otimes B \otimes C)
    (\polmul_{a,b} \otimes \polmul_{1,2})
    &\in
    \bigl( \mathcal{V}_{a-1} \otimes
    ( \mathcal{V}_{b} \oplus \mathcal{V}_{b-2} ) \otimes
    ( \mathcal{V}_{a+b-1}^* \oplus \mathcal{V}_{a+b-3}^* )
    \bigr)^{\SL_2(\CC)}
\end{align}
So $T = T_{1,1} + T_{1,2} + T_{2,1} + T_{2,2}$ for some tensors
\begin{align*}
\renewcommand{\arraystretch}{1.2} 
\begin{array}{ll}
    T_{1,1} \in \mathcal{V}_{a-1} \otimes \mathcal{V}_{b} \otimes \mathcal{V}_{a+b-1}^*,
    \qquad &
    T_{1,2} \in \mathcal{V}_{a-1} \otimes \mathcal{V}_{b} \otimes \mathcal{V}_{a+b-3}^*,
    \\
    T_{2,1} \in \mathcal{V}_{a-1} \otimes \mathcal{V}_{b-2} \otimes \mathcal{V}_{a+b-1}^*,
    &
    T_{2,2} \in \mathcal{V}_{a-1} \otimes \mathcal{V}_{b-2} \otimes \mathcal{V}_{a+b-3}^*.
\end{array}
\end{align*}
(Here $\mathcal{V}_{b}$ is considered as the subspace $\mathcal{V}_{b} \oplus 0\subseteq\mathcal{V}_{b} \oplus \mathcal{V}_{b-2}$, similarly for the other factors.)
Because we can obtain $T_{i,j}$ from $T$ using $\SL_2(\CC)$ equivariant projections along the direct sums, each $T_{i,j}$ is $\SL_2(\CC)$ invariant.
By \cref{lem:pol mult rep theory}, $T_{1,1} = \alpha_1 \polmul_{a,b+1}$ and $T_{2,2} = \alpha_2 \polmul_{a,b-1}$ for some $\alpha_1,\alpha_2 \in \CC$, and $T_{2,1} = 0$. 
Since~$B$ and~$C$ are isomorphisms, $T$, like $\polmul_{a,b}\otimes\polmul_{1,2}$, is \emph{concise}: no leg can be replaced by a proper subspace containing it. So $\alpha_1\neq0$ and $\alpha_2\neq0$, since otherwise~$T$ would lie in $\mathcal{V}_{a-1}\otimes(\mathcal{V}_b\oplus\mathcal{V}_{b-2})\otimes\mathcal{V}_{a+b-3}^*$ or in $\mathcal{V}_{a-1}\otimes\mathcal{V}_b\otimes(\mathcal{V}_{a+b-1}^*\oplus\mathcal{V}_{a+b-3}^*)$.
Composing~$B$ with $\alpha_1^{-1}I_{\mathcal{V}_b}\oplus\alpha_2^{-1}I_{\mathcal{V}_{b-2}}$ makes both diagonal coefficients equal to~$1$.
Then the result follows from \cref{lem:triangular degeneration}.
\end{proof}

\begin{remark}
Over a field of positive characteristic, the Clebsch--Gordan decomposition can fail.
For example, one checks that over a field~$\FF$ of characteristic~$2$, the representation $\FF^2\otimes\FF^2$ of $\SL_2(\FF)$ is not isomorphic to $\Sym^2\FF^2\oplus\FF$. 
This is why the proof over every field uses a subspace and a quotient instead of this direct sum decomposition.
\end{remark}

\begin{remark}
After homogenisation, $\sigma_d$ is the multiplication map $\FF[x,y]_{d-1}\otimes\FF[x,y]_1\to\FF[x,y]_d$ of binary forms, which is equivariant for the action of~$\SL_2(\FF)$ by linear substitutions in~$x$ and~$y$. Over~$\CC$, where $\CC[x,y]_d\cong\mathcal{V}_d$, it is, up to a nonzero scalar, the component $\mathcal{V}_{d-1}\otimes\mathcal{V}_1\to\mathcal{V}_d$ of the $\SL_2(\CC)$-equivariant isomorphism $\mathcal{V}_{d-1}\otimes\mathcal{V}_1\cong\mathcal{V}_d\oplus\mathcal{V}_{d-2}$ that the proof over~$\CC$ uses for $d=b$.
\end{remark}

\begin{remark}\label{rem:det-general}
\Cref{lem:det-degen} generalises as follows. Let $a,b,c\geq1$ and $d\geq0$ be integers with $b\geq c$ and $d<c$. Then $\polmul_{a,b+d}\oplusX\polmul_{a,c-d}\degenTT\polmul_{a,b}\oplusX\polmul_{a,c}$; \cref{lem:det-degen} is the case $b=c$ and $d=1$.
The proof of \cref{lem:det-degen} over every field gives this, with~$t^{b-c+d}$ in place of~$t$ in~$\iota_d$ and~$\sigma_d$.
For an integer $r\geq1$ and a tuple $\lambda=(\lambda_1,\ldots,\lambda_r)$ of positive integers, let $T_\lambda=\polmul_{a,\lambda_1}\oplusX\cdots\oplusX\polmul_{a,\lambda_r}$. Let~$\lambda$ and~$\nu$ be tuples of~$r$ positive integers with the same sum. If~$\lambda$ majorises~$\nu$, that is, if for every $k\in[r]$ the sum of the~$k$ largest entries of~$\lambda$ is at least that of~$\nu$, then $T_\lambda\degenTT T_\nu$, by repeated use of the case $d=1$ on two of the summands.
\end{remark}

\subsection{Tripling degeneration for $\polmul_{a,b}$}\label{sec:tripling}

We will construct a degeneration of a polynomial multiplication tensor to a partial direct sum of three smaller ones, the middle one with its second and third legs exchanged. 
The trivial starting point is this: for all integers $a,b\geq1$,
\begin{equation}\label{eq:two-copies}
    \polmul_{a,b}\oplusX\polmul_{a,b}\leq\polmul_{a,2b+a-1}.
\end{equation}
In terms of polynomials, the idea is as follows. The left-hand side describes the simultaneous multiplication $(f, (g_1,g_2)) \mapsto (fg_1,fg_2)$ of polynomials $f \in \FF[t]_{\leq a-1}$ and $g_1,g_2\in\FF[t]_{\leq b-1}$.
The restriction then follows the embedding of $(g_1,g_2) \mapsto g \coloneqq g_1+t^{a+b-1}g_2 \in \FF[t]_{\leq 2b+a-2}$, and the subsequent observation that the coefficients of $fg_1$ and $fg_2$ appear as the coefficients of $fg$.

However, the coefficients of~$g$ at $t^b,\ldots,t^{a+b-2}$ are not used. 
The tripling degeneration uses these coefficients, and~$b$ further coefficients, for a third summand in the middle. 

\begin{lemma}[Tripling degeneration]\label{lem:tripling-degen}
Let $a,b\geq1$ be integers. Then
\[
\polmul_{a,b}\oplusX \bigl((23)\cdot\polmul_{a,b}\bigr)\oplusX\polmul_{a,b}\degenTT\polmul_{a,3b+a-1}.
\]
\end{lemma}
The degeneration can be seen quite directly from an example. 
For a tensor $T\in\FF^{n_1}\otimes\FF^{n_2}\otimes\FF^{n_3}$ and $i\in[n_1]$, the \emph{$i$th slice} of~$T$ is the $n_2\times n_3$ matrix whose entry in row~$j$ and column~$k$ is the coefficient of $e_i\otimes e_j\otimes e_k$ in~$T$.
Take $a=3$ and $b = 3$. Then the slices of $\polmul_{a,3b+a-1}$ look as follows, where we have split each slice into 9 blocks,
\begin{align*}
\scalebox{0.6}{$
\left[
\begin{array}{ccccc|ccc|ccccc}
1 & 0 & 0 & 0 & 0 &  &  &  &  &  &  &  &  \\
0 & 1 & 0 & 0 & 0 &  & \smash{\text{\Large 0}} &  &  &  & \smash{\text{\Large 0}} &  &  \\
0 & 0 & 1 & 0 & 0 &  &  &  &  &  &  &  &  \\\hline
0 & 0 & 0 & 1 & 0 & 0 & 0 & 0 & 0 & 0 & 0 & 0 & 0 \\
0 & 0 & 0 & 0 & 1 & 0 & 0 & 0 & 0 & 0 & 0 & 0 & 0 \\
0 & 0 & 0 & 0 & 0 & 1 & 0 & 0 & 0 & 0 & 0 & 0 & 0 \\
0 & 0 & 0 & 0 & 0 & 0 & 1 & 0 & 0 & 0 & 0 & 0 & 0 \\
0 & 0 & 0 & 0 & 0 & 0 & 0 & 1 & 0 & 0 & 0 & 0 & 0 \\\hline
 &  &  &  &  &  &  &  & 1 & 0 & 0 & 0 & 0 \\
 &  & \smash{\text{\Large 0}} &  &  &  & \smash{\text{\Large 0}} &  & 0 & 1 & 0 & 0 & 0 \\
 &  &  &  &  &  &  &  & 0 & 0 & 1 & 0 & 0 \\
\end{array}
\right]$}\quad
\scalebox{0.6}{$
\left[
\begin{array}{ccccc|ccc|ccccc}
0 & 1 & 0 & 0 & 0 &  &  &  &  &  &  &  &  \\
0 & 0 & 1 & 0 & 0 &  & \smash{\text{\Large 0}} &  &  &  & \smash{\text{\Large 0}} &  &  \\
0 & 0 & 0 & 1 & 0 &  &  &  &  &  &  &  &  \\\hline
0 & 0 & 0 & 0 & 1 & 0 & 0 & 0 & 0 & 0 & 0 & 0 & 0 \\
0 & 0 & 0 & 0 & 0 & 1 & 0 & 0 & 0 & 0 & 0 & 0 & 0 \\
0 & 0 & 0 & 0 & 0 & 0 & 1 & 0 & 0 & 0 & 0 & 0 & 0 \\
0 & 0 & 0 & 0 & 0 & 0 & 0 & 1 & 0 & 0 & 0 & 0 & 0 \\
0 & 0 & 0 & 0 & 0 & 0 & 0 & 0 & 1 & 0 & 0 & 0 & 0 \\\hline
 &  &  &  &  &  &  &  & 0 & 1 & 0 & 0 & 0 \\
 &  & \smash{\text{\Large 0}} &  &  &  & \smash{\text{\Large 0}} &  & 0 & 0 & 1 & 0 & 0 \\
 &  &  &  &  &  &  &  & 0 & 0 & 0 & 1 & 0 \\
\end{array}
\right]$}\quad
\scalebox{0.6}{$
\left[
\begin{array}{ccccc|ccc|ccccc}
0 & 0 & 1 & 0 & 0 &  &  &  &  &  &  &  &  \\
0 & 0 & 0 & 1 & 0 &  & \smash{\text{\Large 0}} &  &  &  & \smash{\text{\Large 0}} &  &  \\
0 & 0 & 0 & 0 & 1 &  &  &  &  &  &  &  &  \\\hline
0 & 0 & 0 & 0 & 0 & 1 & 0 & 0 & 0 & 0 & 0 & 0 & 0 \\
0 & 0 & 0 & 0 & 0 & 0 & 1 & 0 & 0 & 0 & 0 & 0 & 0 \\
0 & 0 & 0 & 0 & 0 & 0 & 0 & 1 & 0 & 0 & 0 & 0 & 0 \\
0 & 0 & 0 & 0 & 0 & 0 & 0 & 0 & 1 & 0 & 0 & 0 & 0 \\
0 & 0 & 0 & 0 & 0 & 0 & 0 & 0 & 0 & 1 & 0 & 0 & 0 \\\hline
 &  &  &  &  &  &  &  & 0 & 0 & 1 & 0 & 0 \\
 &  & \smash{\text{\Large 0}} &  &  &  & \smash{\text{\Large 0}} &  & 0 & 0 & 0 & 1 & 0 \\
 &  &  &  &  &  &  &  & 0 & 0 & 0 & 0 & 1 \\
\end{array}
\right]$}.
\end{align*}
The top-left blocks of the three slices are
\begin{equation}\label{eq:tripling-outer}
\scalebox{0.8}{$\begin{bmatrix}1&0&0&0&0\\0&1&0&0&0\\0&0&1&0&0\end{bmatrix}\quad\begin{bmatrix}0&1&0&0&0\\0&0&1&0&0\\0&0&0&1&0\end{bmatrix}\quad\begin{bmatrix}0&0&1&0&0\\0&0&0&1&0\\0&0&0&0&1\end{bmatrix}$}.
\end{equation}
These are the slices of~$\polmul_{3,3}$. The bottom-right blocks are the same. The middle blocks are
\begin{equation}\label{eq:tripling-middle}
\scalebox{0.8}{$\begin{bmatrix}0&0&0\\0&0&0\\1&0&0\\0&1&0\\0&0&1\end{bmatrix}\quad\begin{bmatrix}0&0&0\\1&0&0\\0&1&0\\0&0&1\\0&0&0\end{bmatrix}\quad\begin{bmatrix}1&0&0\\0&1&0\\0&0&1\\0&0&0\\0&0&0\end{bmatrix}$}.
\end{equation}
Rotated by 180 degrees, that is, with the order of both the rows and the columns reversed, the matrices in~\eqref{eq:tripling-middle} are the transposes of those in~\eqref{eq:tripling-outer}. These are the slices of $(23)\cdot\polmul_{3,3}$, the tensor~$\polmul_{3,3}$ with its second and third legs exchanged. The other two blocks of the middle block row are not zero. So the three slices together have the block structure
\[
\polmul_{3,11}=\begin{pmatrix}\polmul_{3,3}&0&0\\\ast&(23)\cdot\polmul_{3,3}&\ast\\0&0&\polmul_{3,3}\end{pmatrix},
\]
where each entry stands for the tensor whose slices are the corresponding blocks, the middle one read in reverse order, and~$\ast$ marks the two blocks that are not zero. By \cref{lem:three-block},
\[
\polmul_{3,11}\degengeqTT\begin{pmatrix}\polmul_{3,3}&0&0\\0&(23)\cdot\polmul_{3,3}&0\\0&0&\polmul_{3,3}\end{pmatrix}\cong\polmul_{3,3}\oplusX((23)\cdot\polmul_{3,3})\oplusX\polmul_{3,3}.
\]
\begin{proof}
For every $i\in[a]$, the $i$th slice of $\polmul_{a,3b+a-1}$ has a~$1$ in row~$j$ and column~$i+j-1$ for every $j\in[3b+a-1]$, and zeros elsewhere. As in the example, split the rows and the columns of the slices into three intervals each, where $[p,q]$ denotes the set of integers $p,p+1,\ldots,q$:
\[
\begin{array}{r|ccc}
 \text{rows} & [1,b] & [b+1,2b+a-1] & [2b+a,3b+a-1]\\
 \text{columns} & [1,b+a-1] & [b+a,2b+a-1] & [2b+a,3b+2a-2]
\end{array}
\]
This splits every slice into nine blocks, in three block rows and three block columns. 
For all $i \in [a]$,
if $j\leq b$ then $i+j-1\leq b+a-1$, and if $j\geq2b+a$ then $i+j-1\geq2b+a$. So in the top and the bottom block row of the $i$th slice, only the diagonal block is nonzero. Numbered from its top-left corner, it has a~$1$ in row~$r$ and column~$i+r-1$ for every $r\in[b]$, and zeros elsewhere: it is the $i$th slice of~$\polmul_{a,b}$. In the middle block column, column~$k$ has its~$1$ in row $k-i+1$, which is a middle row. Number the rows and the columns of the middle block in reverse order, from its bottom-right corner. Then column~$k$ becomes column $c=2b+a-k$, and row $k-i+1$ becomes row $2b+a-(k-i+1)=i+c-1$. So the middle block has a~$1$ in row~$i+c-1$ and column~$c$ for every $c\in[b]$, and zeros elsewhere: it is the $i$th slice of $(23)\cdot\polmul_{a,b}$. In block form, where each block stands for the tensor whose slices are the corresponding blocks of the slices, we have shown
\[
\polmul_{a,3b+a-1}=\begin{pmatrix}\polmul_{a,b}&0&0\\\ast&(23)\cdot\polmul_{a,b}&\ast\\0&0&\polmul_{a,b}\end{pmatrix},
\]
where~$\ast$ marks the two blocks that need not be zero. An application of \cref{lem:three-block} makes the $\ast$~blocks zero and proves the lemma.
\end{proof}

\begin{remark}\label{rem:tripling-general}
The proof of \cref{lem:tripling-degen} naturally extends to multiple blocks of different sizes.
Let $s,r,n\geq1$ and $a,b_i,c_j\geq1$ be integers. Then
\[
\polmul_{a,b_1}\oplusX\cdots\oplusX\polmul_{a,b_r}\oplusX\bigl((23)\cdot\polmul_{a,c_1}\bigr)\oplusX\cdots\oplusX\bigl((23)\cdot\polmul_{a,c_s}\bigr)
\degenTT \polmul_{a,n}
\]
if and only if $n \geq N \coloneqq \sum_i b_i + \sum_j c_j + (a-1)\max(s,r-1)$.
The proof interleaves the tensors $\polmul_{a,b_i}$ and $(23) \cdot \polmul_{a,c_j}$ as the diagonal blocks of $\polmul_{a,n}$, and degenerates the off-diagonal blocks using \cref{lem:monomial-degen}, as in the proof of \cref{lem:tripling-degen}. In block form, this looks as follows:
\begin{align*}
\scalebox{0.9}{$
\arraycolsep=1.4pt
    \left(
    \begin{array}{ccccccccc}
        \polmul_{a,b_1} &   &   \\
        * & (23) \cdot \polmul_{a,c_1} & * &   \\
          &   & \polmul_{a,b_2} &   &   \\
          &   & * & (23)\cdot \polmul_{a,c_2} & * &   \\
        & & &  &  &  &   \\
        & & &  & \ddots & \phantom{\polmul_{a}} & * &   \\
        & & & & &  & \polmul_{a,b_{k-1}} &   &   \\
        & & & & &  & * & (23)\cdot \polmul_{a,c_{k-1}} & * \\
        & & & & &  &  &   &   \polmul_{a,b_{k}} \\
    \end{array}\right)$},
\end{align*}
where $k \coloneqq \max\{r,s+1\}$, $(b_1,b_2,\ldots)$ and $(c_1,c_2,\ldots)$ are extended with zeros, and~$\polmul_{a,0}$ is the zero tensor in $\FF^a\otimes\FF^0\otimes\FF^{a-1}$. 
If $n<N$, there is no such degeneration: the second and third flattening ranks are $\sum_ib_i+\sum_jc_j+s(a-1)$ and $\sum_ib_i+\sum_jc_j+r(a-1)$ on the left, and~$n$ and $n+a-1$ on the right.
\end{remark}

\begin{remark}[Tripling from a short exact sequence]\label{rem:exact-sequences}
\Cref{lem:tripling-degen} also follows from \cref{lem:restriction-quotient}, applied to the multiplication map $\FF[t]_{\leq a-1}\times\FF[t]_{\leq 3b+a-2}\to\FF[t]_{\leq 3b+2a-3}\colon(f,g)\mapsto fg$, whose tensor is $\polmul_{a,3b+a-1}$, with~$V_1$ the polynomials $g_1+t^{2b+a-1}g_3$ with $g_1,g_3\in\FF[t]_{\leq b-1}$, and with~$W_1$ the polynomials whose coefficients at $t^{a+b-1},\ldots,t^{a+2b-2}$ vanish. 
The induced map on the subspaces has tensor $\polmul_{a,b}\oplusX\polmul_{a,b}$ (in the same way that we obtained the trivial restriction \eqref{eq:two-copies} before).
The induced map on the quotients sends $(f,g+V_1)$ to $fg+W_1$. Take the basis $t^{a-1-i}$ of $\FF[t]_{\leq a-1}$, for $0\leq i\leq a-1$; take the basis $t^{b+j}+V_1$ of $\FF[t]_{\leq 3b+a-2}/V_1$, for $0\leq j\leq a+b-2$; and take the basis $t^{a+b-1+k}+W_1$ of  $\FF[t]_{\leq 3b+2a-3}/ W_1$, for $0\leq k\leq b-1$. The induced map sends the $i$th and the $j$th basis vector to the $k$th one if $j=i+k$, and to~$0$ if $j-i$ is not in $\{0,\ldots,b-1\}$. So it has tensor $(23)\cdot\polmul_{a,b}$.
\end{remark}

\section{Upper bounds via spectral points}\label{sec:spectral}
In this section we will prove \cref{thm:main} and \cref{thm:rect} using Strassen's asymptotic spectrum. Strassen's duality theorems allow us to prove these bounds by establishing inequalities for the values of each spectral point at matrix multiplication tensors. We will obtain these inequalities from the degenerations of \cref{sec:constructions}.

\subsection{Spectral points, asymptotic spectrum duality, and basic properties}\label{sec:spectral-points}
In this subsection we will define Strassen's asymptotic spectrum, whose elements are called spectral points, and state the properties that we will use, including the asymptotic spectrum duality theorem. Strassen introduced these to study the asymptotic restriction of tensors, and in particular the exponent of matrix multiplication~\cite{strassen1988asymptotic}.  For background, see the survey of Wigderson and Zuiddam~\cite{wigderson2022asymptotic}.

\begin{definition}[{Strassen~\cite[Section~3]{strassen1988asymptotic}}]\label{def:spectral-point}
A \emph{spectral point} is a map $\phi:\mathcal{S}_\FF\to\RR_{\geq0}$ such that $\phi(\langle1\rangle)=1$ and, for all tensors~$S$ and~$T$ over~$\FF$, we have $\phi(S\oplus T)=\phi(S)+\phi(T)$ and $\phi(S\otimes T)=\phi(S)\,\phi(T)$, and if $T\leq S$, then $\phi(T)\leq\phi(S)$. The \emph{asymptotic spectrum} of tensors over~$\FF$ is the set~$\Delta_\FF$ of spectral points. 
\end{definition}

For the upper bounds on~$\omega_\FF$ and~$\omega_\FF(a,b,c)$, we will not need to know what the spectral points are beyond their abstract definition! Spectral points have been constructed explicitly, but it is not known whether these constructions give all of the asymptotic spectrum~\cite{strassen1991degeneration,MR4495838,sakabe2026strassenssupportfunctionalscoincide}.

The simplest spectral points are the flattening ranks~$\flatrank{1}$, $\flatrank{2}$ and~$\flatrank{3}$ of \cref{sec:tensors-rank}, since the rank of matrices is normalised, additive, multiplicative and monotone under restriction. They are the only explicit spectral points that we will use, in examples and to show that some of the bounds are attained.

From the definition it readily follows that spectral points lower bound the matrix multiplication exponent: for every $\phi \in \Delta_\FF$ we have $\log \phi(\langle 2,2,2\rangle) \leq \omega_\FF$. Strassen proved that these spectral point lower bounds are complete:

\begin{theorem}[{Strassen~\cite[Theorem~3.8 and Proposition~4.2]{strassen1988asymptotic}}]\label{thm:spectral}
We have
\[
\max_{\phi\in\Delta_\FF}\log\phi(\langle2,2,2\rangle) = \omega_\FF.
\]
\end{theorem}
To upper bound $\omega_\FF$ it thus suffices to upper bound $\phi(\langle 2,2,2\rangle)$ for every $\phi \in \Delta_\FF$. This we will do by taking the degenerations constructed in \cref{sec:constructions}, applying the spectral points to them, and optimising over these inequalities to upper bound the value of $\phi(\langle 2,2,2\rangle)$. In \cref{sec:optimality} we will show that no combination of these degenerations by the combination rules gives a better bound.
The values of a spectral point~$\phi$ at the tensors $\langle1,1,2\rangle$, $\langle2,1,1\rangle$ and~$\langle1,2,1\rangle$ will play a central role. The following are the basic properties of spectral points that we will need.
For a given $\pi\in\mathfrak{S}_3$, we define the permutation $\phi_\pi$ of $\phi$ as
\begin{align}
    \phi_\pi(T)\coloneqq\phi(\pi\cdot T).
\end{align}

\begin{lemma}[{Strassen~\cite[Proposition~5.10]{MR882307}, \cite[Theorem~3.1 and Proposition~4.3]{strassen1988asymptotic}}]\label{lem:spectral-facts}
Let $\phi\in\Delta_\FF$.
\begin{enumerate}[\upshape(i)]
\item\label{sf:monotone} For all tensors~$S$ and~$T$ over~$\FF$ with $T\degen S$, we have $\phi(T)\leq\phi(S)$.
\item\label{sf:unit} For every integer $n\geq0$, we have $\phi(\langle n\rangle)=n$.
\item\label{sf:nonzero} For every nonzero tensor~$T$ over~$\FF$, we have $\phi(T)\geq1$.
\item\label{sf:power} The values $\phi(\langle1,1,2\rangle)$, $\phi(\langle2,1,1\rangle)$ and~$\phi(\langle1,2,1\rangle)$ lie in the interval~$[1,2]$, and for all integers $n_1,n_2,n_3\geq1$,
\[
\phi(\langle n_1,n_2,n_3\rangle)=\phi(\langle2,1,1\rangle)^{\log n_1}\,\phi(\langle1,2,1\rangle)^{\log n_2}\,\phi(\langle1,1,2\rangle)^{\log n_3}.
\]
In particular, $\phi(\langle m,m,m\rangle)=\phi(\langle2,2,2\rangle)^{\log m}$ for every integer $m\geq1$.
\item \label{sf:permutation}
$\phi_\pi \in \Delta_\FF$.
\end{enumerate}
\end{lemma}
\begin{proof}
Part~(\ref{sf:monotone}) is~\cite[Proposition~5.10]{MR882307} and~\cite[Theorem~3.1]{strassen1988asymptotic}. Part~(\ref{sf:unit}) holds since~$\langle n\rangle$ is the direct sum of~$n$ copies of~$\langle1\rangle$, and part~(\ref{sf:nonzero}) follows since $\langle1\rangle\leq T$ for every nonzero tensor~$T$ (\cref{lem:restriction-facts}(\ref{rf:unit})). The formula in part~(\ref{sf:power}) is~\cite[Proposition~4.3]{strassen1988asymptotic}.
The bounds on the three values follow from parts (\ref{sf:monotone})--(\ref{sf:nonzero}), since each of the tensors $\langle1,1,2\rangle$, $\langle2,1,1\rangle$ and~$\langle1,2,1\rangle$ is nonzero and a restriction of~$\langle2\rangle$. The last claim is the formula for $n_1=n_2=n_3=m$, together with the formula for $n_1=n_2=n_3=2$.
Part (\ref{sf:permutation}) follows since permuting the legs of a tensor commutes with direct sums and tensor products, fixes~$\langle1\rangle$ and preserves restriction (\cref{lem:restriction-facts}(\ref{rf:permutation})). 
\end{proof}

Strassen also expressed the exponents of rectangular matrix multiplication as a maximum over spectral points.
\begin{lemma}[{Strassen~\cite[Proposition~4.5]{strassen1988asymptotic}}]\label{lem:spectral-omega}
Let $a,b,c\geq0$ be real numbers. Then
\[
2^{\omega_\FF(a,b,c)}=\max_{\phi\in\Delta_\FF}\phi(\langle2,1,1\rangle)^a\,\phi(\langle1,2,1\rangle)^b\,\phi(\langle1,1,2\rangle)^c.
\]
\end{lemma}
\begin{proof}
For $a,b,c>0$, this follows directly from~\cite[Propositions~4.5 and~3.3]{strassen1988asymptotic}.

If one of $a$, $b$ and~$c$ is~$0$, say $c=0$, then both sides equal~$2^{a+b}$, since the tensor rank of $\langle n_1,n_2,1\rangle$ is~$n_1n_2$, and $\flatrank{1} \in \Delta_\FF$ satisfies
$\flatrank{1}(\langle2,1,1\rangle) = \flatrank{1}(\langle1,2,1\rangle) = 2$.
\end{proof}

\subsection{Superadditivity on partial direct sums}\label{sec:superadd}
In this subsection we will show that, for every $\phi\in\Delta_\FF$ with $p\coloneqq\log\phi(\langle1,1,2\rangle)>0$, the function~$\phi^{1/p}$ is superadditive on partial direct sums (\cref{lem:superadd-general}). We will first treat isomorphic summands using \cref{prop:sep} (\cref{lem:iso-summands}), and then remove the hypothesis that the summands are isomorphic by applying this case to tensor powers. As a first application, we will obtain $\omega_\FF<2.3212$ from a tensor of Strassen (\cref{prop:laser-sep}).

\begin{lemma}[Isomorphic summands]\label{lem:iso-summands}
Let $\phi\in\Delta_\FF$, and suppose that $p\coloneqq\log\phi(\langle1,1,2\rangle)$ is positive. Let $r\geq1$ be an integer, and let $S=S_1\oplusX\cdots\oplusX S_r$ be a partial direct sum in which each~$S_i$ is isomorphic to~$S_1$. Then
\[
\phi(S)\geq\tfrac12\,r^p\,\phi(S_1).
\]
\end{lemma}
\begin{proof}
We have
\[
(S_1\oplus\cdots\oplus S_r)\otimes\polmul_{1,r}\degen S\otimes\polmul_{r,r}\degen S\otimes\langle2r-1\rangle.
\]
The first degeneration is \cref{prop:sep} with $m=t=r$. The second follows by tensoring the border rank bound $\polmul_{r,r}\degen\langle2r-1\rangle$ (\cref{prop:polmul-rank}) with~$S$. The left-hand side is isomorphic to $\langle r\rangle\otimes S_1\otimes\langle1,1,r\rangle$, since $S_1\oplus\cdots\oplus S_r\cong\langle r\rangle\otimes S_1$ and $\polmul_{1,r}\cong\langle1,1,r\rangle$. The spectral point~$\phi$ is multiplicative and monotone under degeneration, $\phi(\langle n\rangle)=n$ for every integer $n\geq1$, and $\phi(\langle1,1,r\rangle)=r^p$ (\cref{lem:spectral-facts}(\ref{sf:monotone}), (\ref{sf:unit}) and~(\ref{sf:power})). So
\[
r^{1+p}\phi(S_1)\leq(2r-1)\,\phi(S).
\]
Since $r/(2r-1)\geq1/2$, it follows that $\phi(S)\geq\frac12r^p\phi(S_1)$.
\end{proof}

\begin{lemma}[Superadditivity]\label{lem:superadd-general}
Let $\phi\in\Delta_\FF$, and suppose that $p\coloneqq\log\phi(\langle1,1,2\rangle)$ is positive.
Let $m\geq1$ be an integer. For $\ell \in [m]$, let $T_\ell \in U \otimes V_\ell \otimes W_\ell$ be tensors over~$\FF$.
Then
\[
\phi(T_1\oplusX\cdots\oplusX T_m)^{1/p}\geq\phi(T_1)^{1/p}+\cdots+\phi(T_m)^{1/p}.
\]
\end{lemma}
\begin{proof}
Let $T=T_1\oplusX\cdots\oplusX T_m$.
Then~$U$ is the first leg of~$T$. Let $V=V_1\oplus\cdots\oplus V_m$ and $W=W_1\oplus\cdots\oplus W_m$ be its second and third legs with the decompositions of the partial direct sum. 
Let $n\geq1$ be an integer and consider $T^{\otimes n}$. 
For a sequence $\alpha=(\alpha_1,\ldots,\alpha_n)\in[m]^n$, let $V_\alpha\coloneqq V_{\alpha_1}\otimes\cdots\otimes V_{\alpha_n}$, $W_\alpha\coloneqq W_{\alpha_1}\otimes\cdots\otimes W_{\alpha_n}$ and $T_\alpha\coloneqq T_{\alpha_1}\otimes\cdots\otimes T_{\alpha_n}$. 
Since the tensor product distributes over direct sums, the second leg $V^{\otimes n}$ of~$T^{\otimes n}$ is the direct sum of the subspaces~$V_\alpha$, and its third leg $W^{\otimes n}$ is the direct sum of the subspaces~$W_\alpha$. 
Expanding $T^{\otimes n}=(T_1+\cdots+T_m)^{\otimes n}$ gives $T^{\otimes n}=\sum_\alpha T_\alpha$. 
Since $T_i\in U\otimes V_i\otimes W_i$ for every~$i$, and the tensor product of tensors is formed leg by leg (\cref{sec:tensors-basic}), $T_\alpha\in U^{\otimes n}\otimes V_\alpha\otimes W_\alpha$ for every~$\alpha$. 
We observe that $T^{\otimes n}$ is the partial direct sum of the tensors $T_\alpha$, $\alpha \in [m]^n$: the tensor product distributes over partial direct sums.

Fix integers $k_1,\ldots,k_m\geq0$ with $k_1+\cdots+k_m=n$ and let~$J$ be the set of sequences $\alpha\in[m]^n$ in which every $i\in[m]$ occurs exactly~$k_i$ times (i.e.~the elements of $J$ have ``type''~$k$). We have
\[
|J|=\frac{n!}{k_1!\cdots k_m!}.
\]
Projecting the second and third legs of~$T^{\otimes n}$ onto $\bigoplus_{\alpha\in J}V_\alpha$ and $\bigoplus_{\alpha\in J}W_\alpha$ gives a tensor $S$ for which $S\leq T^{\otimes n}$. 
Then $S$ is also a partial direct sum, with shared first leg~$U^{\otimes n}$, of the $|J|$~summands~$T_\alpha$ with $\alpha\in J$. 
Each such summand $T_\alpha =  T_{\alpha_1}\otimes\cdots\otimes T_{\alpha_n}$ is isomorphic to $T_1^{\otimes k_1}\otimes\cdots\otimes T_m^{\otimes k_m}$, 
where $T_i^{\otimes0}=\langle1\rangle$. Since $S\leq T^{\otimes n}$, and~$\phi$ is multiplicative and monotone, \cref{lem:iso-summands}, applied to~$S$, gives
\[
\phi(T)^n=\phi(T^{\otimes n})\geq\phi(S)\geq\tfrac12\,|J|^p\,\phi(T_1)^{k_1}\cdots\phi(T_m)^{k_m}.
\]
There are at most $(n+1)^m$ tuples of integers $k_1,\ldots,k_m\geq0$ with $k_1+\cdots+k_m=n$. Raising to the power~$1/p$, summing over these tuples, and using the multinomial theorem gives
\[
2^{1/p}\,(n+1)^m\,\phi(T)^{n/p}\geq\bigl(\phi(T_1)^{1/p}+\cdots+\phi(T_m)^{1/p}\bigr)^n.
\]
Taking~$n$th roots and letting $n\to\infty$ proves the claim.
\end{proof}

\begin{remark}\label{rem:partial-vs-direct}
Partial direct sums and direct sums need not be equivalent under asymptotic restriction. Indeed, the spectral point~$\flatrank{1}$ is~$1$ on $\langle1\rangle\oplusX\langle1\rangle\cong\langle1,1,2\rangle$ (\cref{lem:two-copies}) and~$2$ on $\langle1\rangle\oplus\langle1\rangle=\langle2\rangle$. So~$\langle2\rangle$ is not an asymptotic restriction of~$\langle1,1,2\rangle$. 
\end{remark}

\label{sec:laser}As a first application of \cref{lem:superadd-general}, we obtain an upper bound on~$\omega_\FF$ from a tensor of Strassen. For an integer $q\geq1$, let
\[
\operatorname{Str}_q\coloneqq\sum_{i=1}^q\bigl(e_i\otimes e_0\otimes e_i+e_i\otimes e_i\otimes e_0\bigr)\in\FF^q\otimes\FF^{q+1}\otimes\FF^{q+1},
\]
where the first leg has basis $e_1,\ldots,e_q$ and the other two have basis $e_0,\ldots,e_q$. This is Strassen's tensor~\cite[Lemma~7.4]{MR882307}, with the first and third legs exchanged so that the shared leg comes first; see also Bl\"aser~\cite[Section~8]{blaser2013fast}. The sums with~$e_0$ on the second and on the third leg are isomorphic to $\langle q,1,1\rangle$ and $\langle1,q,1\rangle$, respectively, and $\operatorname{Str}_q = \langle q,1,1 \rangle \oplus_1 \langle 1,q,1\rangle$, since on the second and on the third leg they use the complementary subspaces $\linspan(e_0)$ and $\linspan(e_1,\ldots,e_q)$. 
It is known that $\operatorname{Str}_q\degen\langle q+1\rangle$, which follows from
\[
\sum_{i=1}^qe_i\otimes(e_0+\eps e_i)\otimes(e_0+\eps e_i)-\Bigl(\sum_{i=1}^qe_i\Bigr)\otimes e_0\otimes e_0=\eps\operatorname{Str}_q+\eps^2\sum_{i=1}^qe_i\otimes e_i\otimes e_i,
\]
where the left-hand side is a sum of~$q+1$ tensors of tensor rank one. %

\begin{proposition}\label{prop:laser-sep}
For every integer $q\geq1$, we have $(2q)^{\omega_\FF}\leq(q+1)^3$. For $q=3$ this gives $6^{\omega_\FF}\leq64$, so $\omega_\FF<2.3212$.
\end{proposition}
\begin{proof}
Let $\phi\in\Delta_\FF$, and let $p_1$, $p_2$ and~$p_3$ be the logarithms of~$\phi(\langle1,1,2\rangle)$, $\phi(\langle2,1,1\rangle)$ and~$\phi(\langle1,2,1\rangle)$. Let $\pi\in\mathfrak{S}_3$, and let $(i,j,k)=(\pi(1),\pi(2),\pi(3))$. 
By \cref{lem:spectral-facts}(\ref{sf:power}), we have that
$\phi_\pi(\langle q,1,1\rangle)=q^{p_j}$, $\phi_\pi(\langle1,q,1\rangle)=q^{p_k}$
and $\phi_\pi(\langle1,1,2\rangle)=2^{p_i}$
(this can be most easily seen by considering which leg of $\pi \cdot T$ is one-dimensional, with $T$ the tensor where $\phi_\pi$ is evaluated).

If $p_i>0$, then the AM-GM inequality, \cref{lem:superadd-general} for~$\phi_\pi$, and $\operatorname{Str}_q\degen\langle q+1\rangle$ give
\[
2q^{(p_j+p_k)/(2p_i)}\leq q^{p_j/p_i}+q^{p_k/p_i}\leq\phi_\pi(\operatorname{Str}_q)^{1/p_i}\leq(q+1)^{1/p_i}.
\]
So $2^{p_i}q^{(p_j+p_k)/2}\leq q+1$. This also holds if $p_i=0$, since $p_j,p_k\leq1$. Multiplying these three inequalities, one for each value of~$i$, gives $(2q)^{p_1+p_2+p_3}\leq(q+1)^3$. Since the maximum of $p_1+p_2+p_3$ over all $\phi\in\Delta_\FF$ is~$\omega_\FF$ (\cref{thm:spectral} and \cref{lem:spectral-facts}(\ref{sf:power})), this proves the claim.
\end{proof}

From the same tensor, Strassen's laser method~\cite[Corollary~7.5]{MR882307} (see also~\cite[Section~8 and Corollary~8.8]{blaser2013fast}) gives $4q^{\omega_\FF}\leq(q+1)^3$ for every integer $q\geq2$. This is best for $q=5$, where it gives $5^{\omega_\FF}\leq54$, equivalently $\omega_\FF<2.48$.

The proof of \cref{prop:laser-sep} is easily generalized to obtain \cref{lem:sym-partial} below, which we use in \cref{sec:intro-outline} for illustrative purposes. We will not need it for the remainder of the text, so it may be skipped.
We recover \cref{prop:laser-sep} by setting $S = \langle q,1,1\rangle$ and $s=t=1$, and applying $\operatorname{Str}_q \degen \langle q+1 \rangle$.

\begin{lemma}[Symmetrised partial direct sums]\label{lem:sym-partial}
Let~$S$ be a tensor over~$\FF$, and let $s,t\geq1$ be integers. Then
\[
\bigotimes_{\pi\in\mathfrak{S}_3}\pi\cdot\bigl(S\otimes\langle1,1,s+t\rangle\bigr)\asympleq\bigotimes_{\pi\in\mathfrak{S}_3}\pi\cdot\Bigl(\bigl(S\otimes\langle1,1,s\rangle\bigr)\oplusX\bigl(((23)\cdot S)\otimes\langle1,1,t\rangle\bigr)\Bigr).
\]
\end{lemma}
\begin{proof}
Let $T\coloneqq(S\otimes\langle1,1,s\rangle)\oplusX(((23)\cdot S)\otimes\langle1,1,t\rangle)$, and let $\phi\in\Delta_\FF$. By Strassen's duality theorem~\cite[Theorem~3.2]{strassen1988asymptotic}, and since~$\phi$ is multiplicative and the six permutations form three pairs $\{\pi,\pi\circ(23)\}$, it suffices to show that
\[
\phi_\pi(T)\,\phi_{\pi\circ(23)}(T)\geq\phi_\pi\bigl(S\otimes\langle1,1,s+t\rangle\bigr)\,\phi_{\pi\circ(23)}\bigl(S\otimes\langle1,1,s+t\rangle\bigr)
\]
for every $\pi\in\mathfrak{S}_3$. Fix~$\pi$, and let $p\coloneqq\log\phi_\pi(\langle1,1,2\rangle)$, $x\coloneqq\phi_\pi(S)$ and $y\coloneqq\phi_{\pi\circ(23)}(S)$. Since $(23)\cdot\langle1,1,2\rangle\cong\langle1,1,2\rangle$, also $\phi_{\pi\circ(23)}(\langle1,1,2\rangle)=2^p$, so $\phi_\pi(\langle1,1,m\rangle)=\phi_{\pi\circ(23)}(\langle1,1,m\rangle)=m^p$ for every integer $m\geq1$ (\cref{lem:spectral-facts}(\ref{sf:power}) and~(\ref{sf:permutation})). So the right-hand side is $(s+t)^{2p}xy$. Moreover, $\phi_\pi((23)\cdot S)=y$ and $\phi_{\pi\circ(23)}((23)\cdot S)=x$. If $p>0$, then \cref{lem:superadd-general} for~$\phi_\pi$ and for~$\phi_{\pi\circ(23)}$ gives
\[
\phi_\pi(T)^{1/p}\geq sx^{1/p}+ty^{1/p}\quad\text{and}\quad\phi_{\pi\circ(23)}(T)^{1/p}\geq sy^{1/p}+tx^{1/p},
\]
and $(sx^{1/p}+ty^{1/p})(sy^{1/p}+tx^{1/p})\geq(s+t)^2(xy)^{1/p}$, since the difference is $st(x^{1/p}-y^{1/p})^2$. Raising to the power~$p$ gives the claim. If $p=0$, then the restriction $T\geq S\otimes\langle1,1,s\rangle$ gives $\phi_\pi(T)\geq x$ and $\phi_{\pi\circ(23)}(T)\geq y$, which is the claim.
\end{proof}

\subsection{Concavity and tripling for each leg}\label{sec:rect-single}
In this subsection we will bound the values $\phi(\polmul_{a,b})$ of a spectral point~$\phi$ at the polynomial multiplication tensors.
We will show that, if $p\coloneqq\log\phi(\langle1,1,2\rangle)>0$, then $b \mapsto \phi(\polmul_{a,b})^{1/p}$ is concave and satisfies a tripling inequality (\cref{lem:profiles}). 
The results follows directly from the degenerations of \cref{sec:constructions}, combined with \cref{lem:superadd-general}, see \cref{tab:profiles}.
\begin{table}[H]
\centering
\renewcommand{\arraystretch}{1.3}
\begin{tabular*}{\textwidth}{@{}l@{\hspace{0.6em}}l@{\extracolsep{\fill}\hspace{0.8em}}l@{}}
\toprule
& Degeneration & Inequality \\
\midrule
(D1) & $\polmul_{a,b} \degen \langle a+b-1\rangle$ & $\phi(\polmul_{a,b}) \leq a+b-1$\\
(D3) & $\polmul_{a,b+1}\oplusX\polmul_{a,b-1} \degen \polmul_{a,b}\oplusX\polmul_{a,b}$, $b\geq2$ & $\phi(\polmul_{a,b+1})^{1/p}+\phi(\polmul_{a,b-1})^{1/p} \leq 2\phi(\polmul_{a,b})^{1/p}$\\
(D4) & $\polmul_{a,b}\oplusX((23)\cdot\polmul_{a,b})\oplusX\polmul_{a,b} \degen \polmul_{a,3b+a-1}$ & $2\phi(\polmul_{a,b})^{1/p}+\phi((23)\cdot\polmul_{a,b})^{1/p}$\\
&& $\qquad\qquad\qquad\leq \phi(\polmul_{a,3b+a-1})^{1/p}$\\
\bottomrule
\end{tabular*}
\caption{The ingredients \ruleref{D1}, \ruleref{D3} and~\ruleref{D4} of \cref{sec:ingredients} (\cref{prop:polmul-rank}, \cref{lem:det-degen} and \cref{lem:tripling-degen}), for integers $a,b\geq1$, and the inequalities that they give for a spectral point $\phi\in\Delta_\FF$, where $p\coloneqq\log\phi(\langle1,1,2\rangle)>0$ in the last two rows.}
\label{tab:profiles}
\end{table}

We call a function $h:\ZZ_{\geq1}\to\RR$ \emph{concave} if $h(m+1)+h(m-1)\leq2h(m)$ for every $m\geq2$.

\begin{lemma}\label{lem:profiles}
Let $\phi\in\Delta_\FF$, and let $p\coloneqq\log\phi(\langle1,1,2\rangle)$.
\begin{enumerate}[\upshape(i)]
\item\label{pr:basic} For all $a,b\geq1$, we have $\min(a,b)\leq\phi(\polmul_{a,b})\leq a+b-1$, and $\phi(\polmul_{a,b})$ is nondecreasing in~$a$ and in~$b$.
\item\label{pr:concave} If $p>0$, then for all $a\geq1$, the function $b\mapsto\phi(\polmul_{a,b})^{1/p}$ is concave.
\item\label{pr:tripling} If $p>0$, then for all $a,b\geq1$, we have $2\phi(\polmul_{a,b})^{1/p}+\phi((23)\cdot\polmul_{a,b})^{1/p}\leq\phi(\polmul_{a,3b+a-1})^{1/p}$.
\end{enumerate}
\end{lemma}
\begin{proof}
(\ref{pr:basic}) For $n=\min(a,b)$, the degenerations $\langle n\rangle\degen\polmul_{n,n}\leq\polmul_{a,b}$ (\cref{lem:unit-polmul}) give $n=\phi(\langle n\rangle)\leq\phi(\polmul_{a,b})$ (\cref{lem:spectral-facts}(\ref{sf:monotone}) and~(\ref{sf:unit})), and $\polmul_{a,b}\degen\langle a+b-1\rangle$ (\cref{prop:polmul-rank}) gives $\phi(\polmul_{a,b})\leq a+b-1$. If $a_1\leq a_2$ and $b_1\leq b_2$, then $\polmul_{a_1,b_1}\leq\polmul_{a_2,b_2}$, so $\phi(\polmul_{a_1,b_1})\leq\phi(\polmul_{a_2,b_2})$.

(\ref{pr:concave}) and~(\ref{pr:tripling}) Suppose that $p>0$. By \cref{lem:det-degen} and \cref{lem:two-copies}, for all integers $a\geq1$ and $b\geq2$,
\[
\polmul_{a,b+1}\oplusX\polmul_{a,b-1}\degen\polmul_{a,b}\otimes\langle1,1,2\rangle,
\]
and by \cref{lem:tripling-degen}, for all integers $a,b\geq1$,
\[
\polmul_{a,b}\oplusX((23)\cdot\polmul_{a,b})\oplusX\polmul_{a,b}\degen\polmul_{a,3b+a-1}.
\]
Moreover, $\phi(\polmul_{a,b}\otimes\langle1,1,2\rangle)=2^p\phi(\polmul_{a,b})$. So \cref{lem:superadd-general} and the monotonicity of~$\phi$ under degeneration give~(\ref{pr:concave}) from the first degeneration and~(\ref{pr:tripling}) from the second.
\end{proof}

In the subsequent sections we will use the following definition as a shorthand:
\begin{align}
    \phi(a,b) \coloneqq \phi(\polmul_{a,b}).
\end{align}

\subsection{Square matrix multiplication}\label{sec:rect-square}
In this subsection we will prove \cref{thm:main}. By \cref{thm:spectral}, it suffices to show that $\phi(\langle2,2,2\rangle)\leq2^{9/4}$ for every $\phi\in\Delta_\FF$. 
For this we will use \cref{lem:profiles} applied to $\phi_{\pi}$ with $\pi$ ranging over $\mathfrak{S}_3$.
In the case that $p\coloneqq\log\phi_\pi(\langle1,1,2\rangle)>0$, we will apply \cref{lem:paired-step} to the functions $b \mapsto \phi_\pi(a,b)^{1/p}$ and $b\mapsto \phi_{\pi \circ (23)}(a,b)^{1/p}$ appearing in \cref{lem:profiles}. 
Roughly speaking, this will then allow us to lower bound the growth of $n \mapsto \phi_{\pi}(n,n)$ in terms of~$p$.
Subsequently, we will obtain a lower bound on $\phi$ evaluated on $\bigotimes_{\pi \in \mathfrak{S}_3} \pi \cdot \polmul_{n,n}$ in terms of $\phi(\langle2,2,2\rangle)$. The upper bound will then follow from the fact that $\bigotimes_{\pi \in \mathfrak{S}_3} \pi \cdot \polmul_{a,b}$ has low border rank (\cref{prop:polmul-rank}).

In \Cref{sec:rect-cap,sec:rect-proof} we will treat the rectangular case. Those sections will not rely on any result from this section, and the results there will also recover the results of this section. So, this section can also safely be skipped. 

We begin with a general lemma about concave functions.

\begin{lemma}\label{lem:paired-step}
Let $a\geq1$ be an integer, and let $g,h\colon\ZZ_{\geq1}\to\RR_{>0}$ be concave functions such that $2g(m)+h(m)\leq g(3m+a-1)$ and $2h(m)+g(m)\leq h(3m+a-1)$ for every integer $m\geq1$. Then for every integer $b\geq1$,
\[
\frac{g(b+1)}{g(b)}\, \frac{h(b+1)}{h(b)}\geq\Bigl(\frac{2b+a+1}{2b+a-1}\Bigr)^2.
\]
\end{lemma}
\begin{proof}
Let $D=2b+a-1$ and $S=g(b)+h(b)$. Since~$g$ is concave, its differences $g(k+1)-g(k)$ do not increase with~$k$. So $g(b+1)-g(b)$ is at least the average of the~$D$ differences $g(j+1)-g(j)$ with $b\leq j<3b+a-1$, and the first hypothesis with $m=b$ gives
\[
g(b+1)-g(b)\geq\frac{g(3b+a-1)-g(b)}{D}\geq\frac SD.
\]
Likewise, the second hypothesis gives $h(b+1)-h(b)\geq S/D$. Since $S^2\geq4g(b)\,h(b)$, we find
\[
g(b+1)\,h(b+1)\geq\Bigl(g(b)+\frac SD\Bigr)\Bigl(h(b)+\frac SD\Bigr)=g(b)\,h(b)+\frac{S^2}D+\frac{S^2}{D^2}\geq g(b)\,h(b)\Bigl(1+\frac2D\Bigr)^2.\qedhere
\]
\end{proof}

\begin{proposition}\label{prop:angular-square}
Let $\phi\in\Delta_\FF$, and let $\tau\coloneqq\log\phi(\langle2,2,2\rangle)$. Then for every integer $n\geq1$,
\[
\prod_{\pi\in\mathfrak{S}_3}\phi(\pi\cdot\polmul_{n,n})\geq\Bigl(\frac{3n-1}2\prod_{k=1}^{n-1}\frac{3k+1}{3k}\Bigr)^{2\tau}\geq n^{8\tau/3}.
\]
\end{proposition}
\begin{proof}
Let $p_1$, $p_2$ and~$p_3$ be the logarithms of~$\phi(\langle1,1,2\rangle)$, $\phi(\langle2,1,1\rangle)$ and~$\phi(\langle1,2,1\rangle)$, so that $\tau=p_1+p_2+p_3$ (\cref{lem:spectral-facts}(\ref{sf:power})). For every integer $n\geq1$, let $G_n\coloneqq\prod_{\pi\in\mathfrak{S}_3}\phi_\pi(n,n)$, the left-hand side of the claim. Since $\phi_\pi(1,1)=1$ (\cref{lem:profiles}(\ref{pr:basic}) for~$\phi_\pi$), we have $G_1=1$.
Now fix an integer $n\geq1$, and let $\pi\in\mathfrak{S}_3$. Using that $\phi_\pi(n+1,n)=\phi_{\pi\circ(12)}(n,n+1)$ and $\phi_\pi(n,n)=\phi_{\pi\circ(12)}(n,n)$, since $\phi_\pi(a,b)=\phi_{\pi\circ(12)}(b,a)$, we obtain
\[
\frac{\phi_\pi(n+1,n+1)}{\phi_\pi(n,n)}=\frac{\phi_\pi(n+1,n+1)}{\phi_\pi(n+1,n)}\cdot\frac{\phi_{\pi\circ(12)}(n,n+1)}{\phi_{\pi\circ(12)}(n,n)}.
\]
We multiply over $\pi\in\mathfrak{S}_3$. Since $\pi\mapsto\pi\circ(12)$ is a bijection of~$\mathfrak{S}_3$, this gives
\[
\frac{G_{n+1}}{G_n}=\prod_{a\in\{n,n+1\}}\ \prod_{\sigma\in\mathfrak{S}_3}\frac{\phi_\sigma(a,n+1)}{\phi_\sigma(a,n)}.
\]
The six permutations form three pairs $\{\sigma,\sigma\circ(23)\}$, one for each $i\in\{1,2,3\}$, with $\sigma(1)=i$. Let $a\in\{n,n+1\}$. If $p_i>0$, then by \cref{lem:profiles}, applied to~$\phi_\sigma$ and~$\phi_{\sigma\circ(23)}$, the functions $m\mapsto\phi_\sigma(a,m)^{1/p_i}$ and $m\mapsto\phi_{\sigma\circ(23)}(a,m)^{1/p_i}$ satisfy the hypotheses of \cref{lem:paired-step}. Here we use that both spectral points take the value~$2^{p_i}$ at~$\langle1,1,2\rangle$, since~$\sigma$ and~$\sigma\circ(23)$ move its one-dimensional first leg to leg~$i$, and that $\phi_\sigma((23)\cdot T)=\phi_{\sigma\circ(23)}(T)$ and $\phi_{\sigma\circ(23)}((23)\cdot T)=\phi_\sigma(T)$ for every tensor~$T$, since $\pi\cdot(\rho\cdot T)=(\pi\circ\rho)\cdot T$ for all $\pi,\rho\in\mathfrak{S}_3$. That lemma with $b=n$, raised to the power~$p_i$, gives
\[
\frac{\phi_\sigma(a,n+1)}{\phi_\sigma(a,n)}\cdot\frac{\phi_{\sigma\circ(23)}(a,n+1)}{\phi_{\sigma\circ(23)}(a,n)}\geq\Bigl(\frac{2n+a+1}{2n+a-1}\Bigr)^{2p_i}.
\]
If $p_i=0$, this holds as well, since the left-hand side is at least~$1$: $\phi_\sigma$ and~$\phi_{\sigma\circ(23)}$ are nondecreasing in the second argument (\cref{lem:profiles}(\ref{pr:basic})). Multiplying over the three pairs and the two values of~$a$ gives
\[
\frac{G_{n+1}}{G_n}\geq\Bigl(\frac{(3n+1)(3n+2)}{(3n-1)\,3n}\Bigr)^{2\tau},\quad\text{so}\quad G_n\geq\Bigl(\prod_{k=1}^{n-1}\frac{(3k+1)(3k+2)}{(3k-1)\,3k}\Bigr)^{2\tau},
\]
where we multiply these inequalities and use that $G_1=1$. The factors $(3k+2)/(3k-1)$ telescope to $(3n-1)/2$, which gives the first inequality of the claim. For the second, we use that $\tau\geq0$, $(3n-1)/2\geq n$, and $(3k+1)/(3k)\geq((k+1)/k)^{1/3}$ by the binomial theorem, so that $\prod_{k=1}^{n-1}(3k+1)/(3k)\geq n^{1/3}$.
\end{proof}

\begin{proof}[Proof of \cref{thm:main}]
Let $\phi\in\Delta_\FF$, and let $\tau\coloneqq\log\phi(\langle2,2,2\rangle)$. Since $\phi(\pi\cdot\polmul_{n,n})\leq2n-1$ for every $\pi\in\mathfrak{S}_3$ (\cref{lem:profiles}(\ref{pr:basic}) for~$\phi_\pi$), \cref{prop:angular-square} gives $n^{8\tau/3}\leq(2n-1)^6<64n^6$ for every integer $n\geq1$, so $8\tau/3\leq6$, that is, $\phi(\langle2,2,2\rangle)\leq2^{9/4}$. Since this holds for every $\phi\in\Delta_\FF$, \cref{thm:spectral} gives $\omega_\FF\leq9/4$.
\end{proof}

For a fixed integer $n\geq2$, the first inequality of \cref{prop:angular-square} and the upper bound $(2n-1)^6$ give, in the same way,
\[
\omega_\FF\leq\frac{3\log(2n-1)}{\log\Bigl(\frac{3n-1}2\prod_{k=1}^{n-1}\frac{3k+1}{3k}\Bigr)}.
\]
For example $\omega_\FF<2.74$ for $n=2$, the bound of the example in \cref{toy-example}, and $\omega_\FF<2.65$ for $n=3$.

\subsection{Quadratic constraints on asymptotic spectrum of matrix multiplication}\label{sec:rect-cap}
We now extend the analysis in \cref{lem:paired-step} and \cref{prop:angular-square} to the rectangular case, which we will use to prove \cref{thm:rect} in \cref{sec:rect-proof}.
The proof ideas are essentially the same, but we need to take into account the ratios of the ``weights'' given to the three legs resulting from the definition of $\omega_\FF(a,b,c)$.

The following lemma generalises \cref{lem:paired-step}, which is the case $u=v=b$ and $r=s=1$.
\begin{lemma}\label{lem:paired-growth}
Let $a\geq1$ be an integer, and let $g,h\colon\ZZ_{\geq1}\to\RR_{>0}$ be concave functions such that $2g(m)+h(m)\leq g(3m+a-1)$ and $2h(m)+g(m)\leq h(3m+a-1)$ for every integer $m\geq1$. Then for all integers $u,v,r,s\geq1$,
\[
\Bigl(\frac{g(u+r)}{g(u)}\Bigr)^s\Bigl(\frac{h(v+s)}{h(v)}\Bigr)^r\geq\Bigl(\frac{D}{D-\frac{4rs}{r+s}}\Bigr)^{r+s},
\]
where $D\coloneqq u+v+r+s+a-1$.
\end{lemma}
\begin{proof}
Let $m\coloneqq\max(u+r,v+s)$ and $\lambda\coloneqq\frac{g(m)+h(m)}{2m+a-1}>0$. By concavity and the first hypothesis, for every integer $1\le k<m$,
\[
g(k+1)-g(k)
\ge \frac{g(3m+a-1)-g(m)}{2m+a-1}
\ge\lambda.
\]
Likewise, \(h(k+1)-h(k)\ge\lambda\). Writing $A\coloneqq g(u+r)$ and $B\coloneqq h(v+s)$, summing these difference bounds gives
$
A-g(u)\ge r\lambda
$ and
$
B-h(v)\ge s\lambda
$.
Similarly, $g(m) - A \geq (m - u - r)\lambda$. Doing the same for $h$, and then using that $g(m)+h(m) = \lambda(2m+a-1)$, we also find 
\[
A+B\le g(m)+h(m)-\lambda(2m-u-r-v-s)=\lambda D.
\]
By the weighted arithmetic-mean geometric-mean inequality,
\begin{align}\label{eq:lower-bound-amgm}
\left(\frac{g(u)}{A}\right)^{s/(r+s)}
\left(\frac{h(v)}{B}\right)^{r/(r+s)}
\le \frac{1}{r+s} \left( s \frac{g(u)}{A} + r \frac{h(v)}{B} \right).
\end{align} 
We have seen that $1 - g(u)/A \geq r\lambda/A$ and $1 - h(v)/B \geq s\lambda/B$, and that $\lambda D \geq A + B$.
Using this and also $A^{-1}+B^{-1}\ge4/(A+B)$, \cref{eq:lower-bound-amgm} can be upper bounded by
\begin{align*}
\frac{s(1-r\lambda/A)+r(1-s\lambda/B)}{r+s}
&=
1-\frac{rs\lambda}{r+s}
\left(\frac1A+\frac1B\right)
\leq
1 - \frac{4rs}{(r+s)D}.
\end{align*}
Taking reciprocals and raising to the power \(r+s\) proves the result.
\end{proof}

\begin{proposition}\label{prop:angular}
For every $\phi\in\Delta_\FF$ and all real numbers $x_1,x_2,x_3\geq0$,
\[
\phi(\langle1,1,2\rangle)^{x_2x_3}\,\phi(\langle2,1,1\rangle)^{x_1x_3}\,\phi(\langle1,2,1\rangle)^{x_1x_2}\leq2^{(x_1+x_2+x_3)^2/4}.
\]
\end{proposition}
\begin{proof}
Let $p_1$, $p_2$ and~$p_3$ be the logarithms of~$\phi(\langle1,1,2\rangle)$, $\phi(\langle2,1,1\rangle)$ and~$\phi(\langle1,2,1\rangle)$, and let $S=x_1+x_2+x_3$ and $B=p_1x_2x_3+p_2x_1x_3+p_3x_1x_2$. Taking logarithms, the claim is $4B\leq S^2$. Both sides of $4B\leq S^2$ are homogeneous of degree~$2$ in $(x_1,x_2,x_3)$ and continuous, so we may assume that $x_1$, $x_2$ and~$x_3$ are positive integers. For every integer $n\geq1$, let
\[
G_n\coloneqq\prod_{\pi\in\mathfrak{S}_3}\phi_\pi\bigl(nx_{\pi(1)},nx_{\pi(2)}\bigr)^{x_{\pi(3)}}.
\]
Since $1\leq\phi_\pi(a,b)<a+b$ (\cref{lem:profiles}(\ref{pr:basic}) for~$\phi_\pi$) and $\sum_{\pi\in\mathfrak{S}_3}x_{\pi(3)}=2S$, we have $G_1\geq1$ and $G_n\leq(nS)^{2S}$.
For $\sigma\in\mathfrak{S}_3$ and $a\in\{n,n+1\}$, let
\[
R_\sigma(a)\coloneqq\frac{\phi_\sigma\bigl(ax_{\sigma(1)},(n+1)x_{\sigma(2)}\bigr)}{\phi_\sigma\bigl(ax_{\sigma(1)},nx_{\sigma(2)}\bigr)}
\]
be the growth factor of~$\phi_\sigma$ in the second argument, at the first argument $ax_{\sigma(1)}$. Let $\pi\in\mathfrak{S}_3$. 
Using $\phi_\pi(a,b)=\phi_{\pi\circ(12)}(b,a)$ we find
\[
\frac{\phi_\pi\bigl((n+1)x_{\pi(1)},(n+1)x_{\pi(2)}\bigr)}{\phi_\pi\bigl(nx_{\pi(1)},nx_{\pi(2)}\bigr)}
=
R_\pi(n+1)\,
R_{\pi\circ(12)}(n).
\]
We raise this to the power $x_{\pi(3)}$, which equals $x_{(\pi\circ(12))(3)}$, and multiply over $\pi\in\mathfrak{S}_3$. Since $\pi\mapsto\pi\circ(12)$ is a bijection of~$\mathfrak{S}_3$, this gives
\[
\frac{G_{n+1}}{G_n}=\prod_{a\in\{n,n+1\}}\ \prod_{\sigma\in\mathfrak{S}_3}R_\sigma(a)^{x_{\sigma(3)}}.
\]

The six permutations form three pairs $\{\sigma,\sigma\circ(23)\}$, one for each $i\in\{1,2,3\}$, with $\sigma(1)=i$. Let $j=\sigma(2)$ and $k=\sigma(3)$. %
Let $a\in\{n,n+1\}$. If $p_i>0$, then by \cref{lem:profiles}, applied to~$\phi_\sigma$ and~$\phi_{\sigma\circ(23)}$, the functions $m\mapsto\phi_\sigma(ax_i,m)^{1/p_i}$ and $m\mapsto\phi_{\sigma\circ(23)}(ax_i,m)^{1/p_i}$ satisfy the hypotheses of \cref{lem:paired-growth} with~$ax_i$ in place of~$a$. Here we use that both spectral points take the value~$2^{p_i}$ at~$\langle1,1,2\rangle$, since~$\sigma$ and~$\sigma\circ(23)$ move its one-dimensional first leg to leg~$i$, and that $\phi_\sigma((23)\cdot T)=\phi_{\sigma\circ(23)}(T)$ and $\phi_{\sigma\circ(23)}((23)\cdot T)=\phi_\sigma(T)$ for every tensor~$T$, since $\pi\cdot(\rho\cdot T)=(\pi\circ\rho)\cdot T$ for all $\pi,\rho\in\mathfrak{S}_3$. That lemma with $u=nx_j$, $v=nx_k$, $r=x_j$ and $s=x_k$, raised to the power~$p_i$, gives
\[
R_\sigma(a)^{x_k}R_{\sigma\circ(23)}(a)^{x_j}\geq\Bigl(\frac{D}{D-c}\Bigr)^{p_i(x_j+x_k)},
\]
where $D=(n+1)(x_j+x_k)+ax_i-1$ and $c=4x_jx_k/(x_j+x_k) > 0$. Since $D\leq(n+1)S$ and $t\coloneqq c/S\leq(x_j+x_k)/S\leq1$, we find, using Bernoulli's inequality (which says that $( 1 + y )^t \leq 1 + t y$ for every real number $0 \leq t \leq 1$ and $y \geq -1$) 
in the last step, that
\[
\frac{D}{D-c}\geq1+\frac cD\geq1+\frac t{n+1}\geq\Bigl(\frac{n+2}{n+1}\Bigr)^t.
\]
Since $tp_i(x_j+x_k)=4p_ix_jx_k/S$, we obtain
\[
R_\sigma(a)^{x_k}R_{\sigma\circ(23)}(a)^{x_j}\geq\Bigl(\frac{n+2}{n+1}\Bigr)^{4p_ix_jx_k/S}.
\]
If $p_i=0$, this lower bound holds as well, since the left-hand side is at least~$1$: $\phi_\sigma$ and~$\phi_{\sigma\circ(23)}$ are nondecreasing in the second argument (\cref{lem:profiles}(\ref{pr:basic})). Multiplying over the three pairs and the two values of~$a$ gives
\[
\frac{G_{n+1}}{G_n}\geq\Bigl(\frac{n+2}{n+1}\Bigr)^{8B/S},
\quad\text{so}\quad
G_n\geq\Bigl(\frac{n+1}2\Bigr)^{8B/S},
\]
where we multiply these inequalities and use $G_1\geq1$. With the upper bound on~$G_n$, this gives $((n+1)/2)^{8B/S}\leq(nS)^{2S}$ for every integer $n\geq1$. Hence $8B/S\leq2S$, that is, $4B\leq S^2$.
\end{proof}

 \cref{prop:angular} gives an ``outer approximation'' of the asymptotic spectrum of rectangular matrix multiplication, which we plot in \Cref{fig:spectrum}. To explain this, let
\[
\Deltam\coloneqq\bigl\{\bigl(\phi(\langle1,1,2\rangle),\phi(\langle2,1,1\rangle),\phi(\langle1,2,1\rangle)\bigr):\phi\in\Delta_\FF\bigr\}\subseteq[1,2]^3
\]
(\cref{lem:spectral-facts}(\ref{sf:power})), and let $\log\Deltam\subseteq[0,1]^3$ be its image under the logarithm in each coordinate. The flattening ranks $\flatrank{1}$, $\flatrank{2}$ and~$\flatrank{3}$ give the points 
\[
\{(0,1,1), (1,0,1), (1,1,0)\}\subseteq\log\Deltam,
\]
and Strassen's support functionals give the triangle with these vertices~\cite[Section~6]{strassen1991degeneration},
\[
P \coloneqq \conv \{(0,1,1), (1,0,1), (1,1,0)\}\subseteq\log\Deltam.
\]
Strassen showed that $\log\Deltam$ lies in the truncated pyramid $\{w\in[0,1]^3: 2\leq w_1+w_2+w_3\leq\omega_\FF\}$ over~$P$~\cite[Proposition~4.2]{strassen1988asymptotic}.
By \cref{prop:angular}, we have
\begin{equation}\label{eq:Q}
\log\Deltam \subseteq Q 
\coloneqq \bigl\{w\in[0,1]^3 : \forall x\in\RR_{\geq0}^3,  4(w_1x_2x_3+w_2x_1x_3+w_3x_1x_2)\leq(x_1+x_2+x_3)^2\bigr\}.
\end{equation}
So $P\subseteq\log\Deltam\subseteq Q\cap\{w: w_1+w_2+w_3\geq2\}$, and these are the best inner and outer approximations of $\log\Deltam$ that we currently know.

\begin{figure}[h]
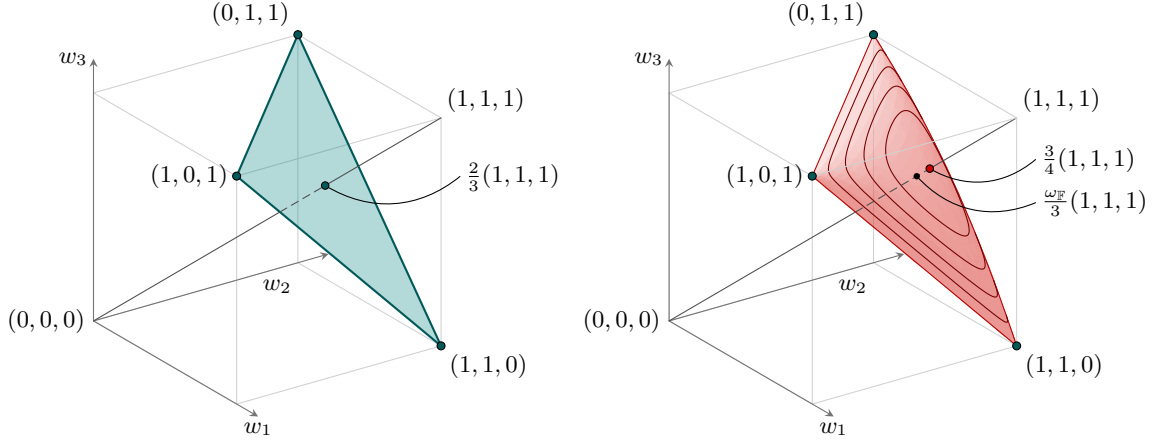

\centering

\caption{$P\subseteq\log\Deltam\subseteq Q\cap\{w: w_1+w_2+w_3\geq2\}$, for the triangle~$P$ on the left, given by Strassen's support functionals, and the part of~$Q$ with $w_1+w_2+w_3\geq2$ on the right, shown with level curves of $w_1+w_2+w_3$. Of course we do not know the position of $\frac{\omega_\FF}3(1,1,1)$, but it is on the line between $\tfrac23(1,1,1)$ and $\tfrac34(1,1,1)$.}
\label{fig:spectrum}
\end{figure}

On $Q$, we have $w_1+w_2+w_3\leq\frac94$, by the choice $x_1=x_2=x_3=1$, with equality at $\frac34(1,1,1)$. Since the maximum of $w_1+w_2+w_3$ on $\log\Deltam$ is~$\omega_\FF$ (\cref{thm:spectral} and \cref{lem:spectral-facts}(\ref{sf:power})), this gives \cref{thm:main} again. 

Likewise, for real numbers $a\geq b\geq c\geq0$, the maximum of $aw_2+bw_3+cw_1$ on~$Q$ is as follows: %

\begin{lemma}\label{lem:rect-bound}
Let $a\geq b\geq c\geq0$ be real numbers, and let $w\in[0,1]^3$ with $4(w_1x_2x_3+w_2x_1x_3+w_3x_1x_2)\leq(x_1+x_2+x_3)^2$ for all $x\in\RR_{\geq0}^3$. Then
\[
aw_2+bw_3+cw_1\leq\begin{cases}a+b&\text{if }c(a+b)\leq ab,\\[2pt]\dfrac{(ab+bc+ca)^2}{4abc}&\text{if }c(a+b)>ab.\end{cases}
\]
\end{lemma}
\begin{proof}
If $a=0$, then $a=b=c=0$, and the claim holds. So let $a>0$. With $x=(ab,bc,ca)$, the left-hand side of the hypothesis is $4abc(aw_2+bw_3+cw_1)$, and at $c=ab/(a+b)$, where the two cases meet, this vector~$x$ is a multiple of $(a+b,b,a)$. If $c(a+b)\leq ab$, then, since $w_1\geq0$,
\[
(a+b)(aw_2+bw_3+cw_1)\leq(a+b)(aw_2+bw_3)+abw_1\leq(a+b)^2,
\]
where the second inequality is the hypothesis with $(x_1,x_2,x_3)=(a+b,b,a)$. Dividing by $a+b$ gives the claim. If $c(a+b)>ab$, then $abc>0$, and the hypothesis with $(x_1,x_2,x_3)=(ab,bc,ca)$ gives $4abc(aw_2+bw_3+cw_1)\leq(ab+bc+ca)^2$.
\end{proof}

The bound of \cref{lem:rect-bound} has two cases which naturally correspond to whether the maximum is at the boundary of $[0,1]^3$ or in the interior. If $c(a+b)\leq ab$, then the maximum is attained at the vertex~$(0,1,1)$ of $Q$. If $c(a+b) > ab$, the the maximum is attained on the curved part of the boundary.

\begin{remark}
The set~$\Deltam$ is invariant under permuting the coordinates (\cref{lem:spectral-facts}(\ref{sf:permutation})), so the convex hull of~$\log\Deltam$ (and therefore $Q$) contains $\frac{\omega_\FF}3(1,1,1)$. Whether $\frac{\omega_\FF}3(1,1,1)$ lies in~$\log\Deltam$ is not known.
See also \cref{fig:spectrum}.
\end{remark}

\subsection{Rectangular matrix multiplication}\label{sec:rect-proof}
In this subsection we will deduce \cref{thm:rect} from \cref{prop:angular} and \cref{lem:spectral-omega}. We state it again for convenience of the reader.

\thmrect*

The value~$\omega_\FF(a,b,c)$ does not depend on the order of~$a$, $b$ and~$c$, since permuting the legs does not change the tensor rank (\cref{lem:restriction-facts}(\ref{rf:cyclic})). If $a\geq b\geq c$, then $\omega_\FF(a,b,c)\geq a+b$, since the tensor rank of $\langle n_1,n_2,n_3\rangle$ is at least $\flatrank{1}(\langle n_1,n_2,n_3\rangle)=n_1n_2$.

\begin{proof}[Proof of \cref{thm:rect}]
Let $\phi\in\Delta_\FF$. 
Let $p_1$, $p_2$ and~$p_3$ be the logarithms of~$\phi(\langle1,1,2\rangle)$, $\phi(\langle2,1,1\rangle)$ and~$\phi(\langle1,2,1\rangle)$. 
By taking the logarithm of the equality in \cref{lem:spectral-omega}, we see it is sufficient to upper bound $ap_2+bp_3+cp_1$.
Because $p_1,p_2,p_3\in[0,1]$ (\cref{lem:spectral-facts}(\ref{sf:power})), the logarithm of the inequality of \cref{prop:angular} is the hypothesis of \cref{lem:rect-bound} for $w=(p_1,p_2,p_3)$. This shows the desired upper bound.
Equality in the first case follows from the lower bound $\omega_\FF(a,b,c)\geq a+b$, obtained from evaluating the flattening rank $\flatrank{1}$. 
\end{proof}

In the second case the bound exceeds the lower bound $a+b$ by $(c(a+b)-ab)^2/(4abc)$, since $(ab+bc+ca)^2-4abc(a+b)=(c(a+b)-ab)^2$. For example, $\omega_\FF(6,3,2)=9$. For $a=b=c=1$, \cref{thm:rect} gives $\omega_\FF\leq9/4$ again, and for the shapes $(1,k,1)$ it gives the following bounds.

\begin{corollary}\label{cor:rect-k}
For every real number $k\geq0$,
\[
\omega_\FF(1,k,1)\leq\begin{cases}2&\text{if }0\leq k\leq1/2,\\[2pt]1+k+\dfrac{1}{4k}&\text{if }k\geq1/2.\end{cases}
\]
In particular, $\alpha_\FF\geq1/2$ and $\omega_\FF(1,2,1)\leq25/8$.
\end{corollary}
\begin{proof}
Since $\omega_\FF(a,b,c)$ does not depend on the order of~$a$, $b$ and~$c$, \cref{thm:rect} applies to $(1,1,k)$ if $k\leq1$ and to $(k,1,1)$ if $k\geq1$. If $k\leq1/2$, its first case applies, since $2k\leq1$, and gives $\omega_\FF(1,k,1)=2$, so $\alpha_\FF\geq1/2$. Otherwise its second case applies, and gives
\[
\omega_\FF(1,k,1)\leq\frac{(1+2k)^2}{4k}=1+k+\frac1{4k},
\]
which is $25/8$ for $k=2$.
\end{proof}

By \cref{cor:rect-k}, we have $\alpha_\FF\geq1/2$, $\omega_\FF(1,0.709,1)<2.0617$ and $\omega_\FF(1,2,1)\leq25/8$. These bounds improve the bounds $\alpha_\CC>0.465$ and $\omega_\CC(1,0.709,1)<2.092$ of OpenAI, which hold in characteristic zero, the second also outside a finite set of positive characteristics~\cite{OAI:Complex-Matrix-Multiplication-Below-2.258-and-Rectangular-Bounds-September-24-2026}. They also improve the bounds from the laser method: $\alpha_\FF>0.321334$ of Vassilevska Williams, Xu, Xu and Zhou~\cite{williams2023new}, and $\omega_\FF(1,2,1)\leq3.250035$ of Alman, Duan, Vassilevska Williams, Xu, Xu and Zhou~\cite{almanMoreAsymmetryYields2025}. \Cref{fig:rect-general} (right) compares them with those of the laser method.

\Cref{fig:rect-general} uses that~$\omega_\FF$ is subadditive, homogeneous and convex, which follows from taking tensor products and rescaling (Lotti and Romani~\cite[Section~1]{MR699269}). So the lower convex hull of upper bounds on~$\omega_\FF(1,k,1)$ is also an upper bound, and for all real numbers $0<b\leq1$ and $c\geq0$,
\[
\omega_\FF(1,b,c)\leq\omega_\FF(1-b,0,0)+\omega_\FF(b,b,c)=1-b+b\,\omega_\FF(1,c/b,1).
\]
The bound of \cref{thm:rect} is homogeneous of degree one in $(a,b,c)$, so its values with $a=1\geq b\geq c$ determine it. \Cref{fig:rect-general} (left) shows these values as a surface.

\section{Operational barrier}\label{sec:optimality}
In this section we will prove that no combination of the ingredients~\ruleref{D0}--\ruleref{D4} by the rules~\ruleref{C1}--\ruleref{C4} improves the bounds of \cref{thm:rect}.

To do this, we will define a semiring~$\mathcal{F}$ of formal combinations of the polynomial multiplication tensors and their permutations, under $\oplus$, $\otimes$ and the partial direct sums $\oplus_i$ (which are not semiring operations, and are only allowed under the natural leg dimension constraint). By ``formal combinations'' here we mean that we do not a priori identify restriction equivalent or isomorphic tensors. Instead, we define a special preorder $\preceq$ on $\mathcal{F}$ which only encodes the degeneration ingredients \ruleref{D0}--\ruleref{D4} and combination rules~\ruleref{C1}--\ruleref{C4} that we use in the proofs of \cref{thm:rect}. Thus any degeneration construction from those ingredients and rules can be written in the preorder $\preceq$. 

To prove no degenerations $\langle n, n, n\rangle \degen \langle r\rangle$ give $\log r / \log n < 9/4$, the natural approach would be to find a degeneration monotone function $\phi$ that obstructs $\langle n, n, n\rangle \degen \langle r\rangle$ in the sense that $\phi (\langle n,n,n \rangle) > \phi(\langle r\rangle)$ whenever $\log r / \log n < 9/4$.

In fact, more abstractly, one may take such $\phi$ in the asymptotic spectrum $\Delta_\FF$ and use $\max_{\phi \in \Delta_\FF} \phi(\langle 2,2,2\rangle) = 2^{\omega_\FF}$; that is, if $\omega_\FF = 9/4$, then there is a $\phi \in \Delta_\FF$ such that $\phi(\langle 2,2,2\rangle) = 2^{9/4}$, and then $\phi(\langle n,n,n\rangle) = n^{9/4}$ by basic powering arguments.

We do not know how to obtain such a degeneration monotone $\phi$; instead, we will restrict ourselves to only work with the semiring of formal expression $\mathcal{F}$ and the restricted operational preorder $\preceq$ on $\mathcal{F}$. We will then construct functions $F^\star_\theta$ on $\mathcal{F}$ that are $\preceq$-monotone (they are in fact in the asymptotic spectrum of $(\mathcal{F}, \preceq)$). These functions $F^\star_\theta$ are constructed in a way that they obstruct all $\langle n,n,n\rangle \preceq \langle r\rangle$ that can lead to $\omega_\FF < 9/4$ (and similarly for the rectangular bounds).
This ``operational barrier'' result is \cref{thm:barrier}.

The challenge lies in constructing the functions $F^\star_\theta$ as above, which we do in \cref{sec:barrier-functions},  
and prove their monotonicity in \cref{sec:properties-of-monotone-function} and \cref{sec:barrier-proof}.

This construction also gives a test for usefulness of different sets of degeneration ingredients to obtain better bounds on $\omega_\FF$: adding them to $\preceq$ keeps the barrier whenever the functions~$F_\theta^\star$ remain monotone.
This we state in \cref{cor:new-rules}.

\subsection{Symmetric notation $\polmul_{x,y,z}$ for permuted polynomial multiplication}\label{sec:combinations}
For convenience in this section, we will use a more symmetric notation for polynomial multiplication tensors and their permutations, $\pi\cdot \polmul_{a,b}$. To see how this is natural, recall
\[
\polmul_{a,b} = \sum_{i=1}^a\sum_{j=1}^be_i\otimes e_j\otimes e_{i+j-1}. %
\]
The terms are of the form $e_i \otimes e_j \otimes e_k$ with $i \in [a]$, $j \in [b]$, $k \in [a+b-1]$ such that $i+j-k=1$.
Replacing~$k$ by $a+b-k$ turns the condition $i+j-k=1$ into $i+j+k=a+b+1=M+2$, where $M=a+b-1$ is the largest of the three leg dimensions of $\polmul_{a,b}$. This definition is more symmetric and is what we use in the following: %

\begin{definition}\label{def:admissible}
A triple $(x,y,z)$ of real numbers $x,y,z\geq1$ is \emph{admissible} if 
\[
x+y+z=2\max(x,y,z)+1,
\]
that is, if its largest entry is the sum of the other two minus~$1$. For an admissible triple of integers, with $M=\max(x,y,z)$, let
\[
\polmul_{x,y,z}\coloneqq\sum_{\alpha,\beta,\gamma}e_\alpha\otimes e_\beta\otimes e_\gamma\in\FF^x\otimes\FF^y\otimes\FF^z,
\]
where the sum is over all $\alpha\in[x]$, $\beta\in[y]$ and~$\gamma\in[z]$ with $\alpha+\beta+\gamma=M+2$.
\end{definition}

By the discussion before, we have $\polmul_{a,b,a+b-1}\cong\polmul_{a,b}$. We see directly that permuting the legs of~$\polmul_{x,y,z}$ permutes $x$, $y$ and~$z$. For example, $(23)\cdot\polmul_{a,b}\cong\polmul_{a,a+b-1,b}$. 
Up to isomorphism, the tensors $\polmul_{x,y,z}$ for admissible triples $(x,y,z)$ of integers are exactly the tensors $\pi\cdot\polmul_{a,b}$ for $\pi\in\mathfrak{S}_3$ and integers $a,b\geq1$. We will thus from now on work with $\polmul_{x,y,z}$ instead of $\pi\cdot\polmul_{a,b}$.

\subsection{Set $\mathcal{F}$ of formal combinations of tensors}

We will work with formal combinations of tensors. This will enable us to be very precise about the degeneration ingredients, isomorphisms and combination rules that we allow to apply, by means of a formal preorder that we will define in the next subsection. 

We will use formal operations $\oplus$, $\otimes$, and also $\oplus_1, \oplus_2, \oplus_3$, which will represent partial direct sums in the three directions.

\begin{definition}\label{def:expressions}
We define a set~$\mathcal{F}$ of formal expressions, and three shorthand symbols for expressions in~$\mathcal{F}$.
\begin{enumerate}[\upshape(i)]
\item Let~$\mathcal{F}$ be the set of formal expressions generated by the symbols~$\polmul_{x,y,z}$, for all admissible triples $(x,y,z)$ of positive integers, under~$\oplus$, $\otimes$, $\oplus_1$, $\oplus_2$ and~$\oplus_3$, where $S\oplus_iT$ is allowed only if the $i$th legs of~$S$ and~$T$ have the same dimension.
\item Let $\langle1\rangle\coloneqq\polmul_{1,1,1}$.
\item For every positive integer~$r$, let~$\langle r\rangle$ be the direct sum of $r$~copies of~$\langle1\rangle$.
\item For all positive integers $n_1$, $n_2$ and~$n_3$, let $\langle n_1,n_2,n_3\rangle\coloneqq\polmul_{n_1,1,n_1}\otimes\polmul_{n_2,n_2,1}\otimes\polmul_{1,n_3,n_3}$.
\end{enumerate}
\end{definition}
Every expression in $\mathcal{F}$ represents a tensor, obtained by reading~$\polmul_{x,y,z}$ as the tensor of \cref{def:admissible} and recursively evaluating~$\oplus$, $\otimes$ and~$\oplus_i$ as the direct sum, the tensor product and the partial direct sums, which are defined analogously to $\oplus_1$. Distinct formal expressions may represent the same tensor.

Note that indeed the expressions~$\langle r\rangle$ and~$\langle n_1,n_2,n_3\rangle$ represent the unit tensors and, up to isomorphism, the matrix multiplication tensors, since $\polmul_{n,1,n}\cong\langle n,1,1\rangle$, $\polmul_{n,n,1}\cong\langle1,n,1\rangle$ and $\polmul_{1,n,n}\cong\langle1,1,n\rangle$ and $\langle n_1, n_2, n_3\rangle = \langle n_1, 1, 1\rangle \otimes \langle 1, n_2, 1\rangle \otimes \langle 1,1, n_3\rangle$. %

\subsection{Formal preorder relation $\preceq$ on $\mathcal{F}$ and operational barrier statement}

We now define the formal preorder $\preceq$ on $\mathcal{F}$, which encodes the degeneration ingredients, combination rules and basic restriction and isomorphisms used in the proofs of \cref{thm:main} and \cref{thm:rect}. We will then state the barrier in terms of it. 

It will be clear to the reader later that $\preceq$ can be extended with many more relations so that the barrier still holds, we give a somewhat minimal definition here.

\begin{definition}\label{def:formal-relation}
Let~$\preceq$ be the smallest preorder on~$\mathcal{F}$ that is compatible with~$\oplus$ and~$\otimes$, that is, $T_1\preceq S_1$ and $T_2\preceq S_2$ imply $T_1\oplus T_2\preceq S_1\oplus S_2$ and $T_1\otimes T_2\preceq S_1\otimes S_2$, and that has the following properties, 
for all expressions $C,S,T,U,S_1,\ldots,S_n,T_1,\ldots,T_m\in\mathcal{F}$ and all positive integers $a$, $b$, $m$, $n$, $t$, $x$, $y$ and~$z$, whenever the triples of the tensors $\polmul_{x,y,z}$ in these properties are admissible and the expressions lie in~$\mathcal{F}$. Each property is also required with the legs renumbered, that is, with the indices of every~$\polmul$ and of every~$\oplus_i$ in it permuted in the same way.
\begin{enumerate}[({D}1)]
\item\label{frule:D1} $\polmul_{x,y,z}\preceq\langle\max(x,y,z)\rangle$.
\item\label{frule:D2} $(T_1\oplus\cdots\oplus T_m)\otimes\polmul_{1,t,t}\preceq(T_1\oplusX\cdots\oplusX T_m)\otimes\polmul_{m,t,m+t-1}$.
\item\label{frule:D3} $\polmul_{x,y+1,z+1}\oplusX\polmul_{x,y-1,z-1}\preceq\polmul_{x,y,z}\oplusX\polmul_{x,y,z}$, also with a further summand~$\oplusX C$ on both sides.
\item\label{frule:D4} $\polmul_{x,y,z}\oplusX\polmul_{x,z,y}\oplusX\polmul_{x,y,z}\preceq\polmul_{x,2y+z,y+2z}$, also with a further summand~$\oplusX C$ on both sides.
\end{enumerate}
\begin{enumerate}
\item[(R)]\phantomsection\label{rule:R} \emph{Restrictions.}
\begin{enumerate}[({R}1)]
\item\label{rule:R1} $T_1\preceq T_1\oplus T_2$, $T_1\preceq T_1\oplusX T_2$ and $T_1\oplusX T_2\preceq T_1\oplus T_2$.
\item\label{rule:R2} $\polmul_{x,y-1,z-1}\preceq\polmul_{x,y,z}$.
\item\label{rule:R3} $\polmul_{a,b,a+b-1}\preceq\polmul_{a,1,a}\otimes\polmul_{1,b,b}$.
\end{enumerate}
\item[(I)]\phantomsection\label{rule:I} \emph{Isomorphisms.} $T\preceq S$ and $S\preceq T$, if $T$ and~$S$ are the two sides of one of the following isomorphisms:
\begin{enumerate}[({I}1)]
\item\label{rule:I1} The associativity and commutativity of~$\oplus$ and~$\otimes$, the distributivity of~$\otimes$ over~$\oplus$, and $T\otimes\langle1\rangle\cong T$.
\item\label{rule:I2} Any two orders and bracketings of a partial direct sum $T_1\oplusX\cdots\oplusX T_m$, and the distributivity of~$\otimes$ over~$\oplusX$: the expression $(T_1\oplusX\cdots\oplusX T_m)\otimes(S_1\oplusX\cdots\oplusX S_n)$ and the partial direct sum of the $mn$ expressions $T_i\otimes S_j$.
\item\label{rule:I3} For all positive integers $n_1,n_2,n_3,m_1,m_2,m_3$, the isomorphism $\langle n_1,n_2,n_3\rangle\otimes\langle m_1,m_2,m_3\rangle\cong\langle n_1m_1,n_2m_2,n_3m_3\rangle$ (\cref{lem:restriction-facts}(\ref{rf:product})).
\item\label{rule:I4} The partial direct sum of $r$~copies of~$U$ and $U\otimes\polmul_{1,r,r}$, for every integer $r\geq1$ (\cref{lem:two-copies}).
\end{enumerate}
\end{enumerate}
\end{definition}
\Cref{def:formal-relation} collects the combination rules and the ingredients of \cref{sec:ingredients}: the preorder and its compatibility with~$\oplus$ and~$\otimes$ are the rules \ruleref{C1}--\ruleref{C3}, the renumbering of the legs takes the place of~\ruleref{C4}, and \fruleref{D1}--\fruleref{D4} are the ingredients. Whenever the triples in \fruleref{D3} and~\fruleref{D4} are admissible, either $z=x+y-1$ or $y=x+z-1$, so these are the ingredients of \cref{sec:ingredients} with the shared leg first. The properties \ruleref{R} and~\ruleref{I} add (standard) restrictions and isomorphisms, and ingredient~\ruleref{D0}, $\polmul_{1,b,b}\cong\langle1,1,b\rangle$, follows from~\ruleref{I}.

If $T\preceq S$, then $T\degen S$, reading both sides as the tensors that they represent. This will make every relation $\langle n_1,n_2,n_3\rangle\preceq\langle r\rangle$ give an upper bound on~$\omega_\FF$. 
Indeed, $\preceq$ is the smallest preorder with the properties of \cref{def:formal-relation}, and each property still holds with~$\degen$ in place of~$\preceq$, by \ruleref{D0}-\ruleref{D4}.

We now state the barrier.

\begin{theorem}[Operational barrier]\label{thm:barrier}
Let $n_1,n_2,n_3,r\geq1$ be integers with $\langle n_1,n_2,n_3\rangle\preceq\langle r\rangle$.
\begin{enumerate}[\upshape(i)]
\item\label{bar:general} Let $a\geq b\geq c$ be the numbers $\log n_1$, $\log n_2$ and~$\log n_3$ in nonincreasing order. Then
\begin{align}
\label{eq:barrier}
\log r\geq\begin{cases}a+b&\text{if }c(a+b)\leq ab,\\[2pt]\dfrac{(ab+bc+ca)^2}{4abc}&\text{if }c(a+b)>ab.\end{cases}
\end{align}
\item\label{bar:square} In particular, if $n_1=n_2=n_3=n$, then $n^{9/4}\leq r$.
\end{enumerate}
\end{theorem}

\Cref{thm:barrier} relates to barriers for $\omega_\FF(a,b,c)$ as follows.
Every relation $\langle n_1,n_2,n_3\rangle\preceq\langle r\rangle$ is a degeneration $\langle n_1,n_2,n_3\rangle\degen\langle r\rangle$, so $\phi(\langle n_1,n_2,n_3\rangle)\leq r$ for every $\phi\in\Delta_\FF$, and \cref{lem:spectral-facts}(\ref{sf:power}) and \cref{lem:spectral-omega} give $\omega_\FF(\log n_1,\log n_2,\log n_3)\leq\log r$. For $n_1=n_2=n_3=n\geq2$, this is $\omega_\FF\leq\log r/\log n$. By \cref{thm:barrier}, these bounds are at least the bound of \cref{thm:rect}, and so are their limits, since that bound is continuous and homogeneous in~$(a,b,c)$. In particular, no such relation, and no limit of such bounds, gives an upper bound on~$\omega_\FF$ below~$9/4$.

\subsection{Definition of functions $F^\star_\theta$ on $\mathcal{F}$}\label{sec:barrier-functions}
Let $\Theta$ be the set of triples $\theta=(\theta_1,\theta_2,\theta_3)$ of positive real numbers with $\theta_1+\theta_2+\theta_3=1$. To prove \cref{thm:barrier}, we will construct, for every $\theta\in\Theta$, positive functions~$F^\star_\theta$ on~$\mathcal{F}$ satisfying $F^\star_\theta(\langle r\rangle)=r$ for every integer $r\geq1$ and which are $\preceq$-monotone.

The idea behind the construction is to make the inequalities that the degenerations \ruleref{D0}-\ruleref{D4} give (\cref{tab:profiles}) essentially tight.
At the matrix multiplication tensors, the functions~$F^\star_\theta$ will have the values $F^\star_\theta(\polmul_{1,n,n})=n^{p_1}$, $F^\star_\theta(\polmul_{n,1,n})=n^{p_2}$ and $F^\star_\theta(\polmul_{n,n,1})=n^{p_3}$, for some exponents $p_1$, $p_2$ and~$p_3$. These cannot be chosen freely: by the proof of \cref{thm:rect}, the exponents of such functions satisfy the inequality of \cref{prop:angular}. 
We want this inequality to be an equality. 
If we write $x_i=1-\theta_i$ with $\theta\in\Theta$, we obtain equality if we take 
\begin{equation}\label{eq:exponents}
p_i\coloneqq\frac{\theta_i}{(1-\theta_j)(1-\theta_k)},
\end{equation}
where $\{i,j,k\}=\{1,2,3\}$, since $x_1+x_2+x_3=2$ and $p_ix_jx_k=\theta_i$. 
Throughout the remainder of this section we will use $p_i$ as defined in \eqref{eq:exponents}, and leave its dependence on $\theta=(\theta_1,\theta_2,\theta_3)$ implicit. Recall that a triple $(d_1,d_2,d_3)$ of real numbers $d_1,d_2,d_3\geq1$ is admissible (\cref{def:admissible}) if its largest entry is the sum of the other two minus~$1$, that is, if $d_1+d_2+d_3=2\max(d_1,d_2,d_3)+1$.

\begin{definition}\label{def:phi-star}
Let $\theta\in\Theta$, and let $p_1$, $p_2$ and~$p_3$ be as in~\eqref{eq:exponents}. 
\begin{enumerate}[\upshape(i)]
\item For an admissible triple $(d_1,d_2,d_3)$ of real numbers, let $A_s\coloneqq(d_s-\theta_s)/(1-\theta_s)$ for $s\in[3]$. 
Let $i$, $j$ and~$k$ be indices with $\{i,j,k\}=\{1,2,3\}$ such that $d_i,d_j\leq d_k$ and $A_i\leq A_j$, and define
\[
F^\star_\theta(d_1,d_2,d_3)\coloneqq A_i^{1-p_i}\bigl(\theta_jA_i+(1-\theta_j)A_j\bigr)^{p_i}.
\]
The value does not depend on the choice of the indices: if two entries are largest, then the third is~$1$, so~$i$ is its index, and both choices of~$k$ give $\max(d_1,d_2,d_3)^{p_i}$; if $A_i=A_j$, both choices of~$i$ give~$A_i$.
On the generators of $\mathcal{F}$, let $F^\star_\theta(\polmul_{x,y,z})\coloneqq F^\star_\theta(x,y,z)$. 
\item We extend $F_\theta^\star$ to $\mathcal{F}$ as follows.
For all $S,T\in\mathcal{F}$ and every $\ell\in[3]$ such that the expressions on the left lie in~$\mathcal{F}$, let
\begin{itemize}
\item $F^\star_\theta(S\oplus T)\coloneqq F^\star_\theta(S)+F^\star_\theta(T)$,
\item $F^\star_\theta(S\otimes T)\coloneqq F^\star_\theta(S)\,F^\star_\theta(T)$, and
\item $F^\star_\theta(S\oplus_\ell T)\coloneqq\bigl(F^\star_\theta(S)^{1/p_\ell}+F^\star_\theta(T)^{1/p_\ell}\bigr)^{p_\ell}$.
\end{itemize}
\end{enumerate}
\end{definition}

By construction, the function~$F^\star_\theta$ is additive and multiplicative on $\mathcal{F}$. But we stress that it is defined on formal combinations in $\mathcal{F}$, not on tensors, and we will prove monotonicity of~$F^\star_\theta$ only for the formal preorder $\preceq$ not for degeneration $\degen$.

\subsection{Choice of $\theta$ for the $\omega(a,b,c)$ barrier}

To prove \cref{thm:barrier}, we need to choose the parameters $\theta \in \Theta$ optimally based on $a,b,c$.
For the square case, we can take $\theta_1=\theta_2=\theta_3=1/3$. Then we have $p_i=3/4$ and $A_s=(3d_s-1)/2$, so~$F^\star_\theta$ is the function
\begin{equation}\label{eq:phi-94}
F^\star(x,y,z)=\Bigl(\frac{3m-1}2\Bigr)^{1/4}\Bigl(\frac{x+y+z-m}2\Bigr)^{3/4},\qquad\text{where }m=\min(x,y,z),
\end{equation}
which satisfies $F^\star(\langle n,n,n\rangle)=n^{9/4}$.
The following lemma shows how to choose $\theta$ to cover the second case of \eqref{eq:barrier}.

\begin{lemma}\label{lem:sharp-parameters}
Let $a\geq b\geq c\geq0$ be real numbers with $c(a+b)>ab$, and let
\[
\theta_1\coloneqq1-\frac{2ab}{ab+bc+ca},\qquad \theta_2\coloneqq1-\frac{2bc}{ab+bc+ca},\qquad \theta_3\coloneqq1-\frac{2ac}{ab+bc+ca}.
\]
Then $\theta=(\theta_1,\theta_2,\theta_3)\in\Theta$ and
\[
\frac{a\theta_2}{(1-\theta_1)(1-\theta_3)}+\frac{b\theta_3}{(1-\theta_1)(1-\theta_2)}+\frac{c\theta_1}{(1-\theta_2)(1-\theta_3)}=\frac{(ab+bc+ca)^2}{4abc}.
\]
\end{lemma}
\begin{proof}
Let $D\coloneqq ab+bc+ca$. Since $c(a+b)>ab\geq0$ and $a\geq b\geq c\geq0$, we have that $D>2ab\geq2ac\geq2bc$. Hence $\theta_1,\theta_2,\theta_3>0$. Also, $\theta_1+\theta_2+\theta_3=3-2D/D=1$. Since $1-\theta_1=2ab/D$, $1-\theta_2=2bc/D$ and $1-\theta_3=2ac/D$, the three terms on the left-hand side of the claim are $D^2\theta_2/(4abc)$, $D^2\theta_3/(4abc)$ and $D^2\theta_1/(4abc)$, so their sum is $D^2/(4abc)$.
\end{proof}

We still need to prove that $F_\theta^\star$ is $\preceq$-monotone, but we can already prove the consequences.
If we have any preorder $\preceq'$ on $\mathcal{F}$ for which $F_\theta^\star$ is monotone, we obtain the barrier by evaluating $F_\theta^\star$ on $\langle n_1,n_2,n_3\rangle$ for the values of $\theta$ from \cref{lem:sharp-parameters}.

\begin{proposition}\label{cor:new-rules}
Let~$\preceq'$ be any preorder on~$\mathcal{F}$, and let $n_1,n_2,n_3,r\geq1$ be integers with $\langle n_1,n_2,n_3\rangle\preceq'\langle r\rangle$.
\begin{enumerate}[\upshape(i)]
\item If $F^\star_\theta(T)\leq F^\star_\theta(S)$ for all $T,S \in \mathcal{F}$ with $T\preceq'S$ and for all $\theta\in\Theta$, then 
\[
\log r\geq\begin{cases}a+b&\text{if }c(a+b)\leq ab,\\[2pt]\dfrac{(ab+bc+ca)^2}{4abc}&\text{if }c(a+b)>ab,\end{cases}
\]
where $a\geq b\geq c$ are the numbers $\log n_1$, $\log n_2$ and~$\log n_3$ in nonincreasing order. 
\item Suppose $n_1=n_2=n_3=n$. If $F_{(\frac13,\frac13,\frac13)}^\star(T)\leq F^\star_{(\frac13,\frac13,\frac13)}(S)$ for all $T,S \in \mathcal{F}$ with $T\preceq'S$, then $r \geq n^{9/4}$.
\end{enumerate}
\end{proposition}
\begin{proof}
We compute the values of~$F^\star_\theta$ at~$\langle r\rangle$ and~$\langle n_1,n_2,n_3\rangle$. 
At the triple $(1,n,n)$, %
\cref{def:phi-star} gives $F^\star_\theta(1,n,n)=\bigl(\theta_3+(1-\theta_3)A_3\bigr)^{p_1}=n^{p_1}$. In the same way, $F^\star_\theta(n,1,n)=n^{p_2}$ and $F^\star_\theta(n,n,1)=n^{p_3}$. In particular, $F^\star_\theta(\langle1\rangle)=F^\star_\theta(1,1,1)=1$, so $F^\star_\theta(\langle r\rangle)=r$, and since $F^\star_\theta$ is multiplicative, $F^\star_\theta(\langle n_1,n_2,n_3\rangle)=n_1^{p_2}n_2^{p_3}n_3^{p_1}$. So the hypothesis gives, for every~$\theta$ to which it applies (every $\theta\in\Theta$ in~(i), and $\theta=(\frac13,\frac13,\frac13)$ in~(ii)),
\begin{equation}\label{eq:monotone-values}
p_2\log n_1+p_3\log n_2+p_1\log n_3\leq\log r.
\end{equation}

(i) Renumbering the entries of~$\theta$ renumbers $p_1$, $p_2$ and~$p_3$ in the same way. Since~\eqref{eq:monotone-values} holds for all~$\theta$, it follows that $ap_2+bp_3+cp_1\leq\log r$ for all~$\theta$. 
If $c(a+b)>ab$, then for the~$\theta$ of \cref{lem:sharp-parameters}, $ap_2+bp_3+cp_1=(ab+bc+ca)^2/(4abc)$, so $(ab+bc+ca)^2/(4abc)\leq\log r$. 

If $c(a+b)\leq ab$, let $0<\theta_1<1$ and $\theta_2=\theta_3=(1-\theta_1)/2$. Then $p_2=p_3=1/(1+\theta_1)$, so $(a+b)/(1+\theta_1) \leq \log r$, since $cp_1\geq0$. Letting $\theta_1\to0$ gives $a+b \leq \log r$.

(ii) For $\theta_1=\theta_2=\theta_3=1/3$, we have $p_1=p_2=p_3=3/4$, so~\eqref{eq:monotone-values} gives $\tfrac94\log n\leq\log r$, that is, $n^{9/4}\leq r$.
\end{proof}

It remains to show monotonicity for the preorder $\preceq$ as defined in \cref{def:formal-relation}. We will first prove the basic properties of $F^\star_\theta$ that we will need for that.

\subsection{Basic properties of $F^\star_\theta$}\label{sec:properties-of-monotone-function}

The following proposition will give the properties of~$F^\star_\theta$ that we will use to prove monotonicity under $\preceq$. We state it for the tensors $\polmul_{x,y,x+y-1}$; since the values of~$F^\star_\theta$ do not change when the legs and the entries of~$\theta$ are renumbered together, it also holds with the legs renumbered. Parts~(\ref{op:basic})--(\ref{op:tripling}) correspond to the properties of spectral points in \cref{lem:profiles}, and part~(\ref{op:bounds}) will be used for~\fruleref{D2} and~\ruleref{R}.

\begin{proposition}\label{prop:rect-sharp}
Let $\theta\in\Theta$, and let $p\coloneqq p_1 = \theta_1(1-\theta_2)^{-1}(1-\theta_3)^{-1}$.
\begin{enumerate}[\upshape(i)]
\item\label{op:basic} For all positive integers $x$, $y$ and~$z$ with $z=x+y-1$, we have $\min(x,y)\leq F^\star_\theta(x,y,z)\leq z$, and $F^\star_\theta(x,y,x+y-1)$ is nondecreasing in~$x$ and in~$y$.
\item\label{op:concave} For every positive integer~$x$, the function $y\mapsto F^\star_\theta(x,y,x+y-1)^{1/p}$ on~$\ZZ_{\geq1}$ is concave.
\item\label{op:tripling} For all positive integers $x$, $y$, $z$ with $z=x+y-1$, we have
\[
2F^\star_\theta(x,y,z)^{1/p}+F^\star_\theta(x,z,y)^{1/p}\leq F^\star_\theta(x,2y+z,y+2z)^{1/p}.
\]
\item\label{op:bounds} For all positive integers $x$, $y$, $z$ with $z=x+y-1$, we have $x^{1-p}y^p\leq F^\star_\theta(x,y,z)\leq x^{p_2}y^p$.
\end{enumerate}
\end{proposition}
\begin{proof}
We begin by showing some useful identities on the $p_i$'s and $\theta_i$'s that we will use throughout the proof.
For $\{i,j,k\}=\{1,2,3\}$, we have $p_i=\theta_i/(\theta_i+\theta_j\theta_k)$, %
so $0<p_i<1$. Moreover, %
\begin{equation}\label{eq:sharp-first}
p_i(1-\theta_j)+p_j(1-\theta_i) = (\theta_i  + \theta_j) / (1-\theta_k) =  1.
\end{equation}
In particular, $p_i+p_j=1+p_i\theta_j+p_j\theta_i>1$, and 
\begin{equation}\label{eq:p_i-equality}
    (1-p_i)+p_i\theta_j=p_j(1-\theta_i).
\end{equation}

Now set $f(a,b)\coloneqq F^\star_\theta(a,b,a+b-1)$ for real numbers $a,b\geq1$, and let $A\coloneqq(a-\theta_1)/(1-\theta_1)$ and $B\coloneqq(b-\theta_2)/(1-\theta_2)$. Then~$A$ and~$B$ are the numbers~$A_1$ and~$A_2$ of \cref{def:phi-star}.
Since $a,b \leq a+b-1$ we take $k=3$ in \cref{def:phi-star}, which gives
\begin{equation}\label{eq:definition small f}
f(a,b)=\begin{cases}A^{1-p_1}\bigl(\theta_2A+(1-\theta_2)B\bigr)^{p_1}&\text{if }A\leq B,\\[2pt]B^{1-p_2}\bigl(\theta_1B+(1-\theta_1)A\bigr)^{p_2}&\text{if }A>B.\end{cases}
\end{equation}
We show the function $f$ is continuously differentiable. To start, both cases are equal whenever $A=B$, so $f$ is continuous. 
We compute the derivatives of the functions in the two cases of \eqref{eq:definition small f} at $A = B$, and show they are equal.
We claim the derivatives in $a$ of both cases are equal to $p_2$. 
For the second case this is directly clear. 
By the chain rule, the derivative in $a$ of the first case is equal to its derivative in $A$ divided by $1-\theta_1$.
This derivative in $A$ is equal to $(1-p_1)A^{-p_1}U^{p_1} + p_1\theta_2A^{1-p_1}U^{p_1-1}$, where $U = \theta_2A+(1-\theta_2)B$.
When $A = B$ this equals $(1-p_1) + p_1\theta_2$, which is $p_2(1-\theta_1)$ by \eqref{eq:p_i-equality}. This proves the claim. 
By a symmetric argument, the derivatives in $b$ of both cases equal $p_1$. We conclude $f$ is continuously differentiable.

We are now ready to prove the four items.
\begin{enumerate}
\item[(\ref{op:basic})] 
    Each case of \eqref{eq:definition small f} is a weighted geometric mean, with nonnegative weights of sum~$1$, and each factor is a convex combination of $A$ and $B$. 
    Thus, $f\geq\min(A,B)\geq\min(a,b)$, where the second inequality follows from $A \ge a$ and $B \ge b$. For the upper bound, the weighted AM-GM inequality and \eqref{eq:p_i-equality} give for the first case an upper bound of 
    \begin{align*}
        (1-p_1)A+p_1\bigl(\theta_2A+(1-\theta_2)B\bigr)
        &=p_2(1-\theta_1)A+p_1(1-\theta_2)B
        \\&=p_2 (a - \theta_1) + p_1 (b - \theta_2).
    \end{align*}
    By \eqref{eq:sharp-first}, this equals $1 + p_2(a-1) + p_1(b-1)$, which is upper bounded by $a+b-1$ because $p_1,p_2 < 1$. 
    The second case is upper bounded in an analogous manner. 
    
    Finally, 
    because $A$ and $B$ are affine functions that are nondecreasing in $a$ and $b$, each case of~\eqref{eq:definition small f} is nondecreasing in~$a$ and in~$b$. Then because $f$ is continuous, $f$ is nondecreasing in each argument. 

\item[(\ref{op:concave})]
Fix $a\geq1$, and let $g\coloneqq f^{1/p_1}$. 
Because~$B$ is an increasing affine function of~$b$, it suffices to show that~$g$ is concave as a function of the real number $B\geq1$.
The same is then true for its restriction to the integers~$b$.
It furthermore suffices to show that $g$ is piece-wise concave,
because $f$ and hence also $g$ is continuously differentiable.

For $A \le B$, we have that $g=A^{(1-p_1)/p_1}\bigl(\theta_2A+(1-\theta_2)B\bigr)$ is affine, and hence concave. 

Now let $A > B$. We have $g=B^uL^v$, where $L\coloneqq\theta_1B+(1-\theta_1)A$, $u\coloneqq(1-p_2)/p_1$ and $v\coloneqq p_2/p_1$. 
We will show that the second derivative of $g$ in $B$ is negative.
A direct computation of the first and second derivatives w.r.t.\ $B$ 
gives
\begin{align*}
\tfrac{\partial}{\partial B}  g%
&= u B^{u-1} L^v + \theta_1 v B^u L^{v-1}, \\
\tfrac{\partial^2}{\partial B^2} g%
&=
B^{u-2}L^{v-2}\bigl(u(u-1)L^2+2uv\theta_1BL+v(v-1)\theta_1^2B^2\bigr).
\end{align*}
Let $Q$ be the expression in the outer brackets. We claim that $Q < 0$, which is sufficient because $B^{u-2}L^{v-2} > 0$.
Rewrite
\[
Q=(uL+v\theta_1B)^2-\bigl(uL^2+v(\theta_1B)^2\bigr).
\]
Since $\partial L/\partial B=\theta_1$ and $u+v=1/p_1>1$, then 
    $\partial Q / \partial B=2\theta_1(u+v-1)(uL+v\theta_1B)>0$,
so~$Q$ is increasing in~$B$. So if we show that $Q$ is equal to $0$ at $A = B$, then we are done.
At $A=B$ we have $L=A$, so $Q=A^2\bigl((u+v\theta_1)^2-(u+v\theta_1^2)\bigr)$. 
By \eqref{eq:p_i-equality} we get $u+v\theta_1=1-\theta_2$, and from the definitions we get $v\theta_1(1-\theta_1)=\theta_2(1-\theta_2)$. 
We find
\begin{align*}
 u+v\theta_1^2=(u+v\theta_1)-v\theta_1(1-\theta_1)
 =
 (1-\theta_2)^2,
\end{align*}
which shows $(u+v\theta_1)^2-(u+v\theta_1^2) = 0$, and the proof is complete.

\item[(\ref{op:tripling})] 
    Write $(x,y,z)=(a,b,a+b-1)$.
    Fix $a\geq1$. Let $\mu \coloneqq A^{(1-p_1)/p_1}$, and, for real $b\geq1$, let $h_2(b) \coloneqq F^\star_\theta(a,b,a+b-1)^{1/p_1}$ and $h_3(b)\coloneqq F^\star_\theta(a,a+b-1,b)^{1/p_1}$.

    In (\ref{op:concave}) we proved that the function $h_2$ is concave, and the same proof applies for $h_3$ if we exchange the second and third legs. 
    For large $b$ (so that we are in the first case of \eqref{eq:definition small f} and similarly for $h_3$), 
    we have $h_\ell(b) = \mu(b+\theta_\ell(A-1))$. 
    By concavity, it follows that $h_\ell(b) \le \mu(b+\theta_\ell(A-1))$ for any $b \ge 1$. 
    Moreover,~$\mu$ lower bounds the slopes at any point, which gives $h_\ell(b)+\mu t\leq h_\ell(b+t)$ for $t \geq 0$. 
    By \eqref{eq:sharp-first} we also obtain $(\theta_2+\theta_3)(A-1)=a-\theta_1 - (\theta_2 + \theta_3) = a-1$. 
    By the above three facts, we find that
    \begin{align*}
        2h_2(b)+h_3(b)
        &\leq 
        h_2(b)+\mu\bigl(b+ \theta_2(A-1)\bigr) + \mu\bigl(b + \theta_3(A-1)\bigr)
        \\&= h_2(b)+\mu(2b+a-1)
        \leq h_2(3b+a-1),
    \end{align*}
    which is the desired result.

\item[(\ref{op:bounds})] 
    Fix $a\geq1$. Let $h\coloneqq h_2$ and~$\mu$ be as in the proof of~(\ref{op:tripling}). 
    
    We prove the first inequality.
    As before, concavity gives $h(b)\geq h(1)+\mu(b-1)$ for every real $b\geq1$. We find that $h(1)=f(a,1)^{1/p_1}=a^{p_2/p_1}$, because when $b = 1$ we have $B = 1$ and $\theta_1B + (1-\theta_1)A = a$.
    Since $A\geq a$ and $p_1+p_2>1$, both~$\mu$ and~$h(1)$ are at least $a^{(1-p_1)/p_1}$. 
    It follows that $h(b)\geq a^{(1-p_1)/p_1}b$, which is the first inequality.
    
    We prove the second inequality. By concavity, the slopes of $h$ do not exceed $h'(1)$, where we take the derivative from the right. Hence, $h(b)\leq h(1)+(b-1)h'(1)$.
    If $a=1$, then $A=1\leq B$ for every $b\geq1$, so the first case of~\eqref{eq:definition small f} gives $h(b)=b$, so $h'(1)=1=h(1)$. Let $a>1$. Then $B=1<A$ at $b=1$, so for~$b$ near~$1$ the second case applies. We have $h=B^uL^v$, with~$L$, $u$ and~$v$ as in the proof of~(\ref{op:concave}). At $b=1$ we have $B=1$ and $L=a$, and 
    $B'=1/(1-\theta_2)$ and $L'=\theta_1/(1-\theta_2)$, so
    \[
    \frac{h'(1)}{h(1)}= \Bigl(\frac{uB'}{B}+\frac{vL'}{L}\Bigr)(1)=\frac{1}{1-\theta_2}\Bigl(u+\frac{v\theta_1}{a}\Bigr)=\frac{1-p_2+p_2\theta_1/a}{p_1(1-\theta_2)}\leq\frac{1-p_2+p_2\theta_1}{p_1(1-\theta_2)}=1,
    \]
    where we use $a>1$, and the last equality is~\eqref{eq:p_i-equality}.
    It follows that $h(b)\leq b\,h(1)$, that is, $f(a,b)\leq a^{p_2}b^{p_1}$. \qedhere
\end{enumerate}
\end{proof}

\subsection{Proof of monotonicity of $F_\theta^\star$ and the operational barrier}\label{sec:barrier-proof}

We prove $F_\theta^\star$ is monotone under $\preceq$.

\begin{proposition}[Combinations of the ingredients]\label{prop:combinations}
Let $\theta\in\Theta$. For all $S,T\in\mathcal{F}$, if $T\preceq S$, then $F^\star_\theta(T)\leq F^\star_\theta(S)$. \end{proposition}
\begin{proof}
Let $\preceq_{F^\star_\theta}$ be the relation on~$\mathcal{F}$ defined as follows: ~$T \preceq_{F^\star_\theta} S$ whenever $F^\star_\theta(T)\leq F^\star_\theta(S)$. Since~$\preceq$ is the smallest preorder on~$\mathcal{F}$ that is compatible with~$\oplus$ and~$\otimes$ and has properties \fruleref{D1}--\fruleref{D4}, \ruleref{R} and~\ruleref{I}, also with the legs renumbered, it suffices to show that $\preceq_{F^\star_\theta}$ is also such a preorder. Indeed, then it contains~$\preceq$, which is the claim. 

By definition of $F^\star_{\theta}$, the relation $\preceq_{F^\star_\theta}$ is compatible with direct sums and tensor products.
Renumbering the legs of an expression changes~$F^\star_\theta$ into~$F^\star_{\theta'}$, where~$\theta'$ is~$\theta$ with its entries renumbered in the same way. 
So, it suffices to check the properties as stated in \cref{def:formal-relation}, without the permutations.
We do this now. We will write $p \coloneqq p_1$.
\begin{enumerate}
\item[\fruleref{D1}] This is the upper bound in \cref{prop:rect-sharp}(\ref{op:basic}).
\item[\fruleref{D2}] Let $w_i=F^\star_\theta(T_i)$ for every $i\in[m]$. Then by H\"older's inequality and \cref{prop:rect-sharp}(\ref{op:bounds}),
\[
t^p\sum_{i=1}^mw_i\leq m^{1-p}t^p\Bigl(\sum_{i=1}^mw_i^{1/p}\Bigr)^p\leq F^\star_\theta(m,t,m+t-1)\Bigl(\sum_{i=1}^mw_i^{1/p}\Bigr)^p.
\]
The first and last quantities are the values of~$F^\star_\theta$ at the two sides of~\fruleref{D2}.
\item[\fruleref{D3}] 
Raising the values at the two sides to the power~$1/p$ gives the inequality of \cref{prop:rect-sharp}(\ref{op:concave}) if $z=x+y-1$, and this inequality with the second and third legs renumbered if $y=x+z-1$; a further summand~$C$ adds $F^\star_\theta(C)^{1/p}$ to both sides.
\item[\fruleref{D4}] 
The same as the proof for \fruleref{D3}, but applying \cref{prop:rect-sharp}(\ref{op:tripling}) instead.
\item[\ruleref{R}] The value at~$T_1$ is at most the values at $T_1\oplus T_2$ and at $T_1\oplusX T_2$, and the value at $T_1\oplusX T_2$ is at most that at $T_1\oplus T_2$, since $(x^{1/p}+y^{1/p})^p\leq x+y$ for all $x,y>0$ (using that $p \in [0,1]$). Moreover, $F^\star_\theta(x,y,x+y-1)$ is nondecreasing in~$y$, also with the second and third legs renumbered, and $F^\star_\theta(a,b,a+b-1)\leq a^{p_2}b^p$, which is the value at $\polmul_{a,1,a}\otimes\polmul_{1,b,b}$ (\cref{prop:rect-sharp}(\ref{op:basic}) and~(\ref{op:bounds})).
\item[\ruleref{I}] The isomorphisms do not change~$F^\star_\theta$. Indeed, $F^\star_\theta$ is additive and multiplicative, and $F^\star_\theta(\langle1\rangle)=F^\star_\theta(1,1,1)=1$. The formula for partial direct sums is symmetric and associative, and homogeneous of degree one, so it respects the distributivity of~$\otimes$ over~$\oplusX$. The value at $\langle n_1,n_2,n_3\rangle$ is multiplicative in $(n_1,n_2,n_3)$. Finally, the value at the partial direct sum of $r$~copies of~$U$ is $r^pF^\star_\theta(U)$, which is the value at $U\otimes\polmul_{1,r,r}$.\qedhere
\end{enumerate}
\end{proof}

\begin{proof}[Proof of \cref{thm:barrier}]
\Cref{prop:combinations} tells us that $F_\theta^\star$ is monotone under $\preceq$ for every $\theta\in\Theta$. Then \cref{cor:new-rules}, with~$\preceq$ as~$\preceq'$, gives us the result.
\end{proof}

\begin{example}
The barrier of \cref{cor:new-rules} applies also to preorders $\preceq'$ extending $\preceq$, as long as $F_\theta^\star$ remains monotone. 
As a basic example, if we add Strassen's degenerations $\langle\lceil3n^2/4\rceil\rangle\degen\langle n,n,n\rangle$~\cite[Theorem~6.6]{MR882307} as one of the generating relations of $\preceq$, the barrier still applies.
Indeed, 
$\lceil3n^2/4\rceil\leq n^2\leq F^\star_\theta(\langle n,n,n\rangle)$,
because $F^\star_\theta(\langle n,n,n\rangle) = n^{p_1+p_2+p_3}$ and
\[
p_1+p_2+p_3-2 = \frac{2\theta_1\theta_2\theta_3}{(1-\theta_1)(1-\theta_2)(1-\theta_3)} > 0.
\]
Similarly, each $F^\star_\theta$ is also monotone on the polynomial multiplication degenerations of \cref{rem:det-general} and \cref{rem:tripling-general}, so the barrier applies.
As another example, which is also clear form the asymptotic spectrum point of view, the barrier remains true if $\preceq$ is replaced by its asymptotic closure. There an expression $T\in\mathcal{F}$ is at most $S\in\mathcal{F}$ if  $T^{\otimes n}\preceq S^{\otimes n}\otimes\langle2^{o(n)}\rangle$ for every integer $n\geq1$. Indeed, then $F^\star_{\theta}(T)^n\leq F^\star_{\theta}(S)^n\,2^{o(n)}$ for every integer $n\geq1$, since~$F^\star_{\theta}$ is multiplicative and $F^\star_{\theta}(\langle2^{o(n)}\rangle)=2^{o(n)}$, so $F^\star_{\theta}(T)\leq F^\star_{\theta}(S)$.
\end{example}

\subsection{Operational asymptotic spectrum $\Delta_{\preceq}$}
Let~$\Delta_{\preceq}$ be the set of functions $F\colon\mathcal{F}\to\RR_{>0}$ with $F(\langle1\rangle)=1$, $F(S\oplus T)=F(S)+F(T)$ and $F(S\otimes T)=F(S)\,F(T)$ for all $S,T\in\mathcal{F}$, and with $F(T)\leq F(S)$ whenever $T\preceq S$. This is the asymptotic spectrum of the semiring of $\approx$-classes of~$\mathcal{F}$, ordered by~$\preceq$, where $T\approx S$ if $T\preceq S$ and $S\preceq T$.

The following theorem phrases the barrier results in the language of this asymptotic spectrum $\Delta_{\preceq}$. %
\begin{theorem}[Maximum over~$\Delta_{\preceq}$]\label{thm:formal-max}
For all real numbers $a\geq b\geq c\geq0$,
\[
\max_{F\in\Delta_{\preceq}}F(\langle2,1,1\rangle)^a\,F(\langle1,2,1\rangle)^b\,F(\langle1,1,2\rangle)^c=\begin{cases}2^{a+b}&\text{if }c(a+b)\leq ab,\\[2pt]2^{(ab+bc+ca)^2/(4abc)}&\text{if }c(a+b)>ab.\end{cases}
\]
In particular, $\max_{F\in\Delta_{\preceq}}F(\langle2,2,2\rangle)=2^{9/4}$.
\end{theorem}
\begin{proof}
We sketch the proof.
The central idea is that in the proof of \cref{thm:rect}, the properties of spectral points that are used are exactly the ones in the definition of $\Delta_\preceq$. This establishes the upper bound (we provide a bit more detail below).

These upper bounds are subsequently attained by~$F^\star_\theta$ with the parameters of \cref{lem:sharp-parameters} (in the second case) and by the flattening rank $\flatrank{1}$ (in the first case).
We obtain equality because $F^\star_\theta \in \Delta_{\preceq}$ (by its definition and \cref{prop:combinations}) and 
$\flatrank{1} \in \Delta_{\preceq}$ (because every relation in~$\preceq$ is a degeneration).

We now provide some more details on the upper bound.
Let $F\in\Delta_{\preceq}$. 
The proof of \cref{prop:angular} uses several properties of spectral points concerning their evaluations at the tensors~$\polmul_{x,y,z}$, and~$F$ also satisfy these properties. 
First, $F(\polmul_{1,m,m})=m^p$ with $p=\log F(\polmul_{1,2,2})\in[0,1]$ for all $m \geq 1$, which is the case because $m\mapsto F(\polmul_{1,m,m})$ is multiplicative by~\ruleref{I3}, nondecreasing by~\ruleref{R2} and at most~$m$ by~\fruleref{D1}. 
Second, $F(\polmul_{x,y,z})\geq F(\langle1\rangle)=1$, which follows from~\ruleref{R2}. 

Finally, the properties in \cref{lem:profiles}, which follow from \fruleref{D1}--\fruleref{D4}, \ruleref{R} and~\ruleref{I} by the proofs of \cref{lem:iso-summands}, \cref{lem:superadd-general} and \cref{lem:profiles}. 
Zooming in on the proof of \cref{lem:iso-summands}, we see that it uses that the summands are isomorphic only to compute the value of their direct sum; for~$F$, additivity gives $F(U_1\oplus\cdots\oplus U_r)=r\,F(U_1)$ whenever $F(U_1)=\cdots=F(U_r)$, so that proof applies to summands with equal values of~$F$. In the proof of \cref{lem:superadd-general}, the summands of one type have equal values of~$F$, since~$F$ is multiplicative. 

These are all the properties that are used in proof of \cref{prop:angular} and \cref{lem:rect-bound}, and hence also in the proof of \cref{thm:rect}. We conclude the upper bound applies.
\end{proof}

Let $\Delta_{\preceq}^{\!\mathsf{m}}\coloneqq\bigl\{\bigl(F(\langle1,1,2\rangle),F(\langle2,1,1\rangle),F(\langle1,2,1\rangle)\bigr):F\in\Delta_{\preceq}\bigr\}.$
Then $\log\Deltam\subseteq\log\Delta_{\preceq}^{\!\mathsf{m}}\subseteq Q$, for the set~$Q$ of~\eqref{eq:Q}, since every relation in~$\preceq$ is a degeneration and since the proof of \cref{thm:formal-max} shows the inequality of \cref{prop:angular} for every $F\in\Delta_{\preceq}$. The functions~$F^\star_\theta$ give the points $(p_1,p_2,p_3)$, since $F^\star_\theta(\langle1,1,2\rangle)=2^{p_1}$, $F^\star_\theta(\langle2,1,1\rangle)=2^{p_2}$ and $F^\star_\theta(\langle1,2,1\rangle)=2^{p_3}$. At these points, the inequality defining~$Q$ is an equality for $x=(1-\theta_1,1-\theta_2,1-\theta_3)$, so they lie on the curved part of the boundary of~$Q$ (\cref{fig:spectrum}, right). So \cref{thm:formal-max} says that $aw_2+bw_3+cw_1$ has the same maximum on $\log\Delta_{\preceq}^{\!\mathsf{m}}$ as on~$Q$.

\begin{remark}
In fact, $\log\Delta_{\preceq}^{\!\mathsf{m}}=Q\cap\{w\in\RR^3:w_1+w_2\geq1,\ w_1+w_3\geq1,\ w_2+w_3\geq1\}$. In particular, $\log\Delta_{\preceq}^{\!\mathsf{m}}$ is convex and contains the part of~$Q$ with $w_1+w_2+w_3\geq2$ shown in \cref{fig:spectrum}. We omit the proof.
\end{remark}

\paragraph{Concurrent work.} OpenAI released a proof that $\omega_\CC\leq9/4$, with a formal proof in Lean, on 6~October 2026~\cite{OAI:Matrix-Multiplication-Nine-Fourths-October-2-2026}. Since then, Sela Navot, Kunal Marwaha and we each released a Lean proof that $\omega_\FF\leq9/4$ for every field~$\FF$ (7~October; \href{https://github.com/selanavot/matrix-multiplication-all-fields}{\nolinkurl{selanavot/matrix-multiplication-all-fields}} on GitHub, \href{https://prove2.me/theorems/f2b5622a-3bca-4b8f-92a1-381d10af16e3}{\nolinkurl{mme_omega_le_225}} on Prove2Me, and \href{https://github.com/jzuiddam/omega-nine-quarters-all-fields}{\nolinkurl{jzuiddam/omega-nine-quarters-all-fields}} on GitHub). Lü released a proof that $\omega_\FF\leq9/4$ for every field~$\FF$~\cite{lü2026polynomialseparationsharedtensor} (12~October). Uzna\'nski released proofs of the bounds of \cref{cor:rect-k}, and then of those of \cref{thm:rect}, in characteristic zero~\cite{uznanski2026rectangularmatrixmultiplicationpolynomial} (7 and 11~October), and, assuming Strassen's spectral theorem, a Lean proof of them for every field (12~October). Zhou and Wang released a Lean proof that the bounds of \cref{cor:rect-k} hold over every associative ring (14~October; on GitHub, \href{https://github.com/Hugin-Hilbert/rectangular-matmul-rings}{\nolinkurl{Hugin-Hilbert/rectangular-matmul-rings}}). The authors of two GitHub repositories, \href{https://github.com/eliottcassidy2000/math}{\nolinkurl{eliottcassidy2000/math}} and \href{https://github.com/acprk/openai-math-followups}{\nolinkurl{acprk/openai-math-followups}}, released proofs that, for square matrices, the kind of analysis we do in Section 4 is optimal, using a function similar to~$F^\star$ of~\eqref{eq:phi-94} %
(7 and 10~October); 
our operational barrier \cref{thm:barrier} is stronger as it applies to any combination of the degeneration ingredients, and applies to all rectangular shapes. 

\paragraph{AI disclosure.} We have used LLMs to understand the OpenAI proofs and find proofs for their generalisations to all fields and all rectangular shapes. Based on these outlines, the writing and exposition choices are ours, in particular the results of the operational barrier.

\paragraph{Acknowledgements.} JZ was supported by NWO Vidi grant VI.Vidi.243.195. JZ, AA and MvdB were supported by ERC Starting Grant 101220349 SPECTRA. 
MvdB was also supported by  
ERC Starting Grant 101040907 SYMOPTIC, Dutch National Growth Fund (NGF), as part of the Quantum Delta NL visitor programme,  the Deutsche Forschungsgemeinschaft (DFG, German Research Foundation, 556164098), and the Quantum Software Consortium.
This work used the Dutch national e-infrastructure with the support of the SURF Cooperative using grant no.~EINF-16525. This work used computing resources of the Open Computing Facility at the University of California, Berkeley.

\appendix
\crefalias{section}{appendix}

\newcommand{\etalchar}[1]{$^{#1}$}

\end{document}